%% file: main.tex
\pdfoutput=1
\documentclass[11pt]{article}

\usepackage[letterpaper,margin=1in]{geometry}

\usepackage{lipsum}

\newcommand\blfootnotea[1]{%
  \begingroup
  \renewcommand\thefootnote{}\footnote{#1}%
  \endgroup
}

\usepackage[utf8]{inputenc} %
\usepackage[T1]{fontenc}    %

\usepackage{url}            %
\usepackage{booktabs}       %
\usepackage{nicefrac}       %
\usepackage{microtype}      %
\usepackage{enumitem}
\usepackage{dsfont}
\usepackage{multicol}
\usepackage{xcolor}

\usepackage{amsthm,amsfonts,amsmath,amssymb,epsfig,color,float,graphicx,verbatim, enumitem}

\usepackage{algpseudocode,algorithm,algorithmicx}  

\usepackage{mathtools}
\usepackage{bbm}

\usepackage{thm-restate}
\usepackage{turnstile}

\usepackage[colorlinks,citecolor=blue,linkcolor=magenta,bookmarks=true,hypertexnames=false]{hyperref}
\usepackage[nameinlink]{cleveref}

\crefname{equation}{equation}{equations}
\crefname{lemma}{lemma}{lemmata}
\crefname{claim}{claim}{claims}
\crefname{theorem}{theorem}{theorems}
\crefname{proposition}{proposition}{propositions}
\crefname{corollary}{corollary}{corollaries}
\crefname{claim}{claim}{claims}
\crefname{remark}{remark}{remarks}
\crefname{definition}{definition}{definitions}
\crefname{fact}{fact}{facts}
\crefname{question}{question}{questions}
\crefname{condition}{condition}{conditions}
\crefname{algorithm}{algorithm}{algorithms}
\crefname{assumption}{assumption}{assumptions}
\crefname{problem}{problem}{problems}

\newtheorem{theorem}{Theorem}[section]
\newtheorem{lemma}[theorem]{Lemma}

\newtheorem{corollary}[theorem]{Corollary}
\newtheorem{claim}[theorem]{Claim}

\newtheorem{definition}[theorem]{Definition}

\newtheorem{fact}[theorem]{Fact}

\theoremstyle{definition}

\newtheorem{remark}[theorem]{Remark}
\newtheorem{problem}[theorem]{Problem}

\newcommand{\cons}{\text{\cons}}

\newcommand{\eps}{\epsilon}

\newcommand{\poly}{\mathrm{poly}}
\newcommand{\polylog}{\mathrm{polylog}}

\newcommand{\Var}{\mathbf{Var}}

\def\D{\mathcal D}
\newcommand{\bbN}{\mathbb N}

\def\R{\mathbb R}

\def\N{\mathbb N}

\def\Z{\mathbb Z}

\newcommand{\cAksparse}{\cA_{\textnormal{$k$-sparse}}}

\newcommand{\cA}{\mathcal{A}}
\newcommand{\cB}{\mathcal{B}}
\newcommand{\cC}{\mathcal{C}}
\newcommand{\cD}{\mathcal{D}}
\newcommand{\cE}{\mathcal{E}}

\newcommand{\cN}{\mathcal{N}}

\newcommand{\cP}{\mathcal{P}}

\newcommand{\cS}{\mathcal{S}}

\newcommand{\cU}{\mathcal{U}}

\newcommand{\cX}{\mathcal{X}}

\newcommand*{\pE}{\tilde{\E}}

\newcommand{\e}{\epsilon}
\newcommand{\paren}[1]{(#1)}
\newcommand{\Paren}[1]{\left(#1\right)}

\newcommand{\brac}[1]{[#1]}
\newcommand{\Brac}[1]{\left[#1\right]}

\newcommand{\abs}[1]{\lvert#1\rvert}
\newcommand{\Abs}[1]{\left\lvert#1\right\rvert}

\newcommand{\Bigabs}[1]{\Big\lvert#1\Big\rvert}

\newcommand{\Set}[1]{\left\{#1\right\}}

\newcommand{\norm}[1]{\lVert#1\rVert}
\newcommand{\Norm}[1]{\left\lVert#1\right\rVert}

\newcommand{\ovl}{\overline}

\newcommand{\iprod}[1]{\langle#1\rangle}
\newcommand{\Iprod}[1]{\left\langle#1\right\rangle}

\newcommand{\hide}[1]{}
\DeclareMathOperator{\supp}{supp}

\DeclareMathOperator*{\pr}{\mathrm{Pr}}
\DeclareMathOperator*{\E}{\mathbf{E}}
\newcommand{\eqdef}{\stackrel{{\mathrm {\footnotesize def}}}{=}}

\def\d{\mathrm{d}}

\newcommand{\tr}{\mathrm{tr}}

\DeclarePairedDelimiter\floor{\lfloor}{\rfloor}

\newcommand{\bbR}{\mathbb R}

\def\colorful{1}

\ifnum\colorful=1

\else

\fi

\allowdisplaybreaks

\title{Robust Sparse Mean Estimation via Sum of Squares\blfootnotea{Authors are in alphabetical order.  Part of this work was done while a subset of the authors were visiting the Simons Institute for the Theory of Computing.}}

\author{
Ilias Diakonikolas\thanks{Supported by NSF Medium Award CCF-2107079,
NSF Award CCF-1652862 (CAREER), a Sloan Research Fellowship, and
a DARPA Learning with Less Labels (LwLL) grant.}\\
University of Wisconsin-Madison\\
{\tt ilias@cs.wisc.edu}\\
\and
Daniel M. Kane\thanks{Supported by NSF Medium Award CCF-2107547,
NSF Award CCF-1553288 (CAREER), a Sloan Research Fellowship, and a grant from CasperLabs.}\\
University of California, San Diego\\
{\tt dakane@cs.ucsd.edu}
\and
Sushrut Karmalkar\thanks{Supported by NSF under Grant \#2127309 to the Computing Research Association for the CIFellows 2021 Project.}\\
University of Wisconsin-Madison\\
{\tt skarmalkar@wisc.edu}\\
\and
Ankit Pensia\thanks{Supported by NSF grants NSF Award CCF-1652862 (CAREER), DMS-1749857, 
and CCF-1841190.}\\
University of Wisconsin-Madison\\
{\tt ankitp@cs.wisc.edu}\\
\and
Thanasis Pittas\thanks{Supported by NSF Award CCF-1652862 (CAREER).}\\
University of Wisconsin-Madison\\
{\tt pittas@wisc.edu}\\
}
 
\begin{document}

\maketitle

\begin{abstract}%
We study the problem of high-dimensional sparse mean estimation in the presence of an $\eps$-fraction of adversarial outliers.
Prior work obtained sample and computationally efficient algorithms 
for this task for identity-covariance subgaussian distributions.
In this work, we develop the first efficient algorithms for robust sparse mean estimation 
without a priori knowledge of the covariance. 
For distributions on $\R^d$ with ``certifiably bounded'' 
$t$-th moments and sufficiently light tails, 
our algorithm achieves error of $O(\eps^{1-1/t})$ with sample complexity $m = (k\log(d))^{O(t)}/\eps^{2-2/t}$.
For the special case of the Gaussian distribution, our algorithm achieves near-optimal error of  
$\tilde O(\eps)$ with sample complexity $m = O(k^4 \polylog(d))/\eps^2$. Our algorithms 
follow the Sum-of-Squares based, proofs to algorithms approach. We complement our upper bounds 
with Statistical Query and low-degree polynomial testing lower bounds, 
providing evidence that the sample-time-error tradeoffs achieved by our algorithms
are qualitatively the best possible.
\end{abstract}

\setcounter{page}{0}

\thispagestyle{empty}

\newpage

\input{intro.tex}

\input{our-results.tex}

\input{prelim.tex}

\input{certifiable-k-sparse-concentration.tex}

\input{poincare-mean-est.tex}

\input{gaussian-mean-est.tex}

\input{SQ-lowerbound.tex}

\newpage

\bibliographystyle{alpha}
\bibliography{allrefs}

\newpage
\appendix
\input{app-preliminaries}

\input{app-certifiably-sparse-moments}

\input{app-bdd-moments}

\input{app-gaussian-o-eps-error}

\input{SQ-appendix}
\input{app-sample-complexity}

\end{document}

%% file: intro.tex
\section{Introduction} %
\label{sec:introduction}

High-dimensional robust statistics~\cite{HampelEtalBook86, Huber09}
aims to design estimators that are tolerant to a {\em constant fraction} 
of outliers, independent of the dimension. Early work in this field, see, e.g.,~\cite{Tukey60, Huber64, Tukey75}, 
developed sample-efficient robust estimators for various basic tasks, alas with runtime exponential in the dimension. 
During the past five years, a line of work in computer science, starting with~\cite{DKKLMS16, LaiRV16}, 
has developed the first {\em computationally efficient} robust high-dimensional estimators 
for a range of tasks. This progress has led to a revival of robust statistics from an algorithmic perspective,
see, e.g.,~\cite{DK19-survey, DK+:cacm21} for recent surveys.

Throughout this work, we focus on the following standard contamination model.

\begin{definition}[Strong Contamination Model]
Fix a parameter $0< \eps<1/2$. We say that a set of $m$ points 
is an $\e$-corrupted set of samples from a distribution $D$ %
if it is generated as follows:
First, a set $S$ of $m$ points is sampled i.i.d.\ from $D$. 
Then an adversary observes $S$, replaces any $\e m$ of points in $S$ 
with any vectors they like to obtain the set $T$. 
We say that $T$ is an $\eps$-corrupted version of $S$.
\end{definition}

Here we study high-dimensional robust statistics tasks in the presence
of {\em sparsity constraints}. Leveraging sparsity in high-dimensional datasets
is a fundamental and practically important problem 
(see, e.g.,~\cite{Hastie15} for a textbook on the topic).
We focus on arguably the most fundamental such problem, 
that of {\em robust sparse mean estimation}. Specifically,
we are given an $\eps$-corrupted set of samples 
from a structured distribution $D$, whose unknown 
mean $\mu = \E_{X \sim D}[X] \in \R^d$ is $k$-sparse (i.e., supported on at most $k$ coordinates), 
and we want to compute  a good approximation $\widehat{\mu}$ of $\mu$. 
Importantly, in the sparse setting, we have access to much fewer samples 
compared to the dense case --- namely $\poly(k, \log d)$ 
instead of $\poly(d)$. Consequently, the design and analysis of 
algorithms for robust sparse estimation requires additional ideas, as compared to 
the standard (dense) setting~\cite{DKKLMS16}.

Prior work on robust sparse mean 
estimation~\cite{BDLS17,li18thesis,DKKPS19-sparse,CDKGS21-sparse}
focused on the case that the covariance matrix of the inliers 
is {\em known or equal to the identity}.
For identity covariance distributions with sufficiently good concentration 
(specifically, subgaussian concentration),
the aforementioned works give efficient algorithms for robust $k$-sparse mean estimation 
that use $\poly(k, \log(d), 1/\eps)$ samples and achieve near-optimal $\ell_2$-error of $\tilde{O}(\eps)$. 
On the other hand, if the covariance matrix of the inlier distribution 
is {\em unknown} and spectrally bounded by the identity, 
the techniques in these works can at best achieve error of $O(\sqrt{\eps})$, even for the special case
of the Gaussian distribution. One can of course use a robust covariance estimation algorithm 
to reduce the problem to the setting of known covariance.
The issue is that the covariance matrix is not necessarily sparse, 
and therefore naive attempts of robustly estimating the covariance 
(e.g., with respect to Frobenius or Mahalanobis distance) 
would require $\poly(d)$ samples. 

Motivated by these drawbacks of prior work, in this paper we aim to design computationally 
efficient algorithms for robust sparse mean estimation, 
using $\poly(k, \log(d), 1/\eps)$ samples, 
{\em that achieve near-optimal error guarantees without a priori knowledge of the covariance matrix}. 
Our main contribution is a comprehensive picture of the tradeoffs
between sample complexity, running time, and error guarantee for a range of inlier distributions.
In more detail, for distributions with appropriate tail bounds and ``certifiably bounded'' $t$-th moments 
{in sparse directions} (see \Cref{def:bounded-moments-k-sparse}), 
we give an efficient algorithm that achieves error $O(\eps^{1-1/t})$. 
For the special case of the Gaussian distribution, we give an algorithm with near-optimal error of $\tilde{O}(\eps)$.
For both settings, we establish Statistical Query (SQ) lower bounds (and low-degree polynomial testing lower bounds)
which give evidence that the error-sample-time tradeoffs achieved 
by our algorithms are qualitatively the best possible.

%% file: our-results.tex
\subsection{Our Results}
We start by recalling prior results for the dense robust mean estimation of bounded moment distributions.
We say that a distribution $D$ on $\R^d$ with mean $\mu$ has $t$-th central moments bounded by $M$ 
if for all unit vectors $v$ it holds $\E_{X \sim D}\Brac{\iprod{v, X - \mu}^t} \leq M$.  
Although it is information-theoretically possible to robustly estimate the mean of such distributions, in the $\ell_2$-norm,
up to error $O(M^{1/t}\eps^{1-1/t})$ using $O(d/\eps^{2-2/t})$ samples 
(see \Cref{app:info-th-sample-comp} for the proof), 
all known efficient algorithms require the following stronger assumption. 

\begin{definition}[Certifiably $(M,t)$-Bounded Central Moments] \label{def:Sosbddmoments}
We say that a distribution $D$ on $\R^d$ with mean $\mu$ has $t$-th central moments \emph{certifiably} bounded by $M$ 
if   $M \, \|v\|_{2}^t- \E_{X \sim D}\Brac{\iprod{v, X - \mu}^t}$ 
can be written as a sum of square polynomials in $v= (v_1, \ldots, v_d)$ of degree $O(t)$. 
\end{definition}

Prior works~\cite{kothari2017outlier, HopkinsL18} gave efficient algorithms for dense
robust mean estimation of distributions with certifiably bounded central moments.
Their algorithms incur sample complexities at least $m=\poly(d^t)/\eps^2$, have running times $\poly((md)^t)$, 
and guarantee $\ell_2$-error of $O(M^{1/t}\eps^{1-1/t})$.  \footnote{ For simplicity of the exposition, 
we will not account for bit complexity in this section. In essence, we assume that the bit complexity 
of all relevant parameters is bounded by $\poly(md)$. 
}

We now turn our attention to the sparse setting, which is the focus of the current work.
In prior work, the term ``sparse mean estimation'' refers to the task of 
computing a $\widehat{\mu}$ such that $\widehat{\mu} - \mu$ 
is small in $\ell_2$-norm, \emph{assuming that $\mu$ is $k$-sparse}. 
We note that estimating a sparse vector in the $\ell_2$-norm is a special case 
of estimating an arbitrary vector in the $(2,k)$-norm, defined below (see \Cref{fact:sparseTruncation}). 

\begin{definition}[(2, $k$)-norm]
We define the $(2, k)$-norm of a vector $x$ to be the maximum correlation with any $k$-sparse unit vector, i.e., 
$\Norm{x}_{2,k} \eqdef \max_{\Norm{v}_2=1, v: k-\mathrm{sparse}} \iprod{v, x}$. 
\end{definition}
We henceforth focus on this more general formulation; 
we will use the term ``sparse mean estimation'' to mean 
that the error guarantees are defined with respect to the $(2, k)$-norm.

For distributions with $(M,t)$-bounded central moments, the information-theoretically optimal 
error for robust sparse mean estimation is $O(M^{1/t}\eps^{1-1/t})$ and 
can be obtained with $(k\log (d/k))/\eps^{2-2/t}$ samples (see \Cref{app:info-th-sample-comp} for the simple proof).

Our first result is a computationally efficient robust sparse mean estimation algorithm that applies 
to any distribution $D$ with certifiably bounded $t$-th moments in $k$-sparse directions 
(\Cref{def:bounded-moments-k-sparse}) and light tails. 
In particular, we assume that $D$ has subexponential tails, i.e., for some universal constant $c$, for all unit vectors $v$ 
and all $p\in \mathbb{N}$ it holds $\E_{ X  \sim D}\Brac{|\langle v, X - \E_{X \sim D}[X] \rangle|^p}^{1/p} \leq cp$. 
(In fact, our algorithmic result  holds as long the distribution $D$ has bounded
$\poly(t\log d)$ moments along coordinate axes; see \Cref{sec:sampling} and \Cref{sec:concentration_appendix}.)
Our algorithm achieves error $O(M^{1/t}\eps^{1-1/t})$ with $m= \poly((k \log d)^t)/ \eps^{2-2/t}$ samples 
and $\poly((md)^t)$ running time.

\begin{theorem}[Robust Sparse Mean Estimation for Certifiably Bounded Moments]\label{thm:main-poincare-informal}
Let $t$ be a power of two, $D$ be a distribution on $\R^d$ with unknown mean $\mu$, 
and $\eps<\eps_0$ for a sufficiently small constant $\eps_0>0$.
Suppose that $D$ has $t$-th moments certifiably bounded in $k$-sparse 
directions by $M$ (cf.\ \Cref{def:bounded-moments-k-sparse}) and subexponential tails.
There is  an algorithm which, given $\eps$, $M$, $t$, $k$, and an $\eps$-corrupted set of 
$m = (t k \log d)^{O(t)}~\max(1,M^{-2})/\eps^{2-2/t}$ samples from $D$, runs in time $\poly((md)^{t^2})$,
and returns a vector $\widehat{\mu}$ satisfying 
$\Norm{\widehat{\mu} - \mu}_{2,k} \leq O( M^{1/t} \eps^{1-1/t})$ with high probability.
\end{theorem}

It is natural to ask which distributions have such ``certifiably bounded moments in $k$-sparse directions''. 
In the dense case, \cite{KStein17} showed that \Cref{def:Sosbddmoments} is satisfied by $\sigma$-Poincar\'e distributions. 
(A distribution $D$ is $\sigma$-Poincar\'e if for every differentiable $f:\R^d \rightarrow \R$,  
$\Var_{X \sim D}[f(X)] \leq \sigma^2 \E_{X \sim D} [ \Norm{\nabla f(X)}_2^2 ]$).  
We show in \Cref{sec:poincareAppendix} that this class  also has certifiably bounded moments in $k$-sparse directions, 
in the sense of \Cref{def:bounded-moments-k-sparse}. 
Combining this with the fact that the tails of $\sigma$-Poincar\'e distributions 
are inherently subexponential, \Cref{thm:main-poincare-informal} is applicable.

We complement our upper bound with a Statistical Query (SQ) lower bound (and low-degree testing lower bound), 
which gives evidence that the factor $k^{O(t)}$ in the sample complexity of \Cref{thm:main-poincare-informal} 
might be necessary for efficient algorithms. 

We remind the reader that SQ algorithms~\cite{Kearns:98} 
do not draw samples from the data distribution, 
but instead have access to an oracle that can return 
the expectation of any bounded function (up to a desired additive error). 
Specifically, an SQ algorithm is able to perform 
adaptive queries to a  $\mathrm{STAT}(\tau)$ oracle, which we define below. 
\begin{restatable}[STAT Oracle]{definition}{STATDEF} \label{def:stat}
Let $D$ be a distribution on $\R^d$. A statistical query is a bounded function $f : \R^d \to [-1,1]$. 
For $\tau>0$, the $\mathrm{STAT}(\tau)$ oracle responds to the query $f$ with a value $v$ 
such that $|v - \E_{X \sim D}[f(X)] | \leq \tau$.
We call $\tau$ the tolerance of the statistical query.
\end{restatable}
An SQ lower bound is an unconditional lower bound showing that for any SQ algorithm, 
either the number of queries $q$ must be large or the tolerance of some query, $\tau$, must be small. 
The standard interpretation of SQ lower bounds hinges on the fact that simulating 
a query to $\mathrm{STAT}(\tau)$ using i.i.d.\ samples may require $\Omega(1/\tau^2)$ many samples. 
Thus, an SQ lower bound stating that any SQ algorithm either makes $r$ queries or needs tolerance $\tau$ 
is interpreted as a tradeoff between runtime $\Omega(r)$ and sample complexity $\Omega(1/\tau^2)$.

Recall that a distribution $D$ is subgaussian if there exists an absolute constant $c$
such that for all unit vectors $v$ it holds that 
$\E_{ X  \sim D}\Brac{|\langle v, X - \E_{X \sim D}[X] \rangle|^p}^{1/p} \leq c\sqrt{p}$. 
We show the following (see \Cref{thm:bounded_moment_sparse} for a detailed formal statement). 

\begin{theorem}[SQ Lower Bound for Subgaussian Distributions, Informal Statement] \label{thm:subgaussianSQ_informal}
Fix $t \in \N$ with $t \geq 2$ and assume that $d \geq k^2$ for $k$ sufficiently large. 
Any SQ algorithm that obtains error $o(\eps^{1-1/t})$ for robust sparse mean estimation 
of a subgaussian distribution (with $t$-th moments certifiably bounded in $k$-sparse directions) 
either requires $d^{\Omega(\sqrt{k})}$ statistical queries or makes at least one query with tolerance $k^{-\Omega(t)}$.    
\end{theorem}

\noindent For the statement of our low-degree testing lower bound, see~\Cref{sec:low_degree}.
Informally speaking, \Cref{thm:subgaussianSQ_informal} shows that any SQ algorithm 
that returns a $\widehat{\mu}$ satisfying $\|\widehat{\mu} - \mu\|_{2,k} = o(\epsilon^{1-1/t})$ 
requires runtime exponential in $k$, unless it uses queries of tolerance $k^{-\Omega(t)}$ --- 
requiring $k^{\Omega(t)}$ samples for simulation. 
We briefly remark that \Cref{thm:subgaussianSQ_informal} also has implications for the dense setting: 
By taking $k = \sqrt{d}$, \Cref{thm:subgaussianSQ_informal} suggests that $d^{\Omega(t)}$ samples 
may be necessary to efficiently obtain $o(\epsilon^{1-1/t})$ error in the dense setting;
this qualitatively matches the algorithmic results of \cite{kothari2017outlier}.

Interestingly, \Cref{thm:subgaussianSQ_informal} does not apply 
when the inlier distribution is Gaussian, 
i.e., the SQ-hard instance of \Cref{thm:subgaussianSQ_informal} is {\em not} a Gaussian distribution. 
In fact, our next algorithmic result shows that it {\em is} possible to achieve the near-optimal error 
of $\tilde O(\eps) \, \sqrt{\Norm{\Sigma}_2}$ for $\mathcal{N}(\mu, \Sigma)$, using 
$(k^4/\eps^2) \polylog(d/\eps)$ samples.

\begin{theorem}[Robust Sparse Gaussian Mean Estimation]\label{thm:main-gaussian-informal}
Let $k,d \in \Z_+$ with $k\leq d$ and $\eps<\eps_0$ for a sufficiently small constant $\eps_0>0$. 
Let $\mu \in \R^d$ and $\Sigma \in \R^{d \times d}$ be a  positive semidefinite matrix. 
There exists an algorithm which, given $\eps$,$k$, and an $\eps$-corrupted set of samples from $\cN(\mu,\Sigma)$ 
of size $m = O(({k^4}/{\eps^2}) \log^5(d/(\eps)))$, runs in time $\poly(md)$,
and returns an estimate $\widehat{\mu}$ such that 
$\|\widehat{\mu} - \mu \|_{2, k} \leq \tilde{O}(\eps) \, \sqrt{\|\Sigma\|_2}$ with high probability.
\end{theorem}

Information-theoretically, $O(k \log(d/k))/\eps^2$ samples suffice to obtain $O(\eps)$ error 
(see \Cref{thm:sample-comp-bdd-mmnts}). Prior work has given evidence
that $\Omega(k^2)$ samples might be necessary for efficient algorithms 
to obtain dimension-independent error~\cite{DKS17-sq, BrennanB20}.
Perhaps surprisingly, here we establish an SQ lower bound suggesting that the $\Omega(k^4)$ 
sample complexity of our algorithm might be inherent for efficient algorithms to achieve error $o(\eps^{1/2})$ 
(see \Cref{thm:k4lower_bound} for a formal  statement):

\begin{theorem}[SQ Lower Bound for Gaussian Sparse Mean Estimation, Informal Statement]\label{thm:SQ4th_informal}
Let $0 < c < 1$ and assume that $d \geq k^2$ for $k$ sufficiently large.
Any SQ algorithm that performs robust sparse mean estimation of Gaussians with $\Sigma \preceq I$ up to error $o(\sqrt{\eps})$
does one of the following: It either requires $d^{\Omega(c{k^c})}$ queries or makes 
at least one query with tolerance $O(k^{-2 + 2c})$.    
\end{theorem}

The intuitive interpretation of \Cref{thm:SQ4th_informal} is that any SQ algorithm for this task 
either has runtime $d^{\poly(k)}$ or uses $\Omega(k^{4})$ samples 
(a similar hardness holds for low-degree polynomial tests; see \Cref{thm:hypothesis-testing-hardness}).

\subsection{Overview of Techniques}

To establish \Cref{thm:main-poincare-informal,thm:main-gaussian-informal},
we use the \emph{sum-of-squares} framework, i.e., solve a \emph{sum-of-squares} (SoS) 
SDP relaxation of a system of polynomial inequalities.

\subsubsection{Robust Sparse Mean Estimation with Bounded Moments}

\paragraph{Identifiability in the Presence of Outliers} 

Similar to \cite{kothari2017outlier}, our starting point is a set of polynomial constraints (\Cref{def:axioms}) 
in which the variables try to identify the uncorrupted samples. The program has a vector of variables 
for every sample in the set (we will refer to these variables as ``ghost samples''), 
and enforces that these ghost samples match a $(1-\eps)$-fraction of the data. 
The constraints also enforce that the uniform distribution over the ghost samples 
has $t$-th moments certifiably bounded by $M$ \emph{in $k$-sparse directions} 
(\Cref{def:bounded-moments-k-sparse}).  A key property of an SoS relaxation 
is that it satisfies any polynomial inequality that is true subject to the constraints 
of the original polynomial system, \emph{as long as} this inequality 
has an ``SoS proof'' of degree $t$ (i.e., the difference of the two sides is a sum of square polynomials, 
where each polynomial is of degree at most $t$). We give an SoS proof of the fact 
that the mean of the ghost samples is close to the true mean in $k$-sparse directions 
(proof of identifiability; see \Cref{sec:proof_of_main_theorem}).

\vspace{-0.2cm}

\paragraph{Sampling Preserves Certifiably Bounded Moments} 

For our identifiability proof (and to ensure feasibility of our program), 
we require that the uniform distribution over the uncorrupted samples 
satisfies certifiably bounded $t$-th moments in $k$-sparse directions. 
However, we initially know only that the distribution from which these samples 
are drawn has $t$-th moments certifiably bounded by $M$ in $k$-sparse directions. 
Given that the underlying inlier distribution satisfies certifiably bounded moments, 
we need to show that the property transfers to the empirical distribution. 
In the dense case, it is relatively easy to prove such a concentration result in $\poly(d)$ samples 
for all $v \in \R^d$ via spectral concentration inequalities. 
In contrast, doing this in the sparse setting with $\poly(k)$ samples requires an alternate approach.

The first step towards showing this is \Cref{claim:sparse-upper-bound}, 
which states that for all polynomials $r(v) = \sum_{I \in [d]^t} r_I \prod_{j \in [t]} v_{I_j}$ 
and $k$-sparse vectors $v$, the inequality $r(v)^2 \leq k^t \max_{I \in [d]^t} r_I^2$ has an SoS proof. 
Applying this to the true and empirical moments,
$p(v) := \E_{X \sim D}\Brac{\iprod{v, X- \mu}^t}$ and 
$\widehat p(v) = \E_{i \sim [m]}\Brac{\iprod{v, X_i-\overline \mu}^t}$, 
we see that there is an SoS proof of 
$(p(v)-\widehat{p}(v))^2 \leq k^t\| \E_{i \sim [m]}[(X_i - \mu)^{\otimes t}] - \E_{X \sim \D}[(X - \mu)^{\otimes t}] \|_{\infty}^2$. 

In \Cref{sec:concentration_appendix}, we show concentration 
of the $\ell_\infty$-norm of the aforementioned tensor.  
This happens because of a union bound that we perform in the concentration result. 
Note that while we need only $O(\log(d))$ moments to be bounded for concentration, 
prior work in the sparse setting assumes that the distribution has known covariance and 
is either subgaussian or subexponential.

\subsubsection{Achieving Near-optimal Error for Gaussian Inliers} \label{sec:techniques_gaussian}

The first polynomial-time algorithm for robustly learning an arbitrary Gaussian (in the dense setting) was given in
\cite{DKKLMS16}. In particular, that work shows how to robustly estimate the mean in $\ell_2$-norm 
and the covariance in Mahalanobis norm up to an error of $\tilde{O}(\eps)$ using $\tilde{O}(d^2/\eps^2)$ samples. 
It is not clear how to directly adapt the approach of \cite{DKKLMS16} to the sparse setting, 
while achieving the desired sample complexity of $\poly(k, \log(d), 1/\eps)$.
Our starting point will be the recent work by~\cite{kothari2021polynomial} that was able to (nearly) match the aforementioned
guarantees (for the dense Gaussian setting) using the SoS method. The difficulty of achieving this lies in the fact that 
standard SoS approaches usually require concentration of degree-$t$ polynomials to get error $O(\eps^{1-1/t})$; 
and $t$ would need to be $\sqrt{\log(1/\eps)}$ to make this $\tilde{O}(\eps)$.
\cite{kothari2021polynomial} was able to achieve this result using SoS certifiability of bounded moments 
only up to degree $4$.

We now explain how we adapt the approach of \cite{kothari2021polynomial} to the sparse setting.
Assume that the covariance of the inliers is spectrally bounded, namely that $\norm{\Sigma}_2 \leq 1$. 
For the dense result obtained in~\cite{kothari2021polynomial}, 
it suffices to show SoS proofs of \emph{multiplicative} concentration inequalities 
for Gaussian polynomials of degree up to $4$. In the absence of sparsity constraints, 
this is done by  standard spectral matrix concentration. 
Unfortunately, the technique from the previous section only gives us an \emph{additive} concentration inequality. 
This difference is significant and makes it challenging to obtain 
a guarantee scaling with $\Norm{\Sigma}_2$. To circumvent this issue, we add  independent noise to each sample generated as $\cN(0, I)$, 
which ensures that $I \preceq \Sigma \preceq 2I$, while keeping the mean unaffected. 
We thus end up with an estimator of $\tilde{O}(\eps)$ error for the case that $\Sigma \preceq 2I$. 
On the other hand, if $\norm{\Sigma}_2$ is much smaller than $1$, 
then the right error guarantee is $\tilde{O}(\eps) \, \sqrt{\norm{\Sigma}_2} $. 
In \Cref{app:lepski}, we use Lepskii's method~\cite{lepskii1991problem} 
to improve the error guarantee so that it scales with $\Norm{\Sigma}_2$, as desired. 
We first obtain a rough estimate of $\|\Sigma\|_2$ that is within $\poly(d)$ factor away 
from the true value (by taking the median of $\norm{X_i - X_j}_2$, where the $X_i$ are samples).
We next run our robust estimation algorithm after convolving the data with noise at various scales $\sigma$ 
and getting a corresponding estimate. With high probability, whenever our candidate upper bound, $\sigma$, 
is bigger than $\sqrt{\norm{\Sigma}_2}$, we get a point within distance $\tilde O (\eps) \sigma$ of the true mean. 
We then find the smallest value of $\sigma$ such that the output is consistent 
with the larger values of $\sigma$ and return the corresponding estimate. 

\subsubsection{Statistical Query and Low-Degree Testing Lower Bounds}

Our SQ lower bounds leverage the framework of \cite{DKS17-sq} which showed the following: 
If $A$ is a one-dimensional distribution matching its first $m$ moments 
with $\cN(0,1)$, then distinguishing between $\cN(0,I)$ and the $d$-dimensional distribution 
that coincides with $A$ in an unknown $k$-sparse direction 
but is standard Gaussian in all perpendicular directions requires either $q = d^{\poly(k)}$ queries 
or tolerance $\tau < \frac{1}{k^{(m+1)/2}}$ in the SQ model. 
The robust sparse mean estimation problems that we consider can be phrased in this form and
the challenge is to construct the appropriate moment-matching distributions.

In \Cref{thm:SQ4th_informal}, we show a lower bound of $\Omega(k^4)$ on the sample complexity
of any efficient SQ algorithm that robustly estimates a sparse mean within error $o(\sqrt{\eps})$.
Interestingly, this lower bound nearly matches the sample complexity of our algorithm (\Cref{thm:main-gaussian-informal}). 
We view this information-computation tradeoff as rather surprising. 
Recall that in the (easier) case where the covariance of the inliers is known to be the identity, 
$O(k^2 \log d)$ samples are sufficient for efficient algorithms~\cite{Li17-sparse}, and there is evidence that this number of samples is also necessary for efficient algorithms~\cite{DKS17-sq, BrennanB20}.
To prove our SQ lower bound in this case, we need to construct a univariate density $A$ 
that matches (i) $m=3$ moments with $\cN(0,1)$, and (ii)
$A$ is $\eps$-corruption of $\cN(\Theta(\sqrt{\eps}),1)$.
To achieve this, we leverage a lemma from~\cite{DKS19} that 
lets $A$ have a Gaussian inlier component with mean $\Theta(\sqrt{\eps})$ 
and variance slightly smaller than $1$. A suitable outlier component 
can then correct the first three moments of the overall mixture, 
so that they match the moments of $\cN(0,1)$. 

A more sophisticated choice of $A$ is needed to establish our \Cref{thm:subgaussianSQ_informal}. 
Specifically, we need to select $A = (1-\eps) G + \eps B$, where (i) $A$ matches its first $t$ moments with $\cN(0,1)$, 
(ii) $G$ is an explicit subgaussian distribution, and (iii) $\E_{X \sim G} [X] = \Omega(\eps^{1 - 1/t})$.
For $G$, we start with a shifted Gaussian, $\cN(\Theta(\eps^{1 - 1/t}),1)$, 
that we modify by adding an $t$-degree polynomial $p(x)$ in $[-1,1]$. 
Since we modify the Gaussian only on $[-1,1]$, the distribution continues to be subgaussian.
By imposing the moment-matching conditions and expanding $p(x)$ in the basis of Legendre polynomials, 
we show that such a $p(\cdot)$ exists, so that (i)-(iii) above hold. 
We also show that the constructed distributions have SoS certifiable bounded $t$-th moments, 
and hence fall into the class of distributions for which our upper bounds apply 
(see \Cref{sec:sos-cert-hard-instance}). 

Finally, by exploiting the relationship between the SQ model and low-degree polynomial tests from~\cite{BBHLS20}, 
we also obtain quantitatively similar lower bounds against low-degree polynomial tests. 
The information-theoretic characterization of error and sample complexity 
appear in \Cref{app:info-th-sample-comp}.

\subsection{Prior and Related Work}

After the early works of \cite{DKKLMS16,LaiRV16}, the field of high-dimensional algorithmic robust statistics has seen a plethora of research activity.
Focusing on the dense setting, the prior work has obtained computationally-efficient
and sample-efficient algorithms for a variety of problems: 
mean estimation~\cite{DKK+17,DepLec19,DHL19,DiaKP20}, covariance and higher moment estimation~\cite{kothari2017outlier,CDGW19}, 
linear regression~\cite{KlivansKM18,DKS19-lr,PenJL20,BakPra21}, mixture models~\cite{KStein17,HopkinsL18,DKS18-list,LM20,BDJKKV20}, 
and  stochastic convex optimization~\cite{PraSBR20,DiakonikolasKKLSS2018sever}.
We remark that a subset of these algorithms also rely on SoS algorithms.

Finally, we discuss results that leverage sparsity to improve sample complexity for computationally-efficient algorithms.
\cite{BDLS17} presented the first computationally-efficient algorithms for a range of 
sparse estimation tasks including mean estimation. 
However, their estimation algorithm crucially relies on the fact that the inlier distribution is Gaussian with identity covariance.
As opposed to the convex programming approach of \cite{BDLS17}, \cite{DKKPS19-sparse} proposed a spectral algorithm for sparse robust mean estimation of identity-covariance Gaussians.
Recently, \cite{CDKGS21-sparse} proposed a non-convex formulation and showed that any approximate stationary point (that can be obtained by efficient first-order algorithms) suffices.
We reiterate that none of these algorithms give $o(\sqrt{\eps})$ error when the covariance of the inliers is unknown.

Finally, we mention that median-of-means preprocessing has been applied to get $O(\sqrt{\eps})$ error optimally for robust mean estimation~\cite{DepLec19,DiaKP20,HopLZ20,LeiLVZ20}. However, median-of-means preprocessing does not give $o(\sqrt{\eps})$ error even when the inliers are  Gaussian with identity covariance, and thus not applicable to our setting.

\subsection{Organization}

The structure of this paper is as follows: 
In \Cref{sec:prelims}, we define the necessary notation and record basic facts about the SoS framework. 
In \Cref{sec:sparsity}, we define the notion of \emph{certifiably bounded central moments} 
in $k$-sparse directions, and show that this property is preserved under sampling. 
In \Cref{sec:sparse-mean-est}, we give an SoS algorithm for robust sparse mean estimation 
under certifiably bounded central moments in sparse directions, 
establishing \Cref{thm:main-poincare-informal}. 
In~\Cref{sec:Gaussian-O-eps}, we give an efficient estimator that achieves near-optimal error 
for Gaussian distributions with unknown covariance, establishing \Cref{thm:main-gaussian-informal}. 
Finally, in \Cref{sec:sq-lowerbd}, we prove our SQ lower bounds for the previous two settings, 
establishing \Cref{thm:subgaussianSQ_informal,thm:SQ4th_informal}.
For ease of the exposition, some technical proofs are deferred to the appendix.

%% file: prelim.tex
\section{Preliminaries} \label{sec:prelims}
\paragraph{Basic Notation} 
 We use $\N$ to denote natural numbers and $\Z_+$ to denote positive integers. For $n \in \Z_+$ we  denote $[n] := \{1,\ldots,n\}$. We denote by $\mathbf{1}(\cE)$ the indicator function of the event $\cE$. 
For $a_1(x),\dots,a_d(x)$ polynomials in $x$ and an ordered tuple $T \in [d]^t$, we use $a_T(x)$ to define the polynomial $a_T(x) := \prod_{i\in T}a_i(x)$. We denote by $\R[x_1,\ldots, x_d]_{\leq t}$ the class of real-valued polynomials of degree at most $t$ in variables $x_1,\ldots, x_d$.
 We use $\poly(\cdot)$ to indicate a quantity that is polynomial in its arguments. Similarly, $\polylog(\cdot)$ denotes a quantity that is polynomial in the logarithm of its arguments. For an ordered set of variables $V = \{x_1, \dots, x_n\}$, we will denote $p(V)$ to mean $p(x_1, \dots, x_n)$.
 
\paragraph{Linear Algebra Notation} We use $I_d$ to denote the $d\times d$ identity matrix. We will drop the subscript when it is clear from the context.
We typically use small case letters for deterministic vectors and scalars. We will specify the dimensionality unless it is clear from the context. 
We denote by $e_1,\ldots,e_d$ the vectors of the standard orthonormal basis, i.e., the $j$-th coordinate of $e_i$ is equal to $\mathbf{1}_{\{i=j \}}$, for $i,j \in [d]$.
We use $\cS^{d-1}$ to denote the $d$-dimensional unit sphere.
 For a vector $v$, we let $\|v\|_2$ denote its $\ell_2$-norm. We call a vector $k$-sparse if it has at most $k$ non-zero coordinates. We define the set of $k$-sparse $d$-dimensional unit-norm vectors as $\cU_k^d := \{ x \in \R^d \; : \; \text{$x$ is $k$-sparse}, \|x\|_2=1\}$. We will often drop the superscript when it is clear from the context. 
 We use $\langle v,u \rangle$ for the inner product of the vectors $u,v$.
For a matrix $A$, we use $\|A\|_F, \|A\|_{2}, \|A\|_\infty$ to denote the Frobenius, spectral,  and entry-wise infinity-norm. 
We denote  the trace of $A$ by $\tr(A)$ and  the number of nonzero entries  in $A$ by $\|A\|_0$. For two matrices $A,B \in \R^{m \times d}$, we define the inner product $\langle A,B\rangle := \tr(A^TB)$.
For a matrix $A \in \R^{d \times d}$, we use $A^\flat$ to denote the flattened vector in $\R^{d^2}$, and for a $v \in \R^{d^2}$, we use $v^\sharp$ to denote the unique matrix $A$ such that $A^\flat = v^\sharp$. We say a symmetric matrix $A$ is PSD (positive semidefinite) and write $A \succeq 0$ if $x^T A x \geq 0$ for all vectors $x$. We write $A\preceq B$ when $B-A$ is PSD.  We will use $\cdot ^{\otimes s}$ to denote the standard Kronecker product. 

\paragraph{Probability Notation} We use capital letters for random variables. For a random variable $X$, we use $\E[X]$ for its expectation.	
We use $\cN(\mu,\Sigma)$ to denote the Gaussian distribution with mean $\mu$ and covariance matrix $\Sigma$. We let $\phi$ denote the pdf of the one-dimensional standard Gaussian.
When $D$ is a distribution, we use $X \sim D$ to denote that the random variable $X$ is distributed according to $D$. When $S$ is a set, we let $\E_{X \sim S}[\cdot]$ denote the expectation under the uniform distribution over $S$. For any sequence $a_1, \dots, a_m \in \mathbb{R}^d$, we will also use $\E_{i \sim [m]}[a_i]$ to denote $\frac 1 m \sum_{i \in [m]} a_i$. 
For a real-valued random variable $X$ and $p \geq 1$, we use $\|X\|_{L_p}$ to denotes its $L_p$ norm, i.e., $\|X\|_{L_p} := (\E[|X|^p])^{1/p}$.

\medskip

The following fact (proved in \Cref{app:addDetails} for completeness) can be used to translate bounds from the $(2,k)$-norm to the usual $\ell_2$-norm when the underlying mean $\mu$ is sparse:
\begin{restatable}{fact}{FactTruncSparse}
	\label{fact:sparseTruncation}
	Let $h_k: \R^d \to \R^d$ denote the function where $h_k(x)$ is defined to truncate $x$ to its $k$ largest coordinates in magnitude and zero out the rest. For all $\mu \in \R^d$ that are $k$-sparse, we have that $\|h_k(x) - \mu\|_2 \leq 3 \|x-\mu\|_{2,k}$.
\end{restatable}

\subsection{SoS Preliminaries} 
\label{sec:sosprelims}
The following notation and preliminaries are specific to the SoS part of this paper.  We refer the reader to \cite{BarakSteurerNotes} for a complete treatment of basic definitions about the SoS hierarchy and SoS proofs. Here, we review the basics. Our algorithms will work under the condition that the numerical precision of all the numbers involved is controlled. To describe these conditions formally, we use the notion of bit complexity, which captures the size of the representation of a number.

\begin{definition}[Bit complexity]
The bit complexity of an integer $z \in \Z$ is $1 + \lceil \log_2 z \rceil$. The bit complexity of a rational number $r/t$ is the sum of the individual bit complexities of $r$ and $t$. The bit complexity of a vector is the sum of the bit complexities of its coordinates and the bit complexity of a set of vectors is the sum of the bit complexities of the set's elements. 
\end{definition}

\begin{definition}[Symbolic polynomial]
	A degree-$t$ symbolic polynomial $p$ is a collection of indeterminates $\widehat{p}(\alpha)$, 
	one for each multiset $\alpha \subseteq [d]$ of size at most $t$.
	We think of it as representing a polynomial $p \, : \, \R^d \rightarrow \R$ whose coefficients are themselves 
	indeterminates via $p(x) = \sum_{\alpha \subseteq [d], |\alpha| \leq t} \widehat{p}(\alpha) x^\alpha$.
\end{definition}

\begin{definition}[SoS Proof]\label{def:sos-proof}
	Let $x_1,\ldots,x_d$ be indeterminates and let $\cA$ be a set of polynomial inequalities $\{ p_1(x) \geq 0,\ldots,p_m(x) \geq 0 \}$.
	An SoS proof of the inequality $r(x) \geq 0$ from axioms $\cA$ is a set of polynomials $\{r_S(x)\}_{S \subseteq [m]}$ such that each $r_S$ is a sum of square polynomials and $r(x) = \sum_{S \subseteq [m]} r_S(x) \prod_{i \in S} p_i(x).$
	If the polynomials $r_S(x) \cdot \prod_{i \in S} p_i(x)$ have degree at most $t$ for all $S \subseteq[m]$, we say that this proof is of degree $t$ and denote it by $\cA \sststile{t}{} r(x) \geq 0$.  The bit complexity of the SoS proof is the sum of the bit complexities of the coefficients of the polynomials $r_S$ and $p_i$.
	
	When we need to emphasize what indeterminates are involved in a particular SoS proof, we denote it by  $\cA \sststile{t}{x} r(x) \geq 0$.
	When $\cA$ is empty, we directly write $\sststile{t}{} r(x) \geq 0$ and $\sststile{t}{x} r(x) \geq 0$.
	We also often refer to $\cA$ containing polynomial equations $q(x) = 0$, by which we mean that $\cA$ contains both $q(x) \geq 0$ and $q(x) \leq 0$. 
\end{definition}

We frequently compose SoS proofs without comment --- see \cite{BarakSteurerNotes} for basic facts about composition of SoS proofs and bounds on the degree of the resulting proofs.

Our algorithm also uses the dual objects to SoS proofs, commonly called \emph{pseudoexpectations}.
\begin{definition}[Pseudoexpectation]
	Let $x_1,\ldots,x_d$ be indeterminates.
	A degree-$t$ pseudoexpectation $\pE$ is a linear map $\pE \, : \, \R[x_1,\ldots,x_d]_{\leq t} \rightarrow \R$
	from degree-$t$ polynomials to $\R$ such that $\pE \Brac{p(x)^2} \geq 0$ for any $p$ of degree at most $t/2$ and $\pE \Brac{1} = 1$.
	If $\cA = \{p_1(x) \geq 0, \ldots, p_m(x) \geq 0\}$ is a set of polynomial inequalities, we say that $\pE$ satisfies $\cA$ if for every $S \subset [m]$,
	the following holds: $\pE [s(x)^2 \prod_{i \in S} p_i(x)] \geq 0$  for all squares $s(x)^2$ such that $s(x)^2 \prod_{i \in S} p_i(x)$ has degree at most $d$. 
	
	We say that a pseudoexpectation is $\tau$-approximate if it satisfies all the conditions up to slack $\tau$, i.e., $\pE[p(x)^2] \geq -\tau \|p\|_2^2$ for any $p$ of degree at most $t/2$ and $\pE [s(x)^2 \prod_{i \in S} p_i(x)] \geq -\tau \|s^2\|_2 \prod_{i \in S} \|p_i\|_2$ for all sets $S$ and polynomials $s(x)$ such that $s(x)^2\prod_{i \in S} p_i(x)$ has degree at most $d$, where $\|p\|_2$ denotes the $\ell_2$-norm of the vector of coefficients of $p$.
\end{definition}

We will also rely on the algorithmic fact that given a satisfiable system $\cA$ of $m$ polynomial  inequalities in $d$ variables, there is an algorithm which runs in time $(dm)^{O(t)}$ and computes a pseudoexpectation of degree $t$ approximately satisfying $\cA$. 
\begin{theorem}[The SoS Algorithm \cite{Sho87,lasserre2001new, nesterov2000squared,bomze1998standard}] \label{thm:sos-algo}
	Let $\cA$ be a satisfiable system of $m$ polynomial inequalities in variables 
	$x_1,\ldots,x_d$, each with coefficients having bit complexity at most $B$ and degree at most $t$. 
	Suppose that $\cA$ contains an inequality of the form $\|x\|_2^2 \leq M$, with $M$ having bit complexity at most $B$.
	There is an algorithm which takes $t \in \Z_+$, $\tau$, and $B$
	and returns in time $\poly(B,\log(1/\tau),d^t,m^t)$ a degree-$t$ pseudo-expectation $\pE$ which satisfies $\cA$ 
	up to error $\tau$.
\end{theorem}
All of our SoS proofs will be of bit complexity $\poly(m^t,d^t)$.
We thus apply \Cref{thm:sos-algo} with  $B=\poly(m^t,d^t)$ and $\tau = 2^{-\poly(tB)}$ to ensure that the total error that we incur is at most  $O( 2^{-md})$.
Since this error is negligible, we will not treat it explicitly in the remainder of the paper. 

Pseudoexpectations satisfy several basic inequalities, some of which are Cauchy-Schwartz, H\"older and a modified version of the triangle inequality. We will use these extensively. See~\Cref{sec:sos_prelims} for details.

%% file: certifiable-k-sparse-concentration.tex
\section{Certifiably Bounded Moments in Sparse Directions}
\label{sec:sparsity}

Our algorithm succeeds whenever the uncorrupted samples have \emph{certifiably bounded moments in $k$-sparse directions}. We first formally define this property in \Cref{sec:certifiable-def}. We then show in \Cref{sec:sampling} that this property is preserved under sampling.

\subsection{Definitions} \label{sec:certifiable-def}
To define the property of certifiably bounded moments in sparse directions we first need to capture the sparsity of vectors using polynomial equations, which we do as follows:

\begin{definition}\label[definition]{def:A-k-sparse}
	We use $\cAksparse$ to denote the following system of equations over $v_1,\dots,v_d,z_1,\dots,z_d$:
\begin{align*}
\cAksparse := \{ z_i^2 = z_i\}_{i \in [d]} \cup \{v_i z_i = v_i\}_{i \in [d]}\cup \Set{ \sum_{i=1}^d z_i = k } \cup \Set{\sum_{i=1}^d v_i^2 = 1 }\;.
\end{align*}	
\end{definition} 
A vector $v= (v_1,\dots,v_d)$ is $k$-sparse if and only if there exists $z = (z_1,\dots,z_d)$ such that $v,z$ satisfy $\cAksparse$. Here, the $z_i$'s (roughly) correspond to the support of the vector $v$ in the sense that $v_i$ being non-zero implies $z_i$ being one.
We will crucially need the notion of the $t$-th moment of a distribution being certifiably bounded. 
\begin{definition}[$(M,t)$ Certifiably Bounded Moments in $k$-sparse Directions]
	\label[definition]{def:bounded-moments-k-sparse}
	For an $M > 0$ and even  $t \in \N$, we say that the distribution $D$ with mean $\mu$ satisfies \emph{$(M,t)$ certifiably bounded moments in $k$-sparse directions} with bit complexity bounded by $B$ if 
	\begin{align}\label{eqn:cert-def}
	\cAksparse \sststile{O(t)}{v,z} \E_{X \sim D}\Brac{\iprod{ v,  X - \mu}^{t} }^2 \leq M^2 \;,
	\end{align}
	and the bit complexity of the SoS proof of \eqref{eqn:cert-def} is bounded by $B$. \footnote{We will assume $B \leq \poly(m^t,d^t)$ in this paper.}

\end{definition}

A broad and natural class of distributions satisfying \Cref{def:bounded-moments-k-sparse} is implicit in Theorem 4.1 from \cite{KStein17}. Their result  says that if a distribution $D$ is $\sigma$-Poincar\'e, i.e.,  for all differentiable functions $f:\R^d \rightarrow \R$,  $\Var_{X \sim D}\Brac{f(X)} \leq \sigma^2 \E_{X \sim D} [ \Norm{\nabla f(X)}_2^2]$, then it has certifiably bounded moments \emph{in every direction $v$},  i.e., the appropriate inequality follows even ignoring the $z$ constraints in $\cAksparse$. Moreover, the bit complexity of this proof is at most $\poly(t,b)$, where $b$ is the bit complexity of the coefficients of the polynomial $M - \E_{X \sim D}\Brac{\iprod{ v,  X - \mu}^{t} }^2$.
It can be seen (see \Cref{sec:poincareAppendix}) that this class also satisfies \Cref{def:bounded-moments-k-sparse}.

\subsection{Sampling and Certifiably Bounded Moments in Sparse Directions} \label{sec:sampling}

Our algorithm will work when the uniform distribution over the samples has certifiably bounded central moments in $k$-sparse directions. In this section we show that sampling from distributions with not too heavy tails preserves this property of having certifiably bounded moments. In particular, we show the following:
\begin{restatable}{lemma}{SamplingPreserves}
\label[lemma]{lem:pop-to-empirical}
	Let $D$ be a distribution over $\mathbb{R}^d$ with mean $\mu$ and suppose $D$ has $c$-sub-exponential tails around $\mu$ for a constant $c$.  Suppose that $D$ satisfies $\cAksparse \sststile{2t}{v, z} \E_{X \sim D} \Brac{\iprod{v, X-\mu}^t}^2 \leq M^2$ with bit complexity at most B . 
	Let $S= \{X_1, \ldots, X_m\}$ be a set of $m$ i.i.d.\ samples from $D$, $D'$ be the uniform distribution over $S$, and $\overline{\mu}:=\E_{X \sim D'}[X]$.
	If $m > C (t k (\log d))^{5 t} ~\max(1,M^{-2})/\eps^{2-2/t}$ for a sufficiently large constant $C$, then with probability at least $0.9$ we have the following two: 
	\begin{enumerate}
	\item $\cAksparse \sststile{2t}{v, z} \E_{X \sim D'} \Brac{\iprod{v, X-\overline{\mu}}^t}^2 \leq 8M^2$,  and the bit complexity of the proof is at most $\poly(tb,d^t,B)$, where $b$ is the bit complexity of the set $S$. 
	\item $\|\overline{\mu} -\mu \|_{2,k} \leq M^{1/t}\eps^{1-1/t}$. 
\end{enumerate}		
\end{restatable}

Before presenting the full proof, we sketch the proof of \Cref{lem:pop-to-empirical}. The claim $\|\overline{\mu} -\mu \|_{2,k} \leq M^{1/t}\eps^{1-1/t}$ follows from a standard Markov inequality. We thus focus on the first claim. The first part depends on the following key lemma, stating that polynomials over $\cAksparse$ are bounded by the square of the maximum coefficient times $k^t$:  

\begin{restatable}[Polynomials of $k$-sparse vectors are bounded]
{lemma}{LemSparsePolyUB}\label{claim:sparse-upper-bound}
	Let $p(v_1, \dots, v_d) = \sum_{T \in  [d]^t} a_T v_T$ be a polynomial of degree $t$, where the coefficients $\{ a_{T} \}_{T \in [d]^t} \subset \mathbb{R}$ are real numbers (not variables of the SoS program), then
 	\[ 
	\cAksparse \sststile{2t}{v, z} p(v_1, \dots, v_d)^2 \leq  k^{t} \max \{a_T^2 \mid T \in [d]^t\} \;. 
 	\] 
The bit complexity of the SoS proof above is $\poly(d^t,  \max_{T} b_T)$, where $b_T$ denotes the bit complexity of $a_T$. 
\end{restatable}
\begin{proof} 
We have the sequence of inequalities
	\begin{align*}
		\cAksparse &\sststile{2t}{v, z}\left( \sum_{T \in [d]^t} a_T v_T\right)^2 = \left( \sum_{T \in [d]^t} a_T z_T v_T\right)^2 \\
		&\leq  \left( \sum_{T \in [d]^t} a_T^2 z_T^2 \right)  \left( \sum_{T \in [d]^t} v_T^2\right) \\
		&\leq \Paren{\max_{T \in [d]^t} (a_T)^2} \left( \sum_{T \in [d]^t} z_T^2 \right) \left( \sum_{T \in [d]^t} v_T^2 \right)  \\
		&= \Paren{\max_{T \in [d]^t} (a_T)^2} \left( \sum_{i = 1}^d z_i^2 \right)^t \left( \sum_{i = 1}^d v_i^2 \right)^t \\
		&= k^{t} \max_{T \in [d]^t} (a_T)^2  \;,
	\end{align*}
where the first line uses $\{v_i z_i = v_i\}_{i \in [d]}$, the second line uses SoS Cauchy-Schwartz (\Cref{fact:sos-holder}), the third line uses that $\sststile{2t}{v, z}\sum_{T \in [d]^t}(\max_{T \in [d]^t} (a_T)^2 - a_T^2) z_T^2 \geq 0$, 
the fourth line uses the identities that $\sum_{T \in [d]^t} z_T  = ( \sum_{i \in [d]} z_i )^t$ and $\sum_{T \in [d]^t} v_T  = ( \sum_{i \in [d]} v_i )^t$,
and the last line uses the axioms $\{z_i^2 = z_i\}_{i \in [d]} \cup \{\sum_{i=1}^d z_i = k\} \cup \{\sum_{i=1}^d v_i^2 = 1\}$.  
\end{proof}

Continuing with the proof sketch of \Cref{lem:pop-to-empirical}, since $D$ has subexponential tails, with high probability, the $\ell_{\infty}$ norm of the difference between the expected and empirical central moments calculated using $m$ samples is at most $M/\sqrt{k^t}$ (see \Cref{lem:basic_linf_consc-full} below). 
    Then, an application of ~\Cref{claim:sparse-upper-bound} to $p(v) = \sum_{ T \in [d]^{t}} \Paren{\E_{i \sim [m]}\left[X_i - \ovl \mu\right]_T - \E_{X \sim D} \left[X - \mu \right]_T } v_T$ and the SoS triangle inequality completes the proof. 
    The formal proof of \Cref{lem:pop-to-empirical} is given in \Cref{sec:sampling_pres_proof}.

    We now state the concentration result and defer its proof to \Cref{sec:concentration_appendix}. 
    Observe that the lemma is in fact applicable to all distributions with bounded $(t^2 \log d)$ moments in the directions of the standard basis, not only subexponential distributions. 
    
\begin{restatable}{lemma}{BadicLinfConscFull}\label{lem:basic_linf_consc-full}
	Let $D$ be a distribution over $\R^d$ with mean $\mu$. Suppose that for all $s \in [1,\infty)$, $D$ has its $s^{th}$ moment bounded by $(f(s))^s$ for some non-decreasing function $f:[1,\infty) \to \R_+$, in the direction $e_j$, i.e., suppose that for all $j \in [d]$ and $X \sim D$:
	$		\|\iprod{e_j, X - \mu}\|_{L_s} \leq f(s)$.
	Let $X_1, \dots, X_m$ be $m$ i.i.d.\ samples from $D$ and define $\overline{\mu}:= \sum_{i=1}^m X_i$. The following are true:
	\begin{enumerate}
	\item If $m \geq \max\left(\frac{1}{\delta^2}, 1\right) C  \left( t \log (d/\gamma)\right)\left(2f(t^2\log(d/\gamma))\right)^{2t} \max\left( 1, \frac{1}{f(t)^{2t}}  \right)$, then with probability $1 -\gamma$, we have that 
\begin{align*}
\left \| \E_{i \sim [m]}[(X_i - \overline{\mu})^{\otimes t}] -  \E_{X \sim D}[(X - \mu)^{\otimes t}] \right \|_{\infty} \leq \delta \;.
\end{align*}

	\item If
		$m  > C (k/ \delta^2)  \log(d/\gamma) ( f(\log(d/\gamma))  )^2 $,
		 then with probability $1 - \gamma$, we have that 
			$\left\| \overline{\mu}- \mu \right\|_{2,k} \leq \delta$.

		\end{enumerate}
\end{restatable}

\subsubsection{Proof of \Cref{lem:pop-to-empirical}}\label{sec:sampling_pres_proof}
We now formalize the above sketch.

\begin{proof}
Suppose for now that with $m$ samples, 
	the 
	$\ell_{\infty}$ norm of the difference between the expected and empirical $t$-th tensors of $D$ is $M/\sqrt{k^t}$, i.e.,
	\[\left \| \E_{i \sim [m]}[(X_i - \ovl \mu)^{\otimes t}] -  \E_{X \sim D}[(X - \mu)^{\otimes t}] \right \|_{\infty}\leq \frac{M}{\sqrt{k^t}}.\] 
	Let $p(v_1, \dots, v_d) := \sum_{ T \in [d]^{t}} (\E_{i \sim [m]}[X_i - \ovl \mu]_T - \E_{X \sim D} [X - \mu ]_T)  v_T$. An easy corollary of \Cref{claim:sparse-upper-bound} is its application to $p(v_1, \dots, v_d)$. Combining these two steps we have that: 
	\begin{align}
		\cAksparse \sststile{2t}{v, z}& \left(\E_{i \sim [m]} \Brac{\iprod{ v,  X_i - \overline{\mu}}^{t}} - \E_{X \sim D}\Brac{\iprod{v, X - \mu}^t} \right)^2 \notag \\
		&\leq k^{t}  \left \| \E_{i \sim [m]}[(X_i - \ovl \mu)^{\otimes t}] -  \E_{X \sim D}[(X - \mu)^{\otimes t}] \right \|_{\infty}^2 \leq M^2 \;.		\label{eqn:moment_bound_1}
	\end{align}
	To prove bounded central moments of the uniform distribution over the samples, observe that, 
	\begin{align*}
		\cAksparse &\sststile{2t}{v, z}\E_{i \sim [m]} \Brac{\iprod{ v,  X_i - \overline{\mu}}^{t}}^2 \\
		&= \Paren{\E_{i \sim [m]} \Brac{\iprod{ v,  X_i - \overline{\mu}}^{t}} - \E_{X \sim D}\Brac{\iprod{v, X - \mu}^t}  + \E_{X \sim D}\Brac{\iprod{v, X - \mu}^t} }^2\\
		&\leq 2^2 ~ \Paren{\E_{i \sim [m]} \Brac{\iprod{ v,  X_i - \overline{\mu}}^{t}} - \E_{X \sim D}\Brac{\iprod{v, X - \mu }^t}}^2 + 2^2 ~ \E_{X \sim D}\Brac{\iprod{v, X - \mu}^t}^2 \\
		&\leq 4 \Paren{ M^2 + \E_{X \sim D}\Brac{\iprod{v, X - \mu}^t}^2} \leq 8M^2 \;,
	\end{align*}
	where the third line uses SoS triangle inequality (\Cref{fact:sos-triangle}), the fourth line uses \Cref{eqn:moment_bound_1} and the last one uses our assumption that $D$ has certifiably bounded moments. 
	
	We now calculate the sample complexity for the first claim.
	Since the distribution is sub-exponential, we have that for $Y = (Y_1,\dots,Y_d)  \sim D$, $\|Y\|_{L_s} \leq c s$, i.e., $f(x) = O(x)$.
	\Cref{lem:basic_linf_consc-full} with $\delta = M/\sqrt{k^t}$, $f(x) = O(x)$, and $\gamma=0.1$ implies that the sample complexity is at most the following:
	\begin{align*}
	    C\max\left(1, \frac{k^t}{M^2}\right) (t\log(d/\gamma))(ct^2 \log(d/\gamma))^{2t}\max\left(1, \frac{1}{(ct)^{2t}}\right) \lesssim  (kt\log(d/\gamma))^{5t}\max(1,M^{-2}).
	\end{align*}
	
	For the second claim we use  Part 2 of \Cref{lem:basic_linf_consc-full} with $f(x)=O(x)$, $\delta=M^{1/t}\eps^{1-1/t}$, and $\gamma=0.1$, which means that the required number of samples is at most $C kM^{-2/t}\eps^{-2+2/t} (\log(d/\gamma))^{3}$. 
	Finally, we note that $\max(M^{-2},M^{-2/t},1) = \max(M^{-2},1)$.
\end{proof}

%% file: poincare-mean-est.tex
\newcommand*{\ak}{\cAksparse}
\newcommand*{\gax}{\cA_{\textnormal {G-sparse-mean-est}}}

\section{Robust Sparse Mean Estimation with Unknown Covariance}
\label[section]{sec:sparse-mean-est}

Given that the inliers have certifiably bounded moments in $k$-sparse directions (which happens with high probability because of \Cref{lem:pop-to-empirical}), we show that our SoS algorithm finds a vector that is within $O(M^{1/{t}} \eps^{1-1/{t}} )$ of the empirical mean of the inliers. In this section, we show the following theorem, which when combined with \Cref{lem:pop-to-empirical} shows \Cref{thm:main-poincare-informal}.

\begin{restatable}{theorem}{SosMainThmBddMoments}\label[theorem]{thm:main_sparse_mean_estimation} 
Let $t\in \N$ be a power of 2 and $\eps \leq \eps_0$ for a sufficiently small constant $\eps_0$. Let $X_1,\ldots, X_m \in \R^d$ be such that the uniform distribution on $\{X_1, \dots, X_m\}$ has $(M,t)$ certifiably bounded moments in $k$-sparse directions (see \Cref{def:bounded-moments-k-sparse}). Given $\eps,k,M,t $ and any $\eps$-corruption of $X_1,\ldots, X_m$, \Cref{alg:sparse_recovery} runs for time\footnote{We will assume that the bit complexity of the input and the proof of $(M,t)$ certifiably bounded moments is at most $\poly(m^t,d^t)$.}
 $(md)^{O(t^2)}$ and returns a vector $\widehat{\mu}$ with $\|\widehat{\mu} - \E_{i \sim [m]}[X_i] \|_{2,k} = O(M^{1/{t}} \eps^{1-1/{t}} )$. 
\end{restatable}

\paragraph{Additional Notation} To avoid confusion, we fix the following notation for the rest of the paper. We use $X_1, \ldots, X_m$ to denote the inlier points. Their empirical mean and covariance is denoted by $\overline{\mu}$ and $\overline{\Sigma}$ respectively. The points $Y_1,\ldots, Y_m$ are the $\eps$-corrupted set of samples. We use $X_1',\ldots, X_m'$ to denote vector-valued variables of length $d$ for the SoS program and $\mu',\Sigma'$ to denote their empirical mean and covariance. Finally,  $w_1,\ldots,w_m$ will be scalar-valued variables of the SoS program. 
\newline

Our algorithm is based on the system of polynomial inequalities defined in \Cref{def:axioms} below, which capture the following properties of the uncorrupted samples: (i) $X_i' = Y_i$ for all but $\e m$ indices, and (ii) The $t$-th moment of the uniform distribution on $\{ X_i' \}_{i=1}^m$ is certifiably bounded in every $k$-sparse direction. Although the last constraint seems complicated,  we show in \Cref{app:quantifier_elim} that it can be expressed as $d^{O(t)}$ polynomial constraints, in an additional $\poly((md)^t)$ variables. 
Finally, our algorithm \textsc{Sparse-mean-est} will solve a semidefinite programming (SDP) relaxation of the polynomial system $\cA_{\textnormal {sparse-mean-est}}$. 

\begin{definition}[Sparse Mean Estimation Axioms $\cA_{\textnormal {sparse-mean-est}}$]
	\label[definition]{def:axioms}
	Let $Y_1,\ldots,Y_m \in \R^d$. 
	Let $t \in \N$ be even and let $\delta, \epsilon > 0$. $\cA_{\textnormal{sparse-mean-est}}$ denotes the system of the following constraints. 
	\begin{enumerate}
	   \itemsep0.1em 
		\item Let $\mu' = \tfrac 1 m \sum_{i=1}^m  X_i'$.
		\item Let $\cA_{\textnormal{corruptions}} := \{ w_i^2 = w_i \}_{i \in [m]} \cup \{ w_i (Y_i - X_i') = 0 \}_{i \in [m]} \cup  \{ \sum_{i \in [m]} w_i = (1-\e)m\}$.
		\item \label{item:cons} $X'_1, \dots, X'_m$ satisfy $(M, t)$ certifiably bounded moments in $k$-sparse directions (\Cref{def:bounded-moments-k-sparse}). 
	\end{enumerate}
\end{definition}

\begin{algorithm}[htb]{}
	\begin{algorithmic}[1]
		\Function{Gaussian-Sparse-mean-est}{$Y_1,\ldots,Y_m,t,M,\epsilon, k$}
		\State Find a pseudo-expectation $\pE$ of degree $10t$ which satisfies the system of \Cref{def:axioms}. 
		\State \Return  $\hat{\mu} := \pE \Brac{\mu'}$.
		\EndFunction
	\end{algorithmic}
	\caption{Robust Sparse Mean Estimation}
	\label{alg:sparse_recovery}
\end{algorithm}
\noindent

\subsection{Proof of  \Cref{thm:main_sparse_mean_estimation}}\label{sec:proof_of_main_theorem}

 We first show that the system given in \Cref{def:axioms} is feasible: Observe that the following assignments satisfy the constraints: $X_i' = X_i$,  $w_i = \mathbf{1}_{(Y_i = X_i)}$, $\mu' = \frac{1}{m} \sum_i X_i$.
It is easy to check that the first two constraints are satisfied. The fact that the final constraint is satisfied follows from the assumption of the theorem and \Cref{fact:quant_el_fact}.

In what follows, we assume that $v$ is a fixed sparse vector.
The proof consists of first showing that $\iprod{v,\ovl \mu - \mu'}^{2t} \leq O(M^2\eps^{2t-2})$ %
has an SoS proof and then showing that $\pE[\mu']$ also satisfies the same inequality as $\mu'$ does. 

We start with the first step.
The program variables $X_i'$ have constraints which ensure that a $(1-\eps)$ fraction of these will match the data, $Y_i$. The following standard claim (shown in \Cref{sec:omittedsec4}) shows that the program variables match a $(1-2\eps)$ fraction of the \emph{uncorrupted} samples $X_i$. 	Note that in the claim below the $r_i$ are constants, even though they are not known to the algorithm. 

	\begin{restatable}{claim}{ClaimSimpleCheck}
\label{claim:simplecheck}
		Let $r_i := \mathbf{1}_{X_i = Y_i}$ and $W_i:=w_ir_i$ and $b$ be the bit complexity of $S = \{X_1,\ldots, X_m\}$, then there exists an SoS proof of  $\Set{ W_i^2 = W_i}_{i=1}^m \cup \Set{ \sum_{i=1}^m (1-W_i) \leq 2\epsilon m} \cup \Set{ W_i \Paren{X_i - X'_i} = 0}_{i=1}^m $ from the axioms $\{ W_i = w_i r_i\}_{i=1}^m \cup \cA_{\textnormal{corruptions}}$ of bit complexity at most $\poly(m, d, b)$. 
		\label{claim:indicator_w}
	\end{restatable}
	\noindent We now work towards an upper bound on $\iprod{v, \overline{\mu} - \mu'}^{2t}$. Let $r_i := \mathbf{1}_{X_i = Y_i}$ and $W_i:=w_i r_i$ as above. We first show that there is an SoS proof
	 for $\iprod{v, \overline{\mu} - \mu'}^{2t} \leq (2\epsilon)^{2t-2} ~ \E_{i \sim [m]} \left[\iprod{v, X_i - X'_i}^{t}\right]^2$:
	\begin{align}
		\cA_{\textnormal {sparse-mean-est}} \sststile{O(t)}{} \iprod{v, \overline{\mu} - \mu'}^{2t} &=  \Big(\E_{i \sim [m]} \left[(1-W_i) \iprod{v, X_i - X_i'} \right]\Big)^{2t} \notag \\ 
		&\leq \Big( \Big(\E_{i \sim [m]} [1-W_i] \Big)^{t-1} \E_{i \sim [m]} \left[\iprod{v, X_i - X'_i}^{t}\right]\Big)^2  \notag\\
		&\leq (2\epsilon)^{2t-2} ~ \E_{i \sim [m]} \left[\iprod{v, X_i - X'_i}^{t}\right]^2 \;, %
	\end{align}
	where we used SoS H\"older (\Cref{fact:sos-holder}) and \Cref{claim:indicator_w}.
	Now, to bound $\E_{i \sim [m]} \left[\iprod{v, X_i - X'_i}^{t}\right]^2$, first observe that $\iprod{v, X_i - X'_i} = \iprod{v, X_i - \overline{\mu}} + \iprod{v, \overline{\mu} - \mu'} + \iprod{v, \mu' - X_i'}$. Applying SoS triangle inequality (\Cref{fact:sos-triangle}) twice, we see that there is an $O(t)$-degree SoS proof of the following
	\begin{align*} 
		\E_{i \sim [m]} \left[\iprod{v, X_i - X'_i}^{t}\right]^2\leq 3^{2t+2} \left(  \E_{i \sim [m]} \left[ \iprod{v, X_i - \overline{\mu}}^{t} \right]^2  + \iprod{v, \overline{\mu} - \mu'}^{2t}  + \E_{i \sim [m]} \left[ \iprod{v, \mu' - X'_i}^{t} \right]^2  \right).
	\end{align*} 
	The first and the last term above can be bounded by $M^2$.  The bound on the first term follows from the assumption that the moments of the uncorrupted dataset are bounded, and the bound on the second term follows from constraint \eqref{item:cons} of our program (\Cref{def:axioms}) as in the \Cref{fact:cons_sos_proof}.

	\noindent Putting these together, thus far we have shown that
	\begin{align*}
		\cA_{\textnormal {sparse-mean-est}} \sststile{O(t)}{} \iprod{v, \overline{\mu} - \mu'}^{2t} &\leq (2\eps)^{2t-2} \cdot 3^{2t + 2} \Paren{2M^2 + \iprod{v, \overline{\mu} - \mu'}^{2t} } \\
		&\leq 6^{2t + 2} \cdot \eps^{2t - 2} \cdot \Paren{M^2 + \iprod{v, \overline{\mu} - \mu'}^{2t}}.
	\end{align*}
	Rearranging and using the assumption that $\eps<3/1000$  implies $6^{2t + 2} \cdot \epsilon^{2t-2} \leq 1/2$, we get that
	\begin{align} \label{eq:what_we_got}
		\cA_{\textnormal {sparse-mean-est}} &\sststile{O(t)}{} \iprod{v, \overline{\mu} - \mu'}^{2t} \leq \eps^{2t-2} \cdot \frac{6^{2t + 2} \cdot M^2}{1- 6^{2t + 2} \cdot \epsilon^{2t-2}} \leq 6^{2t+3} M^2 ~\epsilon^{2t-2}.
	\end{align}

\noindent Finally, taking pseudoexpectations on both sides of \Cref{eq:what_we_got} and  using \Cref{fact:pseudo-expectation-cauchy-schwarz} (pseudoexpectation Cauchy-Schwartz), we see that $\iprod{v, \overline{\mu} - \pE [\mu'] } \leq O(M^{1/{t}} \eps^{1-1/t} )$ for all $k$-sparse unit vectors, or equivalently $\| \overline{\mu} -\pE \mu'\|_{2,k} = O(M^{1/{t}} \eps^{1-1/t} )$. This completes the proof of \Cref{thm:main_sparse_mean_estimation}.

To see that the runtime is $(md)^{O(t^2)}$, note that the algorithm looks for a  pseudoexpectation  which has degree $O(t)$ in the variables $\{w_1, \dots, w_m\} \cup \{X'_1, \dots, X'_m, 
\mu'\}$ and the $d^{O(t)}$ variables that come from (3) in \Cref{def:axioms}.

%% file: gaussian-mean-est.tex
\section{Achieving Near-optimal Error for Gaussian Inliers}
\label{sec:Gaussian-O-eps}

In this section, we show a $(k^4/\eps^2)\polylog(d/\eps)$ sample, polynomial time algorithm to estimate the mean of a multivariate Gaussian distribution in $k$-sparse directions. 
Our starting point is the recent work of \cite{kothari2021polynomial}, who established a similar result in the dense setting. 
While they are able to handle ill-conditioned covariances, it is not possible to directly use their result for the sparse setting because we cannot use off-the-shelf multiplicative concentration inequalities. Instead, we start with assuming a well-conditioned covariance to prove the theorem below, and then use Lepskii's method in \Cref{app:lepski} to get \Cref{thm:main-gaussian-informal} for arbitrary covariances.

\begin{theorem}\label{thm:gaussian_well_cond}
Let $k,d \in \Z_+$ with $k\leq d$ and $\eps<\eps_0$ for a sufficiently small constant $\eps_0$. Let $\mu \in \mathbb{R}^d$ and $\Sigma \in \R^{d \times d}$ such that $I_d \preceq \Sigma \preceq 2 I_d$.
Let $m > C({k^4}/{\eps^2}) \log^5(d/(\gamma\eps))$ for a large enough constant $C>0$.
There exists an algorithm which, given $\eps$,$k$, and an $\eps$-corrupted set of samples from $\cN(\mu,\Sigma)$ of size $m$, returns an estimate $\widehat{\mu}$ in time $\poly(md)$ such that, with probability $1-\gamma$, $\widehat{\mu}$ satisfies $\| \widehat{\mu} - \mu \|_{2, k} \leq \tilde{O}(\eps)  $.
\end{theorem}

An important component of the algorithm for the sparse setting is the SoS program given by \Cref{def:gaxioms}. Our notation is as before, with the addition of $\Sigma'$, which is a $d \times d$ matrix-valued indeterminate. Additionally, define $\widehat{\mu} = \pE[\mu'], \widehat{\Sigma} = \pE[\Sigma'], \overline{\mu}=\E_{i\sim [m]}[X_i], \overline{\Sigma}=\E_{i\sim [m]}[(X_i-\overline{\mu})(X_i-\overline{\mu})^T], Y_{ij} = \frac{1}{2}(Y_i - Y_j)(Y_i - Y_j)^T, X_{ij} = \frac{1}{2}(X_i - X_j)(X_i - X_j)^T$.

\begin{definition}[Gaussian Sparse Mean Estimation Axioms $\gax$]
	\label[definition]{def:gaxioms}
	Let $Y_1,\ldots, Y_m \in \R^d$. 
	Let $0<\eps<1/2$. We define $\gax$ to be the following constraints.
	\begin{enumerate}
	    \itemsep0.1em 
		\item  $\mu' = \tfrac 1 m \sum_{i=1}^m  X_i'$ and $\Sigma' = \frac{1}{m}\sum_{i=1}^m (X_i'-\mu')(X_i'-\mu')^T$. 
		\item  $\cA_{\textnormal{corruptions}} := \{ w_i^2 = w_i \}_{i \in [m]} \cup \{ w_i (Y_i - X_i') = 0 \}_{i \in [m]} \cup  \{ \sum_{i \in [m]} w_i = (1-\e)m\}$. 
		\item $\ak \sststile{8}{}\left(  \E_{i \sim [m]} \Brac{\iprod{v, X_i'-\mu'}^4}- 3(v^T \Sigma' v)^2   \right)^2 \leq \tilde{O}(\eps^2)  (v^T \Sigma' v)^4$.
		\item $\ak \sststile{2}{} (  v^T \Sigma' v  )^2 \leq 9$.

	\end{enumerate}
\end{definition}

More simply, $\gax$ consists of constraints that capture the following: (1) $X_i' = Y_i$ for all but $\e m$ indices; and (2) The fourth moment of the uniform distribution on $\{ X_i' \}_{i \in [m]}$ is bounded in $k$-sparse directions. The algorithm (\Cref{alg:g_sparse_recovery}) consists of finding a degree-12 pseudo-expectation that satisfies $\gax$, and estimates the sparse mean up to an error of  $\tilde{O}(\eps)$. 

\begin{algorithm}[htb]{}
	\begin{algorithmic}[1]
		\Function{Sparse-mean-est}{$Y_1,\ldots,Y_m,\epsilon, k$}
		\State Find a pseudo-expectation $\pE$ of degree-$12$  that satisfies the program of \Cref{def:gaxioms}. 
		\State Let $\widehat{\mu} = \tilde{\E}[\mu']$ \label{line:hatmu} and 
		output $\widehat{\mu}$. 
		\label{line:final_output}
		\EndFunction
	\end{algorithmic}
	\caption{Robust Sparse Mean Estimation}
	\label{alg:g_sparse_recovery}
\end{algorithm}
\noindent

\subsection{Deterministic Conditions for Inliers} \label{sec:resilience_main}

We require a set of deterministic conditions similar to that in \cite{kothari2021polynomial}. However, instead of proving the result for \emph{all} directions, we will instead require this only for \emph{$k$-sparse} directions. 
We show that, with high probability, a set of $(k^4/\eps^2)\polylog(d/\eps)$ samples drawn from $\cN(\mu, \Sigma)$ satisfy the following set of conditions.

\begin{restatable}{lemma}{LemResilienceIntegrated} \label{lem:resilience_integrated} 
Let $m > C (k^4/\eps^2)\log^5(d/(\eps\gamma))$ for a sufficiently large constant $C$. 
Let $X_1,\ldots, X_m \sim \cN(\mu,\Sigma)$ for $\mu \in \R^d$ and a positive definite matrix $I_d \preceq \Sigma \preceq 2 I_d$. 
	Let $T$ denote the set of all $a \in [0,1]^{m \times m}$ and $a' \in [0,1]^{m}$ such that  (i)
 $a_{ij} = a_{ji}$ for all $i,j \in [m]$, (ii)
		 $\E_{ij} [a_{ij}] \geq 1-4\eps$, and (iii)
		 $\E_{j} [a_{ij}] \geq a'_i(1-2\eps)$ for all $i \in [m]$ and $a_{ij}\leq a'_i$ for all $i,j \in [m]$.
		 Denote $X_{ij}:=(1/2)(X_i-X_j)(X_i-X_j)^T$ and $\overline{\Sigma}:=\E_{ij}[X_{ij}]$.
		 With probability $1-\gamma$, the following holds for all $v \in \cU_k$:
	\begin{enumerate}
    	\itemsep0.1em 
		\item $| \iprod{v,\overline{\mu}-\mu}| \leq \tilde{O}(\eps)\sqrt{v^T \Sigma v}$. \label{it:res1}
		\item $| \E_{i \sim [m]}[a'_i\iprod{v,X_i-\overline{\mu}}]   | \leq \tilde{O}(\eps) \sqrt{v^T \overline{\Sigma} v}$.  \label{it:res2}
		\item $\left| \E_{i \sim [m]}\left[a'_i \left(\iprod{v,X_i-\overline{\mu}}^2    -  v^T \overline{\Sigma} v \right)  \right] \right| \leq \tilde{O}(\eps) v^T \overline{\Sigma} v $.  \label{it:res3}
		\item $|v^T(\overline{\Sigma} - \Sigma)v  | \leq \tilde{O}(\eps) v^T \Sigma v$. \label{it:res4}	\item $|\E_{i,j \sim [m]}[a_{ij}(v^TX_{ij}v - v^T\overline{\Sigma}v)  ] | \leq \tilde{O}(\eps) v^T \overline{\Sigma} v$. \label{it:res6} 
		\item $|\E_{i,j \sim [m]}[a_{ij}((v^TX_{ij}v-v^T \overline{\Sigma} v)^2 -        2(v^T\overline{\Sigma} v)^2 )]  | \leq \tilde{O}(\eps) (v^T \overline{\Sigma} v)^2$. \label{it:res7}
	\end{enumerate}
\end{restatable}
The proof of this lemma is provided in \Cref{app:gaussian-resilience}. We also need another deterministic condition to ensure the feasibility of the program of \Cref{def:gaxioms}. For that we first need to argue that after taking enough samples, the empirical fourth moment of Gaussian is certifiably close to its population value. This is a consequence of the results of \Cref{sec:sparsity} and the assumption that $\Sigma \succeq I$.

\begin{restatable}{lemma}{LemCertifiableFourth}\label{lem:certifiable4th}
	Let $X_1,\ldots, X_m \sim \cN(\mu, \Sigma)$ for a $d \times d$ symmetric matrix $\Sigma$ satisfying $I_d \preceq \Sigma \preceq 2I_d$.
	Let $\overline{\mu}$ and $\overline{\Sigma}$ be the empirical mean and covariance of these samples, respectively. 
	If the number of samples is $
		m > C({k^4}/{\eps^2}) \log^5(d/(\epsilon\gamma))\;
	$
	for a sufficiently large constant $C$, then, with probability at least $1-\gamma$,
	we have that
	\begin{align*}
	\ak \sststile{8}{v, z} \left(  \E_{i \sim [m]} \Brac{\iprod{v, X_i-\overline{\mu}}^4}- 3(v^T \overline{\Sigma} v)^2   \right)^2 \leq O(\eps^2)  (v^T \overline{\Sigma} v)^4 \;.
	\end{align*}
\end{restatable}
\begin{proof}
	We have the following by the SoS triangle inequality (\Cref{fact:sos-triangle}):
	\begin{align*}
		&\ak \sststile{8}{v, z} \left(\E_{i \sim [m]} \Brac{\iprod{ v,  X_i - \overline{\mu}}^{4}} - 3(v^T \overline{\Sigma} v)^2  \right)^2\\
		&= \left(\E_{i \sim [m]} \Brac{\iprod{ v,  X_i - \overline{\mu}}^{4}}  {-} \E_{X \sim \cN(\mu,\Sigma)}\Brac{\iprod{ v,  X - \mu}^{4}} {+}  \E_{X \sim \cN(\mu,\Sigma)}\Brac{\iprod{ v,  X - \mu}^{4}} {-} 3(v^T \overline{\Sigma} v)^2 \right)^2 \\
		&\leq 4\left(\E_{i \sim [m]} \Brac{\iprod{ v,  X_i - \overline{\mu}}^{4}} - \E_{X \sim \cN(\mu,\Sigma)}\Brac{\iprod{v, X - \mu}^4} \right)^2 
		 +  4 \left(\E_{X \sim \cN(\mu,\Sigma)}\Brac{\iprod{ v,  X - \mu}^{4}} - 3(v^T \overline{\Sigma} v)^2  \right)^2. 
	\end{align*}
	We will upper bound each of the two terms above separately.  Focusing on the first term, we first define $\delta':= \eps$
	Then, similarly to \Cref{eqn:moment_bound_1},    we use  \Cref{claim:sparse-upper-bound} and \Cref{lem:basic_linf_consc-full}
	with $\delta=\delta'/k^2$ and $t=4$ to get that
	\begin{align*}
		\ak &\sststile{8}{v, z} \left(\E_{i \sim [m]} \Brac{\iprod{ v,  X_i - \overline{\mu}}^{4}} - \E_{X \sim \cN(\mu,\Sigma)}\Brac{\iprod{v, X - \mu}^4} \right)^2 \\
		&\leq k^{4}  \left \| \E_{i \sim [m]}[(X_i - \overline{\mu})^{\otimes 4}] -  \E_{X \sim \cN(\mu,\Sigma)}[(X - \mu)^{\otimes 4}] \right \|_{\infty}^2 \\
		&\leq k^4 \delta^2 \leq  (\delta')^2 \lesssim (\delta')^2 (v^T \overline{\Sigma} v)^4 = {O}(\eps^2) (v^T \overline{\Sigma} v)^4 \;,
	\end{align*}
	where in the last line we used $\Sigma  \succeq I_d$ combined with \Cref{it:res4} from \Cref{lem:resilience_integrated}.
	The sample complexity of $(k^4/\eps^2)\log^5(d/(\epsilon\gamma))$ comes from \Cref{lem:resilience_integrated} and  \Cref{lem:basic_linf_consc-full},
	with $f(s) \leq  \sqrt{C s}$ and $\delta=\eps/k^2$.
	
	We similarly bound the second term:
	\begin{align*}
		\ak &\sststile{8}{v, z} \left(\E_{X \sim \cN(\mu,\Sigma)}\Brac{\iprod{ v,  X - \mu}^{4}} - 3(v^T \overline{\Sigma} v)^2  \right)^2  \\
		&= \left(\E_{X \sim \cN(\mu,\Sigma)}\Brac{\iprod{ v,  X - \mu}^{4}}- \E_{X \sim \cN(\overline{\mu},\overline{\Sigma})}\Brac{\iprod{v, X - \overline{\mu}}^4} \right)^2 \\
&= \left(\E_{Y \sim \cN(0,\Sigma)}\Brac{\iprod{ v,  Y}^{4}}- \E_{Y \sim \cN(0,\overline{\Sigma})}\Brac{\iprod{v, Y}^4} \right)^2 \\
		&\leq k^{4}  \left \| \E_{Y \sim \cN(0,\Sigma)}[Y^{\otimes 4}] -  \E_{Y \sim \cN(0,\overline{\Sigma})}[Y^{\otimes 4}] \right \|_{\infty}^2\;,
	\end{align*} 
	where we used the specific form of Gaussian moments (\Cref{fact:gaussian-moments}) for the first equality. In order to bound all elements of the tensor, we use the following lemma which is shown in \Cref{sec:simple_tensor_concentration}.
	\begin{restatable}{lemma}{LemBoundSimilarTensor}\label{lem:similar-tensor-bound}
		Let $X_1, \ldots, X_m \sim \cN(\mu,\Sigma)$ where $I \preceq \Sigma \preceq 2I$,  and denote $\overline{\mu}=\E_{i \sim [m]}[X_i]$,   $\overline{\Sigma}=\E_{i \sim [m]}[(X_i - \overline{\mu})(X_i - \overline{\mu})^T]$. For any even integer $t$ and $\tau < 1$, if $m > C (1/\tau^2) t^{2t+1} 4^t \log(d/\gamma)$ for some absolute constant $C$, it holds 
		\begin{align*}
		\left\| \E_{Y \sim \cN(0,\Sigma)}[Y^{\otimes t}] - \E_{Y \sim \cN(0,\overline{\Sigma})}[Y^{\otimes t}]  \right\|_{\infty} \leq \tau \;,
		\end{align*}
		with probability $1-\gamma$. 
	\end{restatable}
	Using the above with $t=4$ and $\tau = \delta'/k^2$ with $\delta'=\tilde{O}(\eps)$, we get that
	\begin{align*}
		\ak &\sststile{8}{v, z} \left(\E_{X \sim \cN(\mu,\Sigma)}\Brac{\iprod{ v,  X - \mu}^{4}} - 3(v^T \overline{\Sigma} v)^2  \right)^2 \leq (\delta')^2 \lesssim (\delta')^2 (v^T \overline{\Sigma} v)^4 = \tilde{O}(\eps^2) (v^T \overline{\Sigma} v)^4\;.
	\end{align*}
	This completes the proof of \Cref{lem:certifiable4th}. 
\end{proof}

As a corollary, we establish the feasibility of the system of  \Cref{def:gaxioms}. 

\begin{restatable}{corollary}{ProgramFeasibility} \label{cor:feasibility-check}
Under the conditions of \Cref{lem:resilience_integrated},
	$\gax$ in \Cref{def:gaxioms} is feasible with high probability.
\end{restatable}

\begin{proof}
	The pseudo-distribution that is defined to be the uniform distribution on inliers (i.e., $X_i' = X_i$) satisfies the constraints of the program. The first three conditions are trivially satisfied by choosing the $w_i$'s to be the indicators of whether the $i$-th sample is an inlier. The second last constraint is satisfied if and only if the inequality $\left(  \E_{i \sim [m]} \Brac{\iprod{v, X_i-\overline{\mu}}^4}- 3(v^T \overline{\Sigma} v)^2   \right)^2 \leq \tilde{O}(\eps^2)(v^T \overline{\Sigma} v)^4 $ has an SoS proof. By \Cref{lem:certifiable4th} we know that this is indeed the case.  
	
	We now focus on the last constraint. 
	We need to show  an SoS proof of $(v^T \ovl \Sigma v)^2 < 9$. We will show an SoS proof of $(v^T \overline \Sigma v - v^T \Sigma v)^2 \leq O(\eps^2)$. By techniques similar to the ones used in \Cref{lem:certifiable4th},
	we see that 
\[ \ak \sststile{}{} (v^T \overline \Sigma v - v^T \Sigma v)^2 \leq k^2 \norm{\overline \Sigma - \Sigma}_\infty. \]
Since $m > C({k^4}/{\eps^2}) \log^5(d/\gamma)$ for large enough constant $C$, we have the following with high probability:
\[ \ak \sststile{}{} (v^T \overline \Sigma v - v^T \Sigma v)^2 \leq k^2 \norm{\overline \Sigma - \Sigma}_\infty \leq O(\eps^2) \;.\]
Finally, to get an upper bound on $(v^T \ovl \Sigma v)^2$ we apply the SoS triangle inequality, as shown below
\begin{align*}
\ak \sststile{}{} (v^T \overline \Sigma v)^2 &=  (v^T \overline \Sigma v - v^T \Sigma v + v^T \Sigma v)^2 \\
&\leq 2(v^T \Sigma v)^2 + 2(v^T \overline \Sigma v - v^T \Sigma v)^2 \leq 8 + 2 O(\eps^2) \leq 8 + O(\eps^2) \leq 9,
\end{align*}
where we use that  $\Sigma \preceq 2I$ and $\eps$ is chosen to be small enough. 
\end{proof}

\subsection{Proof of \Cref{thm:gaussian_well_cond}} \label{sec:proof_of_main_thm_right_error}

	In this section we prove  \Cref{thm:gaussian_well_cond}, deferring proofs of intermediate lemmata to \Cref{sec:estimation_lemmata}. As explained in \Cref{sec:techniques_gaussian}, the assumption $I \preceq \Sigma \preceq 2I$ is removed in \Cref{app:lepski}, where we finally prove \Cref{thm:main-gaussian-informal}.

	Given $m>C(k^4/\eps^2)\log^5(d/(\eps\gamma))$ samples, the conclusions of  
	\Cref{lem:resilience_integrated,lem:certifiable4th} are true. 
	Further, by \Cref{cor:feasibility-check}, we know that the program is feasible.
	The first step is to show that our theorem holds given that $\tilde{\E}[\Sigma']$ is a good enough approximation of $\Sigma$. 
	
	\begin{restatable}{lemma}{SpecializedMean} \label{lem:specialized_mean}
		Let $Y_1,\ldots, Y_m$ be an $\eps$-corruption of the set $X_1,\ldots, X_m$, satisfying \Cref{it:res2,it:res3} of \Cref{lem:resilience_integrated}. Let $\pE$ be a degree-6 pseudo-expectation in variables $w_i,X_i',\Sigma',\mu'$ satisfying the system of \Cref{def:gaxioms}. Denote by $\overline{\mu},\overline{\Sigma}$  the empirical mean and covariance of $X_1,\ldots, X_m$ and let $\widehat{\Sigma}:=\pE[\Sigma']$. Then, for all $v \in \cU_k$ it holds 
\begin{align*}
|\iprod{v,\widehat{\mu} - \overline{\mu}}| \leq \tilde{O}(\eps) \sqrt{v^T \overline{\Sigma} v} + \sqrt{O(\eps) v^T(\widehat{\Sigma} - \overline{\Sigma})v + \tilde{O}(\eps^2) v^T(\widehat{\Sigma} + \overline{\Sigma})v} \;.
\end{align*}		
	\end{restatable}
	
	It now suffices to show that $|v^T(\widehat{\Sigma}-\overline{\Sigma})v| \leq \tilde{O}(\eps) v^T \overline{\Sigma} v$  since  \Cref{lem:specialized_mean} combined with \Cref{it:res1,it:res4} of \Cref{lem:resilience_integrated} implies that  $| \iprod{v, \widehat{\mu}-\mu } | \leq \tilde{O}(\eps) \sqrt{v^T \Sigma v} \leq \tilde{O}(\eps)$ and thus proves our main theorem. 
	Thus, we focus on showing that $|v^T(\widehat{\Sigma}-\overline{\Sigma})v| \leq \tilde{O}(\eps) v^T \Sigma v$ for all $v \in \cU_k$.

\begin{restatable}{lemma}{SpecializedCov} \label{lem:specialized_cov}
Let $Y_1,\ldots, Y_m$ be an $\eps$-corruption of $X_1,\ldots, X_m$ 
satisfying \Cref{it:res6,it:res7} of \Cref{lem:resilience_integrated}. 
Let $\pE$ be a degree-$12$ pseudo-expectation in variables 
$w_i,X_i',\Sigma',\mu'$ satisfying the system of \Cref{def:gaxioms}. 
Define $Y_{ij}=(1/2)(Y_i-Y_j) (Y_i-Y_j)^T$, 
$X_{ij}=(1/2)(X_i-X_j) (X_i-X_j)^T$, $X'_{ij}=(1/2)(X'_i-X'_j) (X'_i-X'_j)^T$, 
$\widehat{\Sigma}=\pE[\Sigma']$, $w_{ij}' = w_{i}w_j\mathbf{1}(X_{ij} = Y_{ij})$, 
and $R = \pE [\E_{ij}[(1-w_{ij}')v^T(X_{ij}'-\overline{\Sigma})v]^2]$. 
Then, for every $v \in \cU_k$, we have that, 
\begin{enumerate}
\item $|v^T(\widehat{\Sigma}-\overline{\Sigma})v  | \leq \tilde{O}(\eps) v^T \overline{\Sigma} v +  \sqrt{R}$ and
\item $R \leq O(\eps) ( \pE[(v^T\Sigma'v)^2] - (v^T\overline{\Sigma}v)^2  )  + \tilde{O}(\eps)  ( \pE[(v^T\Sigma'v)^2] + (v^T\overline{\Sigma}v)^2  )$.
\end{enumerate}
	\end{restatable}
	
	The final part of the proof is identical to \cite{kothari2021polynomial} and is provided in \Cref{sec:last_part_of_proof} for completeness. It consists of showing that $R = \tilde{O}(\eps^2) (v^T \overline{\Sigma} v)^2$.

\subsection{Making the Error Scale with $\sqrt{\| \Sigma \|_2}$} \label{app:lepski}

In this section we complete the proof of \Cref{thm:main-gaussian-informal}, which we restate below:

\begin{theorem}\label{thm:main-gaussian-formal}
Let $k,d \in \Z_+$ with $k\leq d$ and $\eps<\eps_0$ for a sufficiently small constant $\eps_0>0$. 
Let $\mu \in \R^d$ and $\Sigma \in \R^{d \times d}$ be a  positive semidefinite matrix. 
There exists an algorithm which, given $\eps$,$k$, and an $\eps$-corrupted set of samples from $\cN(\mu,\Sigma)$ 
of size $m = O(({k^4}/{\eps^2}) \log^5(d/(\eps)))$, runs in time $\poly(md)$,
and returns an estimate $\widehat{\mu}$ such that 
$\|\widehat{\mu} - \mu \|_{2, k} \leq \tilde{O}(\eps) \, \sqrt{\|\Sigma\|_2}$ with high probability.
\end{theorem}

Thus far, we have obtained an estimator that is $\tilde{O}(\eps)$-accurate given samples from $\cN(\mu,\Sigma)$ with $I_d \preceq \Sigma \preceq 2I_d$. Note that the assumption $I_d \preceq \Sigma$ can trivially be removed by having a pre-processing step that adds a zero-mean identity covariance Gaussian noise to all samples (since a zero-mean noise does not affect the mean). However, when $\Sigma \preceq \sigma^2 I_d$ with $\sigma$ much smaller than $1$, the optimal error rate is $\sigma \tilde{O}( \eps)$, which is much better than $\tilde{O}(\eps)$. If $\sigma$ is known to the algorithm in advance, the simple normalization step that is shown in \Cref{alg:black_box_estimator} with $\tilde{\sigma}=\sigma$  is enough to yield the desired error of $\tilde \sigma \tilde{O}( \eps)$. In other words,  we have so far obtained an estimator $\mathrm{RobustMean}(S,\tilde{\sigma},\eps,k)$ that is guaranteed to return a vector within $\tilde \sigma \tilde{O}(\eps)$ from the true mean with probability $1-\gamma$ (given that the number of samples is as specified in \Cref{thm:gaussian_well_cond}), so long as $\tilde{\sigma} \geq \sigma$.

\begin{algorithm}[h!]  
	\caption{Improved estimator when $\sigma$ is known.} 
	\label{alg:black_box_estimator}
	\begin{algorithmic}[1] 
		\Statex  
		\Function{RobustMean}{$S=\{x_1,\ldots, x_m\},\tilde{\sigma},\eps, k$}
		\State Let $e_1,\ldots, e_m \sim \cN(0,I_d)$.
		\State Let $\tilde{S}= \{ x_i/\tilde{\sigma} + e_i \; : \; i \in [m]\}$.
		\State $\tilde{\mu} \gets \mathrm{\textsc{Gaussian-Sparse-mean-est}}(\tilde{S},\eps,k$). \Comment{\Cref{alg:g_sparse_recovery}}\\
		\Return $\tilde{\sigma} \tilde{\mu}$
		\EndFunction
	\end{algorithmic}  
\end{algorithm}

\Cref{thm:lepski}, known as Lepskii's method~\cite{lepskii1991problem,birge2001alternative}, states that even in the case where the only known bounds for $\sigma$ are $\sigma \in [A,B]$ for some $A,B$, a near-optimal error can still be achieved by running  $\mathrm{RobustMean}(S,\tilde{\sigma},\eps,k)$ below.

\begin{theorem}\label{thm:lepski}
	Let $\mu \in \R^d$, $A,B>0$, $\sigma\in [A,B]$, and a non-decreasing function $r:\R^+ \to \R^+$. Suppose $\mathrm{Alg}(\tilde{\sigma},\gamma')$ %
	is a black-box algorithm which is guaranteed to return a vector $\widehat{\mu}$ such that $\|\widehat{\mu}-\mu\|_2\leq r(\tilde{\sigma})$, with probability at least $1-\gamma'$, whenever $\tilde{\sigma} \geq \sigma$. Then, \Cref{alg:lepski}, returns $\widehat{\mu}^{(\widehat{J})}$ such that, with probability at least $1-\gamma$, it holds $\|\widehat{\mu}^{(\widehat{J})} - \mu  \|_2 \leq 3 r( 2\sigma)$. Moreover, \Cref{alg:lepski} calls $\mathrm{Alg}$ at most $O(\log(B/A))$ times.
\end{theorem}
\begin{proof}
 For $j=0,1,\ldots, \log(B/A)$, denote by $\cE_j$ the event that $\| \widehat{\mu}^{(j)} - \mu \|_2 \leq r( \tilde{\sigma}_j)$. Let  $J$ be the index corresponding to the value of the unknown parameter $\sigma$, i.e., $\tilde{\sigma}_{J+1} \leq \sigma \leq \tilde{\sigma}_{J}$.
	Conditioned on the event $\cap_{j=0}^J \cE_j$, we have that $\| \widehat{\mu}^{(j)} - \mu \|_2 \leq r(\tilde{\sigma}_j)$ for all $j=0,1,\ldots, J$. 
	Using the triangle inequality, this gives that $\|\widehat{\mu}^{(J)} - \widehat{\mu}^{(j)} \|_2 \leq r( \tilde{\sigma}_J ) + r(\tilde{\sigma}_j)$. 
	This means that the stopping condition of the while loop in \Cref{alg:lepski} is satisfied during round $J$ and thus, if $\widehat{\mu}^{(\widehat{J})}$ denotes the vector returned by the algorithm, we have that $\widehat{J}\geq J$ and 
	\begin{align*}
		\|\widehat{\mu}^{(\widehat{J})} - \widehat{\mu}^{(J)} \|_2 \leq r( \tilde{\sigma}_{\widehat{J}}) + r(\tilde{\sigma}_J) \leq 2r(\tilde{\sigma}_J) \leq 2 r(2\sigma)\;,
	\end{align*}
	where the first inequality uses  the condition of the while loop, the second uses that $r$ is non-decreasing and $ \tilde{\sigma}_{\widehat{J}} \leq \tilde{\sigma}_{J}$, and the last one uses that $J$ was defined to be such that $\tilde{\sigma}_{J+1} \leq \sigma \leq \tilde{\sigma}_{J}$ so multiplying $\sigma$ by 2 makes it greater than $\tilde{\sigma}_{J}$. Using the triangle inequality once more, we get $\|\widehat{\mu}^{(\widehat{J})} - \mu  \|_2 \leq 3 r( 2\sigma)$. 
	Finally, by union bound on the events $\cE_j$, the probability of error is upper bounded by $\sum_{j=0}^{J} \gamma' \leq \gamma$. 
\end{proof}

\begin{algorithm}[h!]  
    \caption{Adaptive search for $\sigma$} 
    \label{alg:lepski}
    \begin{algorithmic}[1] 
      \Statex  
    \textbf{input:} $A,B,r(\cdot)$,$\gamma$ 
    \State{Denote $\tilde{\sigma}_j := B/2^j$ for $j=0,1,\ldots,\log(B/A)$ and set $\gamma':=\gamma/\log(B/A)$.   }
      \State $J \gets 0$
      \State $\widehat{\mu}^{(0)} \gets \mathrm{Alg}(\tilde{\sigma}_0,\gamma')$
      \While{ $\tilde{\sigma}_j \geq A$ and $\| \widehat{\mu}^{(J)} - \widehat{\mu}^{(j)} \|_2 \leq r(\tilde{\sigma}_J) + r(\tilde{\sigma}_j)$ for all $j =0,1,\ldots,J-1$}   
      \State $J \gets J+1$.
      \State $\widehat{\mu}^{(J)} \gets \mathrm{Alg}(\tilde{\sigma}_J,\gamma'$).
      \EndWhile
      \State $\widehat{J} \gets J-1$\\   
      \Return $\widehat{\mu}^{(\widehat{J})}$
    \end{algorithmic}  
  \end{algorithm}
  
In our setting, we use the following claim to get estimates for $A$ and $B$ such that $B/A$ is at most polynomial in $d$.

\begin{claim}\label{claim:rough_bounds}
    Let $S = \{Y_1,\dots,Y_m\}$ be an $\eps$-corrupted set from $\cN(\mu,\Sigma)$. Then we can obtain estimates $A$ and $B$ such that $B/A = \poly(d)$ and with probability $1 - \exp(-m)$, $\|\Sigma\|_2 \in [A,B]$.
\end{claim}

\begin{proof}
Suppose that $m$ is even
and define $m' := m/2$.   
Let $T = \{Z_1,\dots,Z_{m'}\}$, where  $Z_i = (Y_i - Y_{m' + i})/\sqrt{2}$.
Note that $T$ is an $2\eps$-corrupted set of $m'$ points from $\cN(0, \Sigma)$.
Let $X \sim \cN(0,\Sigma)$. 
We know that there exist constants $0 < c_1 < c_2$ such that $\pr(\|X\|_2^2 \in [c_1 \tr(\Sigma)/d, c_2{\tr(\Sigma)}] ) \geq 3/4$, which follows by anti-concentration of Gaussian and Markov inequality.
Thus a Chernoff bound implies that  with probability at least $1 - \exp(-cm)$, at least $60\%$ percent of the  points have squared norm lying in  $[c_1 {\tr(\Sigma)}/d, c_2{\tr(\Sigma)}]$.
Since $\eps < 0.1$, we have that with same probability, the empirical median of squared norms also lies in the same range.
Assume that this event holds for the remainder of the proof.
Let $D = \text{Median}_{z \in T}(\|z\|_2^2)$.
We have that $c_1 \|\Sigma\|_2/d \leq c_1 \tr(\Sigma)/d \leq D \leq c_2 \tr(\Sigma) \leq c_2 d \|\Sigma\|_2$.
Let $A = D/(c_2d)$ and $B = dD/c_1$.
\end{proof}

Putting everything together, we get our final theorem 
for Gaussian sparse mean estimation with unknown covariance.

\begin{proof}(Proof of \Cref{thm:main-gaussian-formal})
Let $S$ be an $\eps$-corrupted set from $\cN(\mu,\Sigma)$ of size $m$ as specified in the theorem. The algorithm is the following: We first obtain rough bounds $A,B$ for $\|\Sigma\|_2$ using the estimator of \Cref{claim:rough_bounds}. We then use the procedure of \Cref{alg:lepski} with $\mathrm{Alg}(\tilde{\sigma},\gamma)$ being the $\mathrm{RobustMean}(S,\tilde{\sigma},\eps,k)$ from \Cref{alg:black_box_estimator}, which is guaranteed to succeed with probability $1 - \gamma'$, where  $\gamma' = \gamma/(c'\log d)$ for a large enough constant $c'$. By \Cref{thm:gaussian_well_cond}, it suffices to use  $ C({k^4}/{\eps^2}) \log^5(d/(\gamma\eps))$ samples for a large constant $C$.
By \Cref{thm:gaussian_well_cond}, the black-box mean estimator $\mathrm{RobustMean}$ satisfies the guarantees required by \Cref{thm:lepski} with $\sigma= \sqrt{\|\Sigma\|_2}$, $r(\tilde{\sigma}) = \tilde{\sigma} \tilde{O}(\eps)$, and $A,B$ given by those found using the estimator of \Cref{claim:rough_bounds}. Therefore, the final estimate is that   \Cref{alg:lepski} has error $3r(2\sigma)= \sqrt{\|\Sigma\|_2} \tilde{O}(\eps)$ with probability at least $1- \gamma$.
Since Lepskii's method only calls the black-box estimator $\log(B/A) = O(\log(d))$ times, the computational complexity increases only by a logarithmic factor.
\end{proof}

%% file: SQ-lowerbound.tex
\section{Statistical Query Lower Bounds}
\label{sec:sq-lowerbd}

We begin by summarizing the necessary background  and then move to showing our results on Gaussians and distributions with bounded $t$-th moments in \Cref{sec:k4lowerbound,sec:t-bounded-lb} respectively. We refer the reader to \Cref{sec:low_degree} for the implications of the lower bounds of this section to hardness against low-degree polynomial tests.

\subsection{Background}

\subsubsection*{Statistical Query Lower Bounds Framework} We start with the basic definitions and facts from \cite{FGR+13,DKS17-sq} that we will use later. Although we are interested in proving hardness of estimation problems, we will focus on simpler hypothesis testing (or decision) problems. 

\begin{definition}[Decision Problem over Distributions] \label{def:decision}
	Let $D$ be a fixed distribution and $\D$ be a family of distributions.
	We denote by $\mathcal{B}(\D, D)$ the decision (or hypothesis testing) problem
	in which the input distribution $D'$ is promised to satisfy either
	(a) $D' = D$ or (b) $D' \in \D$, and the goal
	is to distinguish between the two cases.
\end{definition}

\begin{definition}[Pairwise Correlation] \label{def:pc}
	The pairwise correlation of two distributions with probability density functions
	$D_1, D_2 : \R^d \to \R_+$ with respect to a distribution with
	density $D: \R^d \to \R_+$, where the support of $D$ contains
	the supports of $D_1$ and $D_2$, is defined as
	$\chi_{D}(D_1, D_2) \eqdef \int_{\R^d} D_1(x) D_2(x)/D(x)\, \d x - 1$.
\end{definition}

\begin{definition} \label{def:uncor}
	We say that a set of $s$ distributions $\mathcal{D} = \{D_1, \ldots , D_s \}$
	over $\R^d$ is $(\gamma, \beta)$-correlated relative to a distribution $D$
	if $|\chi_D(D_i, D_j)| \leq \gamma$ for all $i \neq j$,
	and $|\chi_D(D_i, D_j)| \leq \beta$ for $i=j$.
\end{definition}

\begin{definition}[Statistical Query Dimension] \label{def:sq-dim}
	For $\beta, \gamma > 0$, a decision problem $\mathcal{B}(\D, D)$,
	where $D$ is a fixed distribution and $\D$ is a family of distributions,
	let $s$ be the maximum integer such that there exists a finite set of distributions
	$\mathcal{D}_D \subseteq \D$ such that
	$\mathcal{D}_D$ is $(\gamma, \beta)$-correlated relative to $D$
	and $|\mathcal{D}_D| \geq s.$ The {\em statistical query dimension}
	with pairwise correlations $(\gamma, \beta)$ of $\mathcal{B}$ is defined to be $s$,
	and is denoted by $\mathrm{SD}(\mathcal{B},\gamma,\beta)$.
\end{definition}

A lower bound on the SQ dimension of a decision problem
implies a lower bound on the complexity of any SQ algorithm
for the problem via the following standard result.

\begin{lemma} \label{lem:sq-from-pairwise}
	Let $\mathcal{B}(\D, D)$ be a decision problem, where $D$ is the reference distribution
	and $\mathcal{D}$ is a class of distributions. For $\gamma, \beta >0$,
	let $s= \mathrm{SD}(\mathcal{B}, \gamma, \beta)$.
	For any $\gamma' > 0,$ any SQ algorithm for $\mathcal{B}$ requires queries of tolerance at most $\sqrt{\gamma + \gamma'}$ or makes at least
	$s  \gamma' /(\beta - \gamma)$ queries.
\end{lemma}

\subsubsection*{Sparse Non-Gaussian Component Analysis} We will focus on a specific kind of decision problem given by \Cref{prob:generic_hypothesis_testing} below.

\begin{problem}[Sparse Non-Gaussian Component Analysis]\label{prob:generic_hypothesis_testing}
Let a distribution $A$ on $\R$. For a unit vector $v$, we denote by $P_{A,v}$ the distribution with the density $P_{A,v}(x) := A(v^Tx) \phi_{\perp v}(x)$, where $\phi_{\perp v}(x) = \exp\left(-\|x - (v^Tx)v\|_2^2/2\right)/(2\pi)^{(d-1)/2}$, 
i.e., the distribution that coincides with $A$ on the direction $v$ and is standard Gaussian in every orthogonal direction.  We define the following hypothesis testing problem:
\begin{itemize}
    \item $H_0$: The underlying distribution is $\cN(0,I_d)$.
    \item $H_1$: The underlying distribution is $P_{A,v}$, for some unit vector $v$ that is $k$-sparse.
\end{itemize}
\end{problem}

Specializing the result of \Cref{lem:sq-from-pairwise} for the sparse non-Gaussian component analysis, gives the following SQ lower bound. The proof is standard and is deferred to \Cref{sec:SQappendix}.
\begin{corollary} \label{cor:genericSQbound}
	Let $k,d,m \in \Z_+$ with $k\leq \sqrt{d}$. For any distribution $A$  on $\R$ that matches its first $m$ moments with $\cN(0,1)$, any constant $0<c<1$, and any SQ algorithm $\cA$ that solves the hypothesis testing \Cref{prob:generic_hypothesis_testing}, $\cA$  either makes $\Omega(d^{ck^{c}/8} k^{-(m+1)(1-c)})$ many queries or makes at least one query with tolerance at most $2^{(m/2+1)}k^{-(m+1)(1/2-c/2)}\sqrt{\chi^2(A,\cN(0,1))}$. 
\end{corollary}

When proving our main results, we will apply \Cref{cor:genericSQbound} to different choices of $A$ to get \Cref{thm:SQ4th_informal} and \Cref{thm:subgaussianSQ_informal}.

\subsubsection*{From Estimation to Hypothesis Testing}  \label{sec:reduction}
Our lower bounds will be for estimating the unknown sparse mean in $\ell_2$-error\footnote{Recall that estimating a $k$-sparse vector in $\ell_2$-norm is an easier problem than estimating an arbitrary vector in $(2,k)$-norm.}. To establish these results, we prove a stronger claim: We consider a hypothesis testing version of the robust sparse mean recovery (\Cref{prob:hypothesis_testing}). We first prove that this is an easier task than the corresponding estimation problem (\Cref{prob:search_problem}) in \Cref{lem:reduction}. We then show hardness of the hypothesis testing problem in the SQ model.

\begin{problem}[Robust Sparse Mean Estimation] \label{prob:search_problem}
	Fix $\rho > 0$. Let $\cD$ be a family of distributions such that the mean of each distribution $D$ in $\cD$ is $k$-sparse and has norm at most $\rho$. Given access to the mixture distribution $(1-\eps)D + \eps B$, for some (unknown) $D \in \cD$ and some arbitrary distribution $B$, the goal is to find a vector $u \in \R^d$ such that $\|u-\E_{X \sim D}[X]\|_2 < \rho /2$.
\end{problem}

\begin{problem}[Robust Sparse Mean Hypothesis Testing] \label{prob:hypothesis_testing}
	Fix $\rho > 0$. Let $\cD$ be a family of distributions such that the mean of each distribution $D$ in $\cD$ is $k$-sparse and has norm \emph{exactly} $\rho$. We define the following hypothesis testing problem:
	\begin{itemize}
		\item $H_0$: The underlying distribution is $\cN(0,I_d)$.
		\item $H_1$: The underlying distribution is $(1-\eps)D {+} \eps B$, for a $D {\in} \cD$ and an arbitrary distribution $B$.
	\end{itemize}
\end{problem}

\begin{claim}[Reduction] \label{lem:reduction}
	Given an algorithm $\cA$ that solves \Cref{prob:search_problem} for some $\cD$, then there exists another algorithm that solves \Cref{prob:hypothesis_testing} for $\cD'$, where $\cD'$ is the set of  all distributions in $\cD$ that have norm exactly $\rho$.
\end{claim}
\begin{proof}
	The algorithm is the following: Let $u$ be the estimate returned by $\cA$. If $\|u\|_2<\rho/2$, then return $H_0$, otherwise return $H_1$. Since, in both the null and alternative hypothesis, $u$ is guaranteed to be within $\rho/2$ of the true mean, the correctness follows.

\end{proof}

\subsection {SQ Lower Bound for Robust Sparse Mean Estimation of Gaussian Distribution with Unknown Covariance}\label{sec:k4lowerbound}
 
 Consider the task of robust sparse mean estimation of Gaussian distribution, $\cN(\mu, \Sigma)$, where $\mu$ is $k$-sparse and $\Sigma$ is unknown and bounded, $\Sigma \preceq I$.
  Information-theoretically $O((k \log (d/k) )/\eps^2 )$ samples should suffice to obtain an estimate $\widehat{\mu}$ such that $\|\widehat{\mu}  - \mu\|_2 = O(\eps)$.
The algorithm in \cite{BDLS17} can be shown to achieve the following: there is a computationally-efficient algorithm for robust sparse mean estimation that uses $O((k^2\log d)/\eps^2)$ samples and achieves error $O(\sqrt{\eps})$.
The main result of this section presents an SQ lower bound roughly stating that any SQ algorithm that achieves an error $o(\sqrt{\eps})$ either uses super-polynomially many number of queries or uses a single query that requires $k^4$ samples to simulate. 
\begin{theorem}[Formal version of \Cref{thm:main-gaussian-informal}]\label{thm:k4lower_bound}
	Let $k,d\in \Z_+$ with $k\leq \sqrt{d}$, $0<c<1, 0<\eps<1/2$, and $c_1 = 1/10001$.  Let $\cA$ be an SQ algorithm that is given access to a distribution of the form $(1-\eps)\cN(c_1 \sqrt{\eps}v,I_d - (1/3)vv^T) + \eps B$, where  $v$ is some unit $k$-sparse vector of $\R^d$ and $B$ is some arbitrary noise distribution. If the output of $\cA$ is a vector $u$ such that $\|u-c_1 \sqrt{\eps}v\|_2 \leq c_1\sqrt{\eps} /4$, then $\cA$ does one of the following:
	\begin{itemize}
		\item Makes $\Omega(d^{ck^c/8}  k^{-4+4c})$ queries,
		\item or makes at least one query with tolerance $O( k^{-2+2c} e^{O(1/\eps)})$.
	\end{itemize}
\end{theorem}

\begin{proof}
	First, we note that there exists a one-dimensional distribution which is an $\eps$-corrupted version of a Gaussian with mean $c_1\sqrt{\eps}$ and matches the first three moments with $\cN(0,1)$.

	\begin{lemma}[Lemma E.2 of~\cite{DKS19}] \label{cor:hard-distr}
		Let $\mu=c_1 \sqrt{\eps}$ with $c_1=1/10001$. For any $0<\eps<1$,
		there exists a distribution $B$ on $\R$ such that the mixture $	A = (1-\eps)\cN(\mu,2/3) + \eps B$ matches the first three moments with $\cN(0,1)$ and $\chi^2(A,\cN(0,1))= e^{O(1/\eps)}$.
	\end{lemma}
	
	We now follow the argument of \Cref{sec:reduction}. We consider \Cref{prob:search_problem,prob:hypothesis_testing} specialized to the case where $\cD$ is the family of distributions $\cD = \{(1-\eps)\cN(c_1\sqrt{\eps}v,I_d - (1/3)vv^T)+\eps B'\}_{v \in \cU_k}$, where $\cU_k$ is the set of $k$-sparse unit vectors and  $B'$ denotes a distribution whose one-dimensional projection along $v$  coincides with $B$ and every orthogonal projection is standard Gaussian, i.e., $B' = P_{B,v}$.  
	Given the reduction of \Cref{lem:reduction}, in order to prove \Cref{thm:k4lower_bound}, it remains to show that \Cref{prob:hypothesis_testing} is hard in the SQ model. To this end, we note that this is the same problem as \Cref{prob:generic_hypothesis_testing} with the distribution $A$ being that of \Cref{cor:hard-distr}. 
	An application of \Cref{cor:genericSQbound} completes the proof of \Cref{thm:k4lower_bound}.
	
\end{proof}

\subsection {SQ Lower Bound for Robust Sparse Mean Estimation of Distributions with Bounded $t$-th Moment} \label{sec:t-bounded-lb}
In this section, we will show that any SQ algorithm to obtain error $o(\eps^{1-1/t})$ either uses super-polynomially many queries or uses queries with tolerance $k^{-\Omega(t)}$.
In order to state our results formally, we define the following distribution class: let $\cP_{k,t}$ be the class of all distributions $P$ that satisfy the following:
\begin{enumerate}
    \item The mean of the distribution $P$, $\mu$, is $k$-sparse, and $\|\mu\|_2 \leq 1$.
    \item $P$ has subgaussian tails, i.e., there is a constant $c$ such that for for all unit vectors $v$ and $i \in \N$, $(\E_{X \sim P}[|v^T(X - \mu)|^i])^{1/i} \leq c \sqrt{i}$.
    \item \label{it:certifiability} For a large constant $C$, there is an SoS proof of the following inequality:
    \begin{align*}
        \{\|v\|_2^2=1\} \sststile{O(t)}{v} \E_{X \sim P}[\langle  X - \mu,v \rangle^t ]^2 \leq (C t)^t \;.
    \end{align*}
\end{enumerate}

\begin{theorem}[Formal version of \Cref{thm:subgaussianSQ_informal}]\label{thm:bounded_moment_sparse}
	Let $d,k,t \in \Z_+$ with $k\leq \sqrt{d}$, let  $0<\eps = (O(t))^{-t}$, $0\leq C < 1/2000$,  $0<c<1$, and $\delta = C \eps^{1-1/t}/t$.  Let $\cA$ be an SQ algorithm that, given access to a distribution of the form $(1-\eps)P + \eps B$,  where $P \in \cP_{k,t}$ (defined above) and $B$ is arbitrary,
$\cA$ is guaranteed to find a vector $\hat{\mu}$ such that $\|\hat{\mu}- \E_{X \sim P}[X] \|_2 \leq \delta$. Then $\cA$  does one of the following:
	\begin{itemize}
		\item Makes $\Omega(d^{ck^c/8}  k^{-(t+1)(1-c)})$ queries.
		\item Makes at least one query with tolerance $O\left(  k^{-(t+1)(1/2-c/2)} 2^{(t/2+1)}e^{ O( \delta^2/\eps^2 ) }\right)$.
	\end{itemize}
\end{theorem}

The rest of the section is dedicated to proving \Cref{thm:bounded_moment_sparse}. We first show the existence of a one-dimensional distribution $A$ that matches the first $t$ moments with $\cN(0,1)$ and is an $\eps$-corruption of a distribution with mean $\Omega(\frac{1}{t}\epsilon^{1 - 1/t})$ and bounded $t$-th moments.
At a high level, we follow the structure of \cite[Proposition 5.2]{DKS17-sq} and \cite[Lemma 5.5]{DKS18-list}. In particular, \cite[Lemma 5.5]{DKS18-list} establishes an analogous result to \Cref{lem:ABddMoment} below but in the large $\eps$ setting, i.e., $\eps \to 1$, and thus it is not applicable here. We also show that the family of hard distributions in \Cref{thm:bounded_moment_sparse} has certifiably bounded moments. We defer this work to \Cref{sec:sos-cert-hard-instance}.

\begin{lemma}\label{lem:ABddMoment}
Fix an $t \in \Z_+$
and $\eps =  O(t)^{-t}$.
	There exists a distribution $A$ over $\R$ such that the following holds:
	\begin{enumerate}
		\item There exist two distributions $Q_1$ and $Q_2$ such that $A = (1 - \epsilon)Q_1 + \epsilon Q_2$.
		\item $A$ matches first $t$-moments with $\cN(0,1)$.
		\item $\E_{X \sim Q_1}[X] = \delta$,
		where $\delta =\frac{1}{2000}\frac{1}{t}\epsilon^{1 - 1/t}$.
		\item For all $i\geq 1$, $\left(\E_{X \sim Q_1}[|X - \delta|^i]\right)^{1/i} = O(\sqrt{i})$.
		
		\item $\chi^2(A , \cN(0,1)) < \exp\left(  O( \delta^2/\eps^2 ) \right)$.
	\end{enumerate}
\end{lemma}
\begin{proof}
	Let $G(x)$ be the pdf of the standard normal $\cN(0,1)$. Thus $G(x - \delta)$ represents the pdf of $\cN(\delta,1)$.
	We will choose $A$ of the following form:
	\begin{align*}
	Q_1(x) = G(x-\delta) + \frac{1}{1-\eps}p(x) \mathbf{1}_{[-1,1]}(x), \quad Q_2(x) = G(x-\delta'),
	\end{align*}
	where $p(\cdot)$ is a degree $t$ polynomial (to be chosen below) and $\delta' = - (1 - \epsilon) \delta/ \epsilon$.
	To ensure that $Q_1$ is a valid distribution and has mean $\delta$, the following suffices since $|\delta| \leq 0.1$ and $\epsilon \leq 0.1$:
	\begin{enumerate}
		\item $\int_{-1}^1 p(x) \d x = 0$,
		\item $\max_{x \in [-1,1]}|p(x)| \le 0.1$,
		\item $\int_{-1}^1 p(x) x \d x  = 0$.
	\end{enumerate}
	Let $P_i$ be the $i$-th Legendre polynomial. We will choose $p$ to be of the following form for $a_i \in \R$:
	\begin{align*}
		p(x) = \sum_{i=0}^t a_iP_i(x),
	\end{align*}
	 where $a_0 = a_1 = 0$.
	\begin{fact} 
		\label{fact:Legendre}
		Let $P_i$ be the $i$-th Legendre polynomial. We have the following:
		\begin{enumerate}[label=L.\arabic*]
		    \item \label{item:firstwo} $P_0(x) = 1$ and $P_1(x)=x$.
			\item \label{item:Lorth}$\int_{-1}^1 P_i(x)P_j(x) dx = \frac{2}{2i +1} \delta_{i,j}$.
			\item \label{item:Lmax} $\max_{x\in [-1,1]}|P_i(x)| \leq 1$.
			\item \label{item:Lbasis} $\{P_i\}_{i= 0}^k$ form a basis of polynomials of degree up to $k$.
		\end{enumerate} 
	\end{fact}
	\begin{fact}Let $h_i$ be the $i$-th normalized probabilist's polynomials and let $X \sim \cN(0,1)$. 
		\begin{enumerate}[label=H.\arabic*]
			\item \label{item:Horth}$\E[h_i(X)h_j(X)]= \delta_{i,j}$.
			\item \label{item:Hshift} $\E[h_i(X + \mu)] = \frac{1}{\sqrt{i!}} \E[H_{e_i}(X + \mu)] = \frac{\mu^i}{\sqrt{i!}}$.
			\item \label{item:Hbasis} $\{h_i\}_{i= 0}^k$ form a basis of polynomials of degree up to $k$.
		\end{enumerate} 
	\end{fact}
	
	Using \ref{item:firstwo} and \ref{item:Lorth} we  have that {$\int_{-1}^1 p(x)\d x = 0$ and $\int_{-1}^1 p(x) x \d x = 0$.}
	Using \ref{item:Lmax}, we have that $\max_{x \in [-1,1]}|p(x)| \leq \sum_{i=1}^t |a_i|$.
	We will now ensure that it is possible to match moments while keeping $\sum_i|a_i|$ small.
	
	Recall that in order to match the first $t$ moments of $A$ with $\cN(0,1)$, we need to ensure the following holds for all $i \in \{0,\dots,t\}$:
	\begin{align*}
		(1 -\epsilon) \int_{-\infty}^{\infty} x^iG(x - \delta)dx +\int_{-1}^1 x^ip (x) dx + \epsilon \int_{-\infty}^{\infty} x^iG(x - \delta')dx = \int_{-\infty}^{\infty} x^i G(x) dx
	\end{align*}
	Equivalently, letting $X \sim \cN(0,1)$, we need the following for all $i \in \{0,\dots,t\}$:
	\begin{align*}
		\int_{-1}^1 x^ip (x) dx = \E_{X \sim \cN(0,1)}[X^i - (1 - \epsilon)(X + \delta)^i - \epsilon (X + \delta')^i].
	\end{align*}
	By \ref{item:Lbasis}, it suffices to ensure the following for all $i \in \{0,\dots,t\}$:
	\begin{align}
		\label{eq:pMatchMoment}\int_{-1}^1 P_i(x) p (x) dx =  \E_{X \sim \cN(0,1)}[P_i(X) - (1 - \epsilon)P_i(X + \delta) - \epsilon P_i(X + \delta')].
	\end{align}
	Since $\int_{-1}^1 p(x) \d x= 0$, $P_0(x) = 1$, and $ P_1(x) = x$, we have that \Cref{eq:pMatchMoment} holds for $i = 0$ and $i=1$ as both sides are zero. 
	Note that for any $i \in \{0,\dots,t\}$, the left-hand side above can be calculated using \ref{item:Lorth}:
	\begin{align}
	\label{eq:a_i_values}
		\int_{-1}^1 P_i(x) p (x) dx = \sum_{j=0}^t \int_{-1}^1 a_j P_i(x)P_j(x) = \frac{2a_i}{2i + 1}.  
	\end{align}
	We will now bound the expression on the right-hand side in \Cref{eq:pMatchMoment} to show that $a_i$ are small.
	Let $h_i$ be the $i$-th normalized probabilist's Hermite polynomials. Using \ref{item:Hbasis}, we can write $P_i(x) = \sum_{j=0}^i b_{i,j}h_j(x)$ for some $b_{i,j} \in \R$.
	We now calculate the right-hand side of \Cref{eq:pMatchMoment} as follows for a fixed $i \in \{0,\dots,t\}$:
	\begin{align*}
		\E_{X \sim \cN(0,1)}[P_i(X) &- (1 - \epsilon)P_i(X + \delta) - \epsilon P_i(X + \delta')] \\
		&=  \sum_{j=0}^i b_{i,j} \left( \E_{X \sim \cN(0,1)}[h_j(X)] - (1- \epsilon) \E_{X \sim \cN(0,1)}[h_j(X + \delta) - \epsilon h_j(X + \delta')] \right) \\
		&= \sum_{j=0}^i b_{i,j} \left( 0 - (1- \epsilon) \frac{\delta^j}{\sqrt{j!}} - \epsilon \frac{(\delta')^j}{\sqrt{j!}} \right) \\
		&= \sum_{j=0}^i  \frac{-1}{\sqrt{j!}}b_{i,j} \left( (1- \epsilon) \delta^j + \epsilon (\delta')^j \right)\;,
	\end{align*}
	where the second line uses \ref{item:Hshift}. From the proof of \cite[Claim 5.6]{DKS18-list},
	we have that $\sum_{j=0}^i b_{i,j}^2 = O((2i)^i)$.
	We are now ready to calculate the upper bound on $|a_i|$ using \Cref{eq:a_i_values}:
	\begin{align*}
		|a_i| 
		&= \left(\frac{2i + 1}{2}  \right) \left|   \sum_{j=0}^i \frac{-1}{\sqrt{j!}} b_{i,j} ( (1 - \epsilon) \delta^j + \epsilon (\delta')^j) \right|\\
		&\leq  2i     \sum_{j=0}^i \frac{1}{\sqrt{j!}} |b_{i,j}| \left( (1 - \epsilon) |\delta|^j + \epsilon |\delta'|^j \right)\\
		&\leq   4i     \sum_{j=0}^i \frac{1}{\sqrt{j!}} |b_{i,j}| \epsilon \left| \delta' \right|^j\\
		&\leq 8i^2  \epsilon \max(|\delta'|,|\delta'|^i) \max_{j \in [i]}|b_{i,j}| \\
		&\leq (2i)^{i+ 4} \epsilon \max(|\delta'|,|\delta'|^i) \;,
	\end{align*}
	where the third line uses that $\delta'=-(1-\eps)\delta/\eps$ thus $(1-\eps)|\delta|^{j}=|\delta'|^j\eps (\eps/(1-\eps))^{j-1}\leq |\delta'|^j\eps$. 
	Thus, we get the following:
	\begin{align}
		\max_{x \in [-1,1]} |p(x)| &\leq \sum_{i=1}^t|a_i| \leq \sum_{i=1}^t (2i)^{i+ 4} \epsilon \max(|\delta'|,|\delta'|^i) \leq (2t)^{t + 5} \epsilon \max(|\delta'|,|\delta'|^t) \notag \\
		&\leq  \epsilon |100t|^{t}\max(|\delta'|,|\delta'|^t) \;,
	\label{eq:SQLowerBdKmoments}
	\end{align}
	where we bounded the sum by $t$ times its last term. We would like to show that the last expression in \Cref{eq:SQLowerBdKmoments} is less than $0.1$ when $\delta = C\eps^{1-1/t}/t$ for some constant $C$.
	Note that this choice of $\delta$ implies that $|\delta'| \geq 0.5(\delta/\eps) = 0.5C (\eps^{-1/t}/t)$, which is larger than $1$ when $\eps = (O(t))^{-t}$.
	Thus the last expression in \Cref{eq:SQLowerBdKmoments} is at most $\eps (100 t |\delta'|)^t \leq \eps 100^t t^t \delta^t / \eps^t \leq \eps 100^t t^t C^t \eps^{t-1} / (\eps^t t^t) = (100C)^t$, which  is less than $0.1$ if $C \leq 0.0005$.

	Finally, the bound on the $t$-moment of $Q_1$ centered around $\delta$ 
	follows by combining the moment bounds of $\cN(\delta,1)$ 
	and noting that $p(\cdot)$ modifies the Gaussian only on the interval $[-1,1]$.
	
	It remains to bound the $\chi^2$-divergence between our distribution $A$ and $\cN(0,1)$.
	\begin{align*}
		1+\chi^2(A,\cN(0,1)) &= \int_{-\infty}^{\infty} \frac{1}{G(x)}((1-\eps) G(x-\delta)+p(x)\mathbf{1}_{[-1,1]} +\eps G(x-\delta'))^2 \d x \\
		&\leq 9 \left( \int_{-\infty}^{\infty} \frac{G^2(x-\delta)}{G(x)} \d x +  \int_{-1}^{1} \frac{p^2(x)}{G(x)}\d x + \eps^2 \int_{-\infty}^{\infty} \frac{G^2(x-\delta')}{G(x)}\d x \right) \;.
	\end{align*}
	Working with each term separately, the first one is bounded as
	\begin{align*}
		\int_{-\infty}^{\infty} \frac{G^2(x-\delta)}{G(x)} \d x \leq 1+\chi^2(\cN(\delta,1), \cN(0,1)) = e^{\delta^2} \;,
	\end{align*}
	the last one is similarly bounded above by $\eps^2 e^{\delta'^2}$ and for the first one we have that
	\begin{align*}
		\int_{-1}^{1} \frac{p^2(x)}{G(x)}\d x \leq \left(\max_{x \in [-1,1]} |p(x)| \right)  \max_{x \in [-1,1]} \frac{1}{G(x)} = O(1) \;.
	\end{align*}
	Given $(\delta')^2=\Theta(\delta^2/\eps^2)$, all three terms are at most $\exp(  O( \delta^2/\eps^2 )$
	, therefore we have that $\chi^2(A,\cN(0,1)) =  \exp\left(  O( \delta^2/\eps^2 ) \right)$. 
\end{proof}

\subsubsection{Proof of \Cref{thm:bounded_moment_sparse}}

\begin{proof}[Proof of \Cref{thm:bounded_moment_sparse}]
	We will prove \Cref{thm:bounded_moment_sparse} using \Cref{lem:ABddMoment} with the argument of \Cref{sec:reduction}. Let $A$ be the distribution from \Cref{lem:ABddMoment}. We consider \Cref{prob:search_problem,prob:hypothesis_testing} with $\cD=\{P_{A,v} \}_{v \in \cU_k}$ (using the notation from \Cref{prob:generic_hypothesis_testing}). 
Using the notation of \Cref{lem:ABddMoment}, we see that every $P_{A,v}$ in this choice of $\cD$ is of the following form $P_{A,v} = (1 - \eps) P_{Q_1, v} + \eps P_{Q_2,v} $, where $P_{Q_1, v}$ belongs to $\cP_{k,t}$ as defined in the beginning of this section:  (i) its mean is $k$-sparse (since $v$ is $k$-sparse), (ii) it satisfies subgaussian tail bounds (since $Q_1$ has subgaussian tails, see \Cref{lem:ABddMoment}), and (iii) it has $t$-certifiably bounded moments (\Cref{cl:SoS_certifiability}).
	\Cref{prob:search_problem} is then the same as \Cref{prob:generic_hypothesis_testing}. By the reduction of \Cref{lem:reduction}, it remains to show the SQ-hardness of the latter problem. We then use \Cref{cor:genericSQbound}. 
\end{proof}

As a note, by simply replacing the set $S$ of the $k$-sparse direction of \Cref{fact:vectors} by an analogous set of dense $2^{d^c}$ vectors (see, e.g., \cite[Lemma 3.7]{DKS17-sq}) we can get an analog of the previous theorem for the dense case.

\begin{theorem}[SQ Lower Bound in Dense Case]
\label{thm:sq-dense}
Let $t \in \Z_+$, $0<\eps = (O(t))^{-t}$, $C<1/2000$, $0<c<1/2$,and $\delta = C \eps^{1-1/t}/t$.
Any SQ algorithm that, given access to a distribution of the form $(1-\eps)P + \eps N$ where $P$ is a distribution with $\E_{X \sim P}[|v^T X|^i]^{1/i} = O(\sqrt{i})$ for every $i \leq t$ and every $v \in \cS^{d-1}$ and finds a vector $\hat{\mu}$ such that $\|\hat{\mu}- \E_{X \sim P}[X] \|_2 \leq \delta$ does one of the following:
	\begin{itemize}
		\item Makes $2^{\Omega(d^c)} d^{-(t+1)(1/2-c)}$ queries.
		\item Makes at least one query with tolerance $\left(  O(d)^{-(t+1)(1/4-c/2)} e^{ O( \delta^2/\eps^2 )} \right)$.
	\end{itemize}
\end{theorem}

\section*{Acknowledgements}

We thank Pravesh Kothari for useful clarifications regarding prior work. We also thank Jonas Riehle and Markus Schweighofer for pointing out 
a minor oversight in the runtime calculation 
appearing in an earlier version of this paper.

%% file: app-preliminaries.tex
\section{Omitted Background from \Cref{sec:prelims}}
\label{app:addDetails}

\medskip

\subsection{Basic facts}

For completeness, we prove \Cref{fact:sparseTruncation} below, stating that estimating in $(2,k)$-norm implies sparse mean estimation. 
\FactTruncSparse*

\begin{proof}
For a vector $x = (x_1,\ldots,x_d) \in \R^d$ and a set of indices $S \subset [d]$, we will use the notation $x_S$ to denote the vector that contains only those elements of $x$ whose indices lie in $S$.
Let $\|x-\mu \|_{2,k} = b$.
	Let $S^* := \supp(\mu)$ and $S' := \supp(h_k(x))$. Then,  $ \| (\mu - h_k(x))_{S^*} \|_2 \leq B$ and $\| (x)_{S' \setminus S^*} \|_2 = \| (\mu - h_k(x))_{S' \setminus S^*} \|_2   \leq b$.
	
	If $h_k(x) = x$, then we are done because $\| (\mu - x)_{(S' \setminus S^*) \cup S^*}  \|_2 \leq 2b$.
	If not, then $|S'| = k$. Since $|S^*| \leq k$, it follows that $|S'\setminus S^*| \geq |S^* \setminus S'|$. 
	Since $S'$ contains the indices for the $k$ largest entries (in magnitude) of $x$, for any $i \in S'\setminus S^*$ and $j \in S^* \setminus S'$, $|x_i| \geq |x_j|$. Since $\| (x)_{S' \setminus S^*} \|_2 \leq b$, at least one coordinate $j \in S' \setminus S^*$ must satisfy $(x_j)^2 \leq  b^2/|S' \setminus S^*|$. Therefore, for every $i \in S^* \setminus S'$ we have $(x_i)^2 \leq  b^2/|S' \setminus S^*|$. Adding these up we get the following upper bound on $\| (x)_{S^* \setminus S'} \|_2$.
	\[ \| (x)_{S^* \setminus S'} \|_2^2 = \sum_{i \in S^* \setminus S'} (x)_i^2 \leq  b^2 \cdot \frac{|S^* \setminus S'|}{|S' \setminus S^*|} \leq  b^2.\]
	Finally, we have that
	\[ \| \mu - h_k(x)\|_2^2 = \| (\mu - x)_{S' \cap S^*}\|_2^2 + \|(\mu)_{S^* \setminus S'} \|_2^2 + \|(x)_{S' \setminus S^*}\|_2^2 \leq 6 b^2, \]
	where the  bound on $\|(\mu)_{S^* \setminus S'} \|^2$ follows by a triangle inequality and the fact $\| (\mu - x)_{S^*} \|_2 \leq b$. 
\end{proof}

The following fact gives an explicit expression for the moments of a Gaussian distribution:
\begin{fact}[Moments of Gaussian]\label[fact]{fact:gaussian-moments}
	For \emph{any} $v \in \mathbb{R}^d$ and any $s \in \bbN$, the moments of $\mathcal{N}(\mu, \Sigma)$ are 
	$\E_{X \sim \mathcal{N}(\mu, \Sigma)}\Brac{\iprod{v, X - \mu}^{2s}} = (2s-1)!! \E_{X \sim \mathcal{N}(\mu, \Sigma)}\Brac{\iprod{v, X - \mu}^{2}}^s$.
\end{fact}
	
\subsection{Additional SoS Background}\label{sec:sos_prelims}
We record some additional facts that we will use in our proofs.
\begin{fact}[Cauchy-Schwarz for Pseudoexpectations]\label{fact:CS_pseudo_exp}
	Let $f,g$ be polynomials of degree at most $t$. Then, for any degree-$2t$ pseudoexpectation $\pE$,
	$\pE[fg] \leq \sqrt{\pE [f^2]} \sqrt{\pE[g^2]}$.
	Consequently, for every squared polynomial $p$ of degree $t$, and $k$ a power of two, 
	$\pE[p^k] \geq  (\pE[p])^k$ for every $\pE$ of degree-$2tk$.
	\label{fact:pseudo-expectation-cauchy-schwarz}
\end{fact} 

\begin{fact}[SoS Cauchy-Schwartz and H\"older (see, e.g., \cite{hopkins2018clustering})]\label{fact:sos-holder}
	Let $f_1,g_1,  \ldots, f_n, g_n$ be indeterminates
	over $\R$. Then, 
	\begin{align*}
	\sststile{2}{f_1, \ldots, f_n,g_1, \ldots, g_n} \Set{ \Paren{\frac{1}{n} \sum_{i=1}^n f_i g_i }^{2} \leq \Paren{\frac{1}{n} \sum_{i=1}^n f_i^2} \Paren{\frac{1}{n} \sum_{i=1}^n g_i^2} } \;.
	\end{align*} 
	The total bit complexity of the SoS proof is $\poly(n)$. Moverover, if $p_1, \dots, p_n$ are indeterminates, for any $t \in \Z_+$ that is a power of $2$, we have that
	\begin{align*}
	\{w_i^2 = w_i \mid i \in [n] \} \sststile{O(t)}{p_1, \dots, p_n} \left( \sum_i w_i p_i\right)^t &\leq \left( \sum_{i \in [n]} w_i\right)^{t-1} \cdot \sum_{i \in [n]} p_i^t \quad \text{ and} \\
	\{w_i^2 = w_i \mid i \in [n] \} \sststile{O(t)}{p_1, \dots, p_n} \left( \sum_i w_i p_i\right)^t &\leq \left( \sum_{i \in [n]} w_i\right)^{t-1} \cdot \sum_{i \in [n]} w_ip_i^t \;.
	\end{align*} 
	The total bit complexity of the SoS proof is $\poly(n^t)$.
\end{fact}

\begin{fact}[SoS Triangle Inequality]\label{fact:sos-triangle}
	If $k$ is a power of two, 
	$\sststile{k}{a_1, a_2, \ldots, a_n} \Set{ \left(\sum_i a_i \right)^k \leq n^k \Paren{\sum_i a_i^k} }.$ The total bit complexity of the SoS proof is $\poly(n^k)$.
\end{fact}

We will apply the above facts in a way so that the final bit complexity of these SoS proofs will be bounded by $\poly(m^t,d^t)$.

\begin{restatable}{fact}{factAtMostDegree}
 \label{fact:AtMostDeg}
    Any degree-$t$ polynomial $r(x)$ in $d$ variables which is a sum of square polynomials, can always be written as a sum of at most $d^{t/2}$ square polynomials.
\end{restatable}
 \begin{proof}
Let $r(x) = \sum_{j} q_j(x)^2$. Observe that $q_j(x) = \Iprod{u_j, m(x)}$ where $m(x)$ is the vector of all possible monomials up to degree $t/2$ of the variables $x_1, \dots, x_d$ and $u_j$ is the vector containing the coefficients used for each of them in the polynomial $q_j$. Let $\sum_j u_j u_j^T = U$, then $r(x) = m(x)^T U ~m(x)$. Note that $U$ is a positive semidefinite matrix. It therefore has an eigen-decomposition of at most $d^{t/2}$ vectors $v_1, \dots, v_{d^{t/2}}$ with eigenvalues $\lambda_1,\ldots,\lambda_{d^{t/2}}\geq 0$. This means that we can write $r(x) = \sum_{j=1}^{d^{t/2}} \lambda_j m(x)^T v_j v_j^T m(x) = \sum_{j=1}^{d^{t/2}} h_j(x)^2$ where $h_j(x) = \sqrt{\lambda_j}\Iprod{v_j, m(x)}$. 
\end{proof}

The following fact is a simple corollary of the fundamental theorem of algebra:
\begin{fact} \label{fact:univariate}
	For any univariate degree $d$ polynomial $p(x)$, with $p(x) \geq 0$ for all $x \in \R$, 
	$\sststile{t}{x} \Set{p(x) \geq 0}$.
 \end{fact}

This can be extended to univariate polynomial inequalities over intervals of $\R$. 

\begin{fact}[Fekete and Markov-Lukács, see \cite{laurent2009sums}]
	For any univariate degree $d$ polynomial $p(x) \geq 0$ for $x \in [a, b]$,  $\Set{x\geq a, x \leq b} \sststile{d}{x} \Set{p(x) \geq 0}$.  \label{fact:univariate-interval}
\end{fact}

\subsubsection{Quantifier Elimination}\label{app:quantifier_elim}

In this section, we describe a set of constraints that guarantee that the  variables of a given SoS program satisfy a certain polynomial inequality \emph{for all} (possibly infinite) values of some subset of the given variables, i.e., essentially leave a desired subset of the variables free. 
This is particularly useful to us since we would like to ensure that our samples have certifiably bounded moments in \emph{all} $k$-sparse directions.
Concretely, let $V$ be the set of variables, let $F \subseteq V$ be the set of free variables, $\cA$ be a set of polynomial constraints on $F$, and let $b \in \R[V]$. Suppose we like to ensure that $b(V) \geq 0$ for all values of $F$ that satisfy $\cA$. 
The basic idea here is to observe that it is enough to ensure that there is an SoS proof of this inequality in the variables $F$, and that this proof can be obtained by ensuring that a certain list of polynomials exist whose coefficients satisfy specific equalities. Hence it is sufficient to add a list of polynomial equality constraints. 
These constraints will become clearer in the following discussion. 

We will need the following notation: if $a_1(x),\dots,a_d(x)$ are polynomials in $x$ and $T \in [d]^t$ is an ordered tuple,  $a_T(x)$ is defined to be $a_T(x) := \prod_{i\in T}a_i(x)$. Also, let $d, t  \in \mathbb{N}$ and $V := \{x_1, \dots, x_d\}$ be formal variables and let $b \in \R[x_1, \dots, x_d]$ of degree at most $t$.

We are now ready to provide the details below. Define the following:
	\begin{enumerate}
		\item Let $F \subset V$ denote the subset of variables that we would like to leave free.
		\item Let $\cA = \{a_1, \dots, a_r \} \subset \R[F]$ be a set of polynomials in $F$ of degree at least 1. Suppose the variables $F$ satisfy $\{ a(F) = 0 \mid a \in \cA \}$.
	\end{enumerate}
	Consider an assignment $\pi$ to the variables $V \setminus F$. We define $b_\pi(F)$ to be the polynomial that is obtained by assigning the variables in $V \setminus F$ in $b(V)$ according to the assignment $\pi$. 
	We know from \Cref{def:sos-proof} that $ \{ a \geq 0 \mid a \in \cA \} \sststile{t}{F} b_\pi(F) \geq 0$, if and only if 
	\begin{align*}
	\label{eq:quant-elimin}
	    b_\pi(F) = \sum_{T \subset [r], |T| \leq t} a_T(F) q_T(F),
	\end{align*} 
	where each $q_T$ is a sum of $D$ square polynomials,
	where by 
     \Cref{fact:AtMostDeg} we can assume that $D < |F|^{O(t)}$. If the constraints are instead $\{ a = 0 \mid a \in \cA \}$, then the condition can be changed to 
	\begin{equation} \label{eqn:sos-pf-b-nng}
		b_\pi(F) = \sum_{i \in [r]} a_i(F) p_i(F) + q(F) ,
	\end{equation}
	where each $p_i$ is an arbitrary polynomial in $F$	and $q$ is a sum of at most $D$ square polynomials in $F$ for $D < |F|^{O(t)}$, and the degree of each term on the right-hand side is at most $t$.  In the context of our paper, $\cA$  above will be a set of polynomial equalities which are satisfied only by sparse vectors.

\begin{definition}[Quantifier Elimination]\label[definition]{def:proof_constraint}
Let $d, t  \in \mathbb{N}$ and $V := \{x_1, \dots, x_d\}$ be formal variables and let $b \in \R[x_1, \dots, x_d]$ of degree at most $t$. Let $F \subset V$ and $\cA=\{a_1,\ldots, a_r\} \subset \R[F]$ polynomial axioms of degree at least 1. 
We define $\textbf{cons}_F(\cA, \{b\}, t)$ to be the set of equality constraints that equate the coefficients of $F$ of the polynomials in \Cref{eqn:sos-pf-b-nng}, where the coefficients may involve polynomials of $V\setminus F$.
This is done by introducing variable vectors $\{ P_{i} \mid i \in [r] \}$ for the coefficients of $p_i$ and $ \{ Q_{j} \mid j \in [D] \}$ for the coefficients of $q$ in \Cref{eqn:sos-pf-b-nng} (where $D<|F|^{O(t)}$) and equating the coefficients of the LHS and RHS when both sides are interpreted to be polynomials in $F$.
This leads to at most $|F|^{O(t)}$ many equality constraints in the variables $\{x_i \mid  i \in V \setminus F\}$, $\{P_i \mid i \in [r]\}$, $\{Q_j \mid j \in [D]\}  $, and each $P_i$ and $Q_j$ is of dimension at most $|F|^{O(t)}$. 
\footnote{Note that if there is an SoS proof of $b$ subject to $\cA$ having bounded bit-complexity, then there is a solution to $\textbf{cons}_F(\cA, \{b\}, t)$ which has bounded $\ell_2$ norm. }

\end{definition}
The following fact from \cite{kothari2017outlier} allows us to effectively use the constraints defined above.

\begin{fact}\label{fact:quant_el_fact}
In the setting of \Cref{def:proof_constraint}, for any fixed $F \subset V$ and fixed assignment $\pi$ to $V \setminus F$, we can extend this assignment to a solution of $\textbf{cons}_F(\cA, \{b\}, t)$ 
	iff  $\cA \sststile{t}{F} \{b_\pi(F) \geq 0\},$ where $b_\pi  \in \R[F]$ is obtained by assigning the variables in $V \setminus F$ in $b(V)$ according to the assignment $\pi$.
\end{fact}
\begin{restatable}{fact}{BASICFACT}\label{fact:cons_sos_proof}
Consider the setting in \Cref{def:proof_constraint}. 
Let $V' = \{P_i \mid i \in [r]\} \cup \{Q_j \mid j \in [D]\}$.
Let $\pi$ be an assignment to $F$ that satisfies $a_i \in \cA$, i.e., $a_i(\pi(F)) = 0$ for each $i \in [r]$.
Let $b_\pi(V\setminus F)$ be the polynomial in $V \setminus F$ that is obtained by assigning the variables in $F$ in $b(V)$ according to the assignment $\pi$.
Then $\textbf{cons}_F(\cA, \{b\}, t) \sststile{t}{V \setminus F, V'} b_\pi(V\setminus F) \geq 0$.
\end{restatable} 
\begin{proof}
Consider the polynomial $h(F,V') = \sum_{i=1}^r a_i(F) p_i(F) + \sum_{j=1}^D q_j^2(F)$, where $\{P_i\}$ and $\{Q_j\}$ are coefficients of $p_i$ and $q_j$ respectively. 
Note that $\textbf{cons}_F(\cA, \{b\}, t)$ is a set of polynomial equality constraints in the variables $(V \setminus F) \cup V' $ that enforce the coefficients of the two polynomials $b(F, V\setminus F)$ and $h(F,V')$, when expanded in the monomial basis in $F$, to be equal. 
That is, for each $S \in [|F|^t]$, $\textbf{cons}_F(\cA, \{b\}, t)$ contains the constraint  $c_S(V\setminus F,V') = 0$, where $b(F, V \setminus F) - h(F,V') = \sum_{S \in [|F|^t]} c_S(V \setminus F, V') F_S$ and $c_S(V \setminus F, V')$ is a polynomial in $V \setminus F$ and $V'$.  

Our goal is to show that the inequality $b_\pi(V\setminus F) \geq 0$ has an SoS proof subject to $\textbf{cons}_F(\cA, \{b\}, t)$. We show this below. Observe that, 
\begin{align*}
    b(F, V\setminus F) &= (b(F, V\setminus F) - h(F,V')) + h(F,V') = \sum_{S \in [|F|]^t}   c_S(V\setminus F, V') F_S  + h(F,V').
\end{align*}  
Let $f = \pi(F)$. Since the assignment $\pi$ satisfies the $a_i$'s, we see that $h(f, V') = \sum_{i=1}^r a_i(f) p_i(f) + \sum_{j=1}^D q_j^2(f) = \sum_{j=1}^D q_j^2(f)$. Hence, 
\begin{align*}
    b(f, V\setminus F) &= \sum_{S \in [|F|]^t} f_S c_S(V\setminus F, V') + \sum_{j=1}^D q_j^2(f).
\end{align*}  
This is a valid SoS proof from the axioms $\textbf{cons}_F(\cA, \{b\}, t)$.
\end{proof}

%% file: app-certifiably-sparse-moments.tex
\section{Certifiability for $\sigma$-Poincar\'e Distributions} \label{sec:poincareAppendix}

Previous work has shown that $\sigma$-Poincar\'e distributions have certifiably bounded moments. In this section we show that this implies that $\sigma$-Poincar\'e distributions also have certifiably bounded moments in $k$-sparse directions. At the end of the section, we demonstrate that certifiability of the moments in $k$-sparse directions does not always imply the same condition for all (possibly dense) directions.

\begin{restatable}{lemma}{PoincareToBounded}\label{lem:poincareToBounded} 
	If $D$ is a $\sigma$-Poincar\'e distribution over $\R^d$ with mean $\mu$, then for some constant $C_t$ depending on $t$, we have that 
	$	\cAksparse \sststile{O(t)}{v,z} \E_{X \sim D}\Brac{\iprod{v, X - \mu}^t}^2 \leq C_t^2 \sigma^{2t}$. The bit complexity of the proof is a factor of at most some $\poly(t)$ more than the bit complexity of the polynomial  $\E_{X \sim D}\Brac{\iprod{v, X - \mu}^t}^2 - C_t^2 \sigma^{2t}$.
\end{restatable}
\begin{proof}
Previous work focused on the notion of \emph{certifiably bounded moments} in the absence of sparsity constraints, i.e.,  $\{\sum_i v_i^2 = 1 \}\sststile{t}{v} M^2 \geq \E_{X \sim D}\Brac{\iprod{ v,  X - \mu}^{t} }^2$. The following claim implies that if a distribution has certifiably bounded moments, then it also satisfies \Cref{def:bounded-moments-k-sparse}. 
\begin{restatable}
[Proofs transfer to unit $k$-sparse vectors]{claim}{ClaimPrfTransfer}
	\label{cl:PrfTransfer}
	For every polynomial $p:\bbR^d \rightarrow \bbR$, if there is a proof of  $\Set{\sum_i v_i^2 = 1}\sststile{t}{v} p(v_1, \dots, v_d) \geq 0$ with bit complexity $B$,  then there is a proof of $\cAksparse \sststile{t}{v, z} p(v_1, \dots, v_d) \geq 0$ with bit complexity at most $B$. 
\end{restatable}
\begin{proof}
	To show $\cAksparse \sststile{t}{v, z} p(v_1, \dots, v_d) \geq 0$, it suffices to demonstrate that there exists a set of polynomials $\{r_c(v, z)\}_{c \in  \cAksparse  }$ and a sum of square polynomials $Q(\cdot)$ such that:
	\[
	p(v_1, \dots, v_d) = \sum_{c \in \cAksparse} r_c(v, z) c(v,z) + Q(v, z),
	\]
	where the polynomials $r_c(v, z) \cdot c(v,z)$ and $Q(v,z)$ have degree at most $t$. However, we know that $p(v) = q(v,z) + (\sum_j v_j^2 - 1) q'(v,z)$ for some polynomial $q'$ of degree $t$ and some sum of square polynomials $q$ also of degree $t$. Setting $r_{\{\sum_j v_j^2 - 1\}} = q'$, $Q = q$, and  $r_{c} = 0$ for all $c \neq  \{\sum_j v_j^2 - 1\} $ proves our claim.
\end{proof}
The following lemma, implicit in Theorem 1.1 from \cite{KStein17}, says that if a distribution $D$ is $\sigma$-Poincar\'e, i.e.,  it holds  $\Var_{X \sim D}\Brac{f(X)} \leq \sigma^2 \E_{X \sim D} \Brac{ \Norm{\nabla f(X)}_2^2}$ for all differentiable functions $f:\R^d \rightarrow \R$,
then it has certifiably bounded moments in every (possibly dense) direction.
\begin{lemma}[\cite{KStein17}] 
	\label{lem:PoincareKStein}
	If $D$ is a $\sigma$-Poincar\'e distribution over $\R^d$ with mean $\mu$, then there exists some constant $C_t$ depending only on $t$,  such that 
	\begin{align*}
	\Set{\sum_{i=1}^d v_i^2 = 1} \sststile{O(t)}{v} \E_{X \sim D}\Brac{\iprod{v, X - \mu}^t} \leq C_t \sigma^t \;.
	\end{align*}
	Moreover, the bit complexity of this proof is at most $\poly(t,b)$, where $b$ is the bit complexity of the coefficients of the polynomial $ C_t \sigma^t - \E_{X \sim D}\Brac{\iprod{v, X - \mu}^t} $.  
\end{lemma}
 
Combining \Cref{cl:PrfTransfer} and \Cref{lem:PoincareKStein}, and using the fact that for any polynomials $A, B$,  $\{ 0 < A < B\}\sststile{}{} A^2 < B^2$,   completes the proof of \Cref{lem:poincareToBounded}. 
\end{proof}

Regarding the difference between the two definitions of certifiably bounded moments, one for the dense setting (\Cref{def:Sosbddmoments}) and one for the sparse setting (\Cref{def:bounded-moments-k-sparse}), we note that there exist distributions that satisfy \Cref{def:bounded-moments-k-sparse} but do not have certifiably bounded moments in every direction (with a dimension-independent $M$):
Let $\xi$ be the Rademacher random variable and define $D$ to be the distribution of the random variable $X = (X_1,\dots,X_d)$, where each $X_i = \xi$.
Let $\mu$ and $\Sigma$ be the mean and covariance matrix of $D$. 
Since the operator norm of $\Sigma$ is $\sqrt{d}$, it follows that there exists a unit vector $v^*$ (we can take $v^*= (1/\sqrt{d},\dots,1/\sqrt{d})$) 
such that for any even $t$, $\E_{X\sim D}[\langle v^*, X - \mu  \rangle^{t}]^2 \geq d^{t}$.
Thus the distribution $D$ does not satisfy \Cref{def:Sosbddmoments} with any dimension-independent bound.
However, we have that $\cAksparse \sststile{O(t)}{v,z} \E_{X \sim D}[ \langle v, X - \mu\rangle ^t]^2  \leq k^{t}$ by noting that $\E_{X \sim D}[ \langle v, X - \mu\rangle ^t] = \E [(\sum_i v_i \xi)^{t}] = (\sum_i v_i)^t$ and applying \Cref{claim:sparse-upper-bound}.

\section{Concentration Inequalities for SoS-sparse-certifiability}\label{sec:concentration_appendix}
The goal of this section is to understand the sample complexity required for the result of \Cref{sec:sampling}. Throughout this section, we let $\| X \|_{L_p}$ denote the $L_p$-norm of the real-valued random variable $X$, which is defined as $(\E[|X|^p])^{1/p}$.
We begin by showing the following concentration result that will be useful in the subsequent proofs.
\begin{lemma}
	\label[lemma]{lem:ConcTensorInf}
	Let $P$ be a random variable over $\R^d$ with mean $\mu$ and suppose that for all $s \in [1,\infty)$, $P$ has its $s^{th}$ moment bounded by $(f(s))^s$ for a non-decreasing function $f: [1,\infty) \to \R_+$, in the direction $e_j$, i.e., suppose that for all $j \in [d]$ and $X \sim P$:
$
		\|\iprod{e_j, X - \mu}\|_{L_s} \leq  f(s).
$
	Let $S$ be a set of $m$ i.i.d.\ samples of $P$.
	For some sufficiently large absolute constant $C>0$, we have the following:
	\begin{enumerate}
		\item ($t$-th moment tensor)  If $m > C \frac{1}{\delta^2} \left( t \log d + \log(1/\gamma)  \right) \left( f(t^2 \log d + t\log(1/\gamma))   \right)^{2t}$, then, with probability $1 - \gamma$, the $t^{th}$ central moment tensor of $P$ is bounded in $\ell_{\infty}$ by $\delta$, i.e. 
		\begin{align}
			\label{eq:concInfNorm}
			\left\| \E_S \Brac{(X - \mu)^{\otimes t}} -\E_{X \sim P} \Brac{(X - \mu)^{\otimes t}} \right\|_\infty \leq \delta.
		\end{align}
		\item (Absolute moments) If $m > C  \log(d/\gamma)(f(t\log(d/\gamma))/f(t))^{2t}$, then,  with probability $ 1 - \gamma$, for all $i \in [d]$ and for all $r \in [t]$:
		\begin{align}
			\label{eq:absoluteMoments}
			\left[  \E_S[|(X - \mu)_i|^r]\right]^{\frac{1}{r}} \leq 2 f(t). 
		\end{align}
		\item (Sample Mean)  If
		$m  > C (1/ \delta^2)  \log(d/\gamma) ( f(\log(d/\gamma))  )^2 $,
		then, with probability $1 - \gamma$,
		\begin{align}
			\norm{\E_S[X]- \mu}_{\infty} \leq \delta.
		\end{align}
		
	\end{enumerate}
\end{lemma}

\begin{proof}
	It suffices to consider the case when $\mu = 0$.
	\paragraph{Part 1}
	For any ordered tuple $T \in  [d]^t$, we define $p_T : \R^d \to \R$ as $p_T(x) := \prod_{j \in T } x_j$.
	Let $Y \sim P$. 
	It suffices to show the following:
	\begin{align*}
		\forall T \in [d]^t:  \left|\frac{1}{m}\sum_{i=1}^m \Paren{ p_T(X_i) - \E[p_T(Y)]} \right| \leq \tau.	\end{align*}
	Define $Z_{T,i}:= p_T(X_i) - \E[p_T(Y)]$ for $i\in[m]$ and $Z_T = \frac{1}{m} \sum_{i=1}^m Z_{T,i}$.
	Let $s \in \Z_+$. We will control the $s$-th moment of $Z_T$ using the bound on the $s$-th moment of $Z_{T,i}$ and independence of $(Z_{T,i})_{i=1}^m$.
	Recall that $X_i$ has the same distribution as $Y$.
	\begin{align*}
		\|Z_{T,i}\|_{L_s} = \|p_T(X_i) - \E[p_T(Y)]\|_{L_s} \leq 2 \|p_T(Y)\|_{L_s},
	\end{align*}
	where we use triangle inequality and Jensen's inequality.
	We use $Y_j$ to denote the $j$-th coordinate of $Y$. Using the moment bounds on $p_T(Y)$ and H\"older inequality, we get the following:
	\begin{align}
	\label{eq:sthmomentBound}
\|p_T(Y)\|_{L_s}^s &=		\E\left[\left( \prod_{j \in T} Y_j \right)^s\right] = \E\left[\prod_{j \in T}\left( Y_j^s \right)\right]
		\leq \prod_{j \in T}\left( \E\left[Y_j^{st}\right]  \right)^{\frac{1}{t}}\leq \prod_{j \in T} \left( f(st)  \right)^s =  (f(st))^{st},
	\end{align}
	where the first inequality above uses the Cauchy-Schwarz inequality for products of $t$ variables and the second inequality uses the assumption on the moments of $Y_j$. 
	Thus, $\|Z_{T,i}\|_{L_s}  \leq 2\|p_T(Y)\|_{L_s} \leq 2 (f(st))^{t}$.
	
	We will use the following inequalities:
	\begin{fact} (Marcinkiewicz-Zygmund’s inequality)
		\label{fact:MZineq} Let $W_1,\dots,W_m, W$ be identical and
		independent centered random variables on $\R$ with a finite $s$-th moment for $s\geq 2$. Then,
		\begin{align*}
			\left\| \frac{1}{m} \sum_{i=1}^m W_i\right\|_{L_s} \leq \frac{3 \sqrt{s}}{\sqrt{m}} \|W\|_{L_s}. 
		\end{align*}
	\end{fact}
	\begin{fact}
		\label{fact:tailsMoments}
		For a random variable $X$, we have that w.p. $1- \gamma$, $|X - \E [X]| \leq e\|X - \E [X] \|_{L_{\log(1/\gamma)}}$.
	\end{fact} 
	\begin{proof}
		Let $Y = X - \E [X]$.
		We have the following:
		\begin{align*}
			\pr \left[|Y| \geq e \|Y\|_{L_{\log(1/\gamma)}}\right] &\leq
			\frac{\E [|Y|^{\log(1/\gamma)}] } { e^{\log(1/\gamma)} \E [|Y|^{\log(1/\gamma)}]} 
			= \frac{1}{e^{\log(1/\gamma)}} = \gamma. 
		\end{align*}
	\end{proof}
\noindent Using \Cref{fact:MZineq} and the moment bounds in \eqref{eq:sthmomentBound}, we get that for any $T \in [d]^t$,
	\begin{align*}
		\|Z_{T}\|_{L_s} \lesssim \frac{\sqrt{s}}{\sqrt{m}} ( f(st))^{t}. 
	\end{align*}
	Using \Cref{fact:tailsMoments} with the above claim, we have that with probability $1 - \gamma'$,
	\begin{align*}
		|Z_T| \leq e \|Z_T\|_{L_{\log(1/\gamma')}} \lesssim \frac{\sqrt{\log(1/\gamma')}}{\sqrt{m}} ( f(t\log(1/\gamma')))^{t}.
	\end{align*}
	Taking a union bound over $T\in [d]^t$ ordered tuples and taking $\gamma' = \gamma/d^t$, we get that with probability $1 - \gamma$,
	\begin{align*}
		\forall T \in [d]^t: |Z_T| \lesssim \frac{\sqrt{t \log d + \log(1/\gamma)}}{\sqrt{m}} ( f(t^2 \log d + t\log(1/\gamma)))^{t}   \;.
	\end{align*}
	This completes the proof of the first claim.
	\paragraph{Part 2} %
	\label{par:part_2}
	
	Let $Y:= (Y_1,\dots,Y_d)$ be distributed as $P$.
	Using monotonicity of $L_p$ norms, it suffices to bound, for all $i \in [d]$, $\left[  \E_S[|(X - \mu)_i|^t]\right]^{\frac{1}{t}}$.
	
	Recall that we assume $\mu = 0$ without loss of generality.  For an $i \in [d]$, $j \in [m]$, let $Z_{i,j}:= |(X_{j})_i|^t$ and $Z_i := \frac{1}{m} \sum_{j=1}^mZ_{i,j}$.
	By assumption, we have the following for all $r \geq 1$:
	\begin{align*}
		\E[|Z_{i,j}|^r] = \E[|Y_i|^{rt}] \leq \left( f(rt)  \right)^{rt}.
	\end{align*}
	Thus  $\|Z_{i,j}\|_{L_r} \leq f(rt) ^t$. In particular, for all $r \geq 1$, we have $|\E[Z_i]| = |\E[Z_{i,j}]| \leq  \|[Z_{i,j}\|_{L_r} \leq (f(rt))^t$, where the first inequality follows from the monotonicity of $L_p$-norms.
	Thus we have that $ \|Z_{i,j} - \E[Z_{i,j}]\|_{L_r} \leq 2(f(rt) )^t$.

	Applying \Cref{fact:MZineq}, we have that for all $r\geq 1$
	\begin{align*}
		\|Z_i - \E[Z_i]\|_{L_r} \lesssim \sqrt{\frac{r}{m}} (f(rt) )^t.
	\end{align*}
	Applying \Cref{fact:tailsMoments}, we have that, with probability $1 - \gamma'$, we have that
	\begin{align*}
		|Z_i - \E[Z_i]| \lesssim  \sqrt{\frac{\log(1/\gamma')}{m}} \left(f(t\log(1/\gamma')) \right)^t . 
	\end{align*}
	Taking $\gamma' = \gamma/d$ with a union bound, we have the following:
		\begin{align*}
	\forall i \in [d]:\,\,\,\,	Z_i \leq  (f(t))^t \left( 1 +  C \sqrt{\frac{\log(d/\gamma)}{m }} \left(  \frac{f(t\log(d/\gamma)}{f(t)}\right) ^t\right),
	\end{align*}
	where $C$ is a large enough constant.
	The bound follows by noting that $Z_i^{1/t} = \left[  \E_S[|(X - \mu)_i|^t]\right]^{\frac{1}{t}}$.
	
	\paragraph{Part 3} %
	\label{par:part_3}
	For $i \in [d]$, $j \in [m]$, let $Z_{i,j} := (X_j)_i$ and  $Z_i := \frac{1}{m}  \sum_{j=1}^m Z_{i,j}$. 
	
	We have that $\|Z_{i,j}\|_{L_s} \leq f(s)$. Applying \Cref{fact:MZineq}, we get the following: with probability $1 - \gamma/d$,
		\begin{align*}
	\|Z_i - \E[Z_i]\|_{L_{\log(d/\gamma)}} \lesssim \frac{ \sqrt{\log(d/\gamma)}}{\sqrt{m}}  f(\log(d/\gamma)) .
	\end{align*}
	Applying a union bound, we get the following: with probability $1 - \gamma$,
	\begin{align*}
		\|\E_S[X] - \mu \|_ \infty \lesssim \frac{ \sqrt{\log(d/\gamma)}}{\sqrt{m}}  f(\log(d/\gamma))  .
	\end{align*}

\end{proof}

\noindent Using the above result, we are now ready to prove the concentration result that was required in \Cref{sec:sampling}.

\BadicLinfConscFull*
\begin{proof}
The second part follows from Part 3 of \Cref{lem:ConcTensorInf} and the fact that $\|x\|_{2,k} \leq \sqrt{k} \|x\|_{\infty}$. We now show the first part.
We can safely assume that  $\delta \leq 1$.
    Let $S := \{X_1, \dots, X_m\}$. The goal is to bound the following: 
    \[\Norm{\E_{X \sim S}\Brac{( X-\ovl \mu)^{\otimes t}}  - \E_{X \sim D}\Brac{(X-\mu)^{\otimes t}} }_{\infty}.\]
    We first add and subtract $\mu$ in the first term. To prove our lemma, we will bound each entry indexed by an ordered tuple  $T \in [d]^t$ of the resulting tensor. 
We will use the following guarantees on our samples: (i) $	\left\| \E_S \Brac{(X - \mu)^{\otimes t}} -\E_{X \sim D} \Brac{(X - \mu)^{\otimes t}} \right\|_\infty \leq \delta_1$,  (ii) 
$\max_{i \in d} \max_{r \leq t } (\E_S|X_i - \mu_i|^r)^{1/r} \leq \delta_2$, and (iii) $\| \mu - \ovl \mu\|_\infty \leq \delta_3$, which appear in \Cref{lem:ConcTensorInf}, for some values of $\delta_1,\delta_2, \delta_3$ to be defined later.
We begin with the following decomposition:
\begin{align}
	\nonumber&\Abs{ \E_{X \sim S}\Brac{(X- \mu + \mu -\ovl \mu)^{\otimes t}}_T - \E_{X \sim D}\Brac{(X-\mu)^{\otimes t}}_T } \\
	&\qquad\qquad= \Abs{ \E_{X \sim S}\Brac{\prod_{q \in T}  (X- \mu + \mu -\ovl \mu)_{q}}  - \E_{X \sim D}\Brac{(X-\mu)^{\otimes t}}_T }.
	\label{eq:ConcOrigExp}
\end{align}
We can expand $\prod_{q \in T} (X- \mu + \mu -\ovl \mu)_{q} = \sum_{Q \subseteq T} \prod_{q \in Q} (X- \mu)_{q} \prod_{q \in T \setminus Q} (\mu -\ovl \mu)_{q} = \prod_{q \in T} (X-\mu)_q + \sum_{Q \subsetneq T} \prod_{q \in Q} (X- \mu)_{q} \prod_{q \in T \setminus Q} (\mu -\ovl \mu)_{q}$, and apply the triangle inequality to get
\begin{align}
	\Bigg| \E_{X \sim S}\Brac{(X -\ovl \mu)^{\otimes t}}_T - \E_{X \sim D}\Brac{ (X-\mu)^{\otimes t}}_T \Bigg|  &\leq \Norm{\E_{X \sim S}\Brac{( X- \mu)^{\otimes t}} - \E_{X \sim D}\Brac{(X-\mu)^{\otimes t}} }_{\infty} \nonumber \\
	\label{eq:ConcTrianIneq} & \quad \quad+ \Abs{\E_{X \sim S} \Brac{\sum_{Q \subsetneq T} \Brac{  \prod_{q \in Q} (X- \mu)_{q} \prod_{q \in T \setminus Q} (\mu -\ovl \mu)_{q}}}}.
\end{align}
By assumption, the first term is upper bounded by $\delta_1$.
We will now focus on the second term.  For a particular $Q \subsetneq T$,  we get the following using Holder's inequality:
\begin{align*}
	\Bigg|\E_{X \sim S}  \Brac{  \prod_{q \in Q} (X- \mu)_{q} \prod_{q \in T \setminus Q} (\mu -\ovl \mu)_{q}}\Bigg| 
	& \leq 	 \|\mu - \ovl \mu\|_{\infty}^{|T \setminus Q|} \E_{X \sim S} \Brac{  \prod_{q \in Q} |(X- \mu)_q|} \\
	&\leq (\delta_3)^{|T \setminus Q|} \prod_{q \in Q} \Brac{ \E_{X \sim S}    |(X- \mu)_q|^{|Q|} }^{\frac{1}{|Q|}} \\
	&\leq \delta_3^{|T \setminus Q|} \delta_2^{|Q|}.
\end{align*}
Using the fact that $|\{Q: Q \subsetneq T\}| \leq 2^t$ and $|T\setminus Q| \geq 1$, we get the following bound on the second term in \eqref{eq:ConcTrianIneq},
\begin{align*}
	\Abs{\E_{X \sim S} \Brac{\sum_{Q \subsetneq T} \Brac{  \prod_{q \in Q} (X- \mu)_{q} \prod_{q \in T \setminus Q} (\mu -\ovl \mu)_{q}}}} &\leq 2^t \delta_3  \max(1, \delta_2^{t-1}, \delta_3^{t-1}). 
\end{align*}
This leads to the following bound on the expression in \eqref{eq:ConcOrigExp}:
\begin{align}
	&\Abs{ \E_{X \sim S}\Brac{(X- \mu + \mu -\ovl \mu)^{\otimes t}}_T - \E_{X \sim D}\Brac{(X-\mu)^{\otimes t}}_T } \leq \delta_1 + \delta_32^t    \max(1, \delta_2^{t-1}, \delta_3^{t-1}). 
\label{eq:ConcOrigExp2}
\end{align}
We can choose $\delta_1 = \delta/2$, $\delta_2 = 2f(t) $ and $\delta_3 = 2^{-t}\max(1, \delta_2)^{-t+1} \delta/2$, we get that the expression in \eqref{eq:ConcOrigExp2} is upper bounded by $\delta$ by noting that $\delta_3 \leq 1$  (since $\delta \leq 1$) and $2^t\delta_3\max(1,\delta_2,\delta_3)^{t-1} \leq 2^t \delta_3 \max(1, \delta_2)^{-t+1}) \leq (\delta/2)$.
By \Cref{lem:ConcTensorInf}, we get that the total sample complexity is at most 
\begin{align}
	m &= \frac{1}{\delta^2}   
	\left( t \log (d/\gamma)\right)\left(Cf(t^2\log(d/\gamma))\right)^{2t}	\max\left( 1 , \frac{1}{f(t)^{2t}}\right),
	\label{eq:linftyconcbasecase}
\end{align}
for a sufficiently large constant $C > 0$, where we performed the following crude upper bounds on the sample complexity guarantee in \Cref{lem:ConcTensorInf} for the ease of presentation:
\begin{align*}
    &\frac{1}{\delta^2}(t \log(d/\gamma))f(t^2\log d/\gamma)  + \log(d/\gamma) \left(\frac{f(\log(d/\gamma))}{f(t)}\right) ^{2t} \\
    &\qquad \qquad + \frac{1}{\delta^2} (\log(d/\gamma)) f(\log d/\gamma)^2 2^{8t}( \max(1,  2f(t)  )^{2t-2} \\
& \leq (t \log d/\gamma) (10f(t^2 \log d/\gamma))^{2t} \Big( \frac{1}{\delta^2}  + \frac{1}{f(t)^{2t}} + \frac{1}{\delta^2}\max\left( 1, \frac{1}{f(t)}  \right)^{2t-2}\Big)\\
& \leq \frac{1}{\delta^2}(t \log d/\gamma) (10f(t^2 \log d/\gamma))^{2t} \Big(1  + \frac{1}{f(t)^{2t}} + \max\left( 1, \frac{1}{f(t)}  \right)^{2t-2}\Big)\\%
& \leq 3\frac{1}{\delta^2}(t \log d/\gamma) (10f(t^2 \log d/\gamma))^{2t} \max\Big(1, \frac{1}{f(t)^{2t}}, \frac{1}{f(t)  ^{2t-2}}\Big)\\
& =3 \frac{1}{\delta^2}(t \log d/\gamma) (10f(t^2 \log d/\gamma))^{2t} \max\Big(1, \frac{1}{f(t)^{2t}}\Big).
\end{align*}

\end{proof}

%% file: app-bdd-moments.tex
\section{Omitted Proof from \Cref{sec:sparse-mean-est}} \label{sec:omittedsec4}
We provide the proof of \Cref{claim:simplecheck} below that was omitted from  \Cref{sec:sparse-mean-est}.
\ClaimSimpleCheck*

	\begin{proof} 
		Since $r_i = \mathbf{1}_{X_i = Y_i}$,  then $\sum_i r_i = (1-\epsilon) m, r_i^2 = r_i$ and $r_i(X_i - Y_i) = 0$ for all $i \in [m]$. We see that 
		\begin{enumerate}
			\item $W_i^2 = w_i^2 r_i^2 = w_i r_i = W_i$. 
			\item $\cA_{\textnormal{corruptions}} \sststile{}{} W_i (X_i - X'_i) = w_i r_i (X_i - Y_i + Y_i - X'_i) = r_i w_i (Y_i - X_i) + w_i r_i (Y_i - X_i) = 0$.
			\item Additionally, since $\{W_i^2 = W_i\} \sststile{}{} (1-W_i)^2 = 1-2W_i + W_i^2 = 1-W_i$,
			and using the fact that $\{x^2 = x\} \sststile{O(1)}{} \{x > 0, x < 1\}$, we see
			\begin{align*}
				\cA_{\textnormal{corruptions}} \sststile{O(1)}{} 1-W_i \leq 2(1-W_i) = 1-w_ir_i + 1-w_ir_i \leq (1-w_i) + (1-r_i).
			\end{align*} 
			A sum over $i \in [m]$ gives us $\sststile{O(1)}{} \sum_i (1-W_i) \leq 2\epsilon m$. 
		\end{enumerate}
		
	\end{proof}

%% file: app-gaussian-o-eps-error.tex
\section{Omitted Proofs from \Cref{sec:Gaussian-O-eps}}
\label{app:gaussian-o-eps}
 
This section contains the omitted details from \Cref{sec:Gaussian-O-eps}. We begin by proving that inliers satisfy deterministic conditions with high probability in \Cref{app:gaussian-resilience}. We prove the proofs of estimation lemmata (\Cref{lem:specialized_mean,lem:specialized_cov}) in \Cref{sec:estimation_lemmata}. Remaining technical details are provided in \Cref{sec:last_part_of_proof}.

\subsection{Deterministic Conditions on Inliers}
\label{app:gaussian-resilience}

In this section, we prove that the deterministic conditions required in \Cref{sec:resilience_main} hold with high probability. In particular, we provide the proofs of \Cref{lem:resilience_integrated,lem:similar-tensor-bound}. 
\subsubsection{Proof of \Cref{lem:resilience_integrated}}
We prove \Cref{lem:resilience_integrated} in this section. To this end, we need the  series of lemmata below.

\begin{lemma}[\cite{li18thesis}]\label{lem:jerry-resilience}
 Let $k\in \Z_+$ with $k\leq d$, $0<\eps\leq 1/2$ and $0<\gamma<1$. Let $X_1,\ldots, X_m \sim \cN(0,\Sigma)$. There exists an absolute constant $C$ such that, if
\begin{align*}
m> C\frac{\min(d,k^2) + \log\binom{d^2}{k^2}+\log(1/\gamma)}{\eps^2 \log(1/\eps)}
\end{align*} 
   then, with probability at least $1-\gamma$, we have that for any choice of weights $a_i \in [0,1]$ with $\E_{i \sim [m]}[a_i]\geq 1-\eps$,  the following two inequalities hold for all vectors $v \in \cU_k$:
   \begin{enumerate}
   \item $\left| \E_{i \sim [m]}[a_i \iprod{v,X_i}]  \right| \leq O\left(\eps \sqrt{\log(1/\eps)} \right)\sqrt{v^T \Sigma v}$.
   \item $\left| \E_{i \sim [m]}\left[ a_i\iprod{v,X_i}^2  \right]  - v^T \Sigma v  \right| \leq O(\eps \log(1/\eps) )v^T \Sigma v$.
   \end{enumerate} 
\end{lemma}

\noindent The result in  \cite{li18thesis} is for $\Sigma = I_d$. This version follows by taking a union bound over the support and re-normalizing the distribution. We also require a similar  property for the fourth moment of the inliers.
\begin{restatable}{lemma}{LemResilienceFourth}
     \label{lem:resilience-4moment}
	Let $k\in \Z_+$ with $k\leq d$, $0<\eps\leq 1/2$,   $0<\gamma<1$.
	Let $m> C({k^4}/{\eps^2})\log^4(d/(\eps\gamma))$ for a sufficiently large constant $C$ and  $X_1,\ldots, X_m \sim \cN(0,\Sigma)$. Then, with probability at least $1-\gamma$, for any weights $a_i \in [0,1]$ with $\E_{i \sim [m]}[a_i]\geq 1-\eps$ it holds
	\begin{align*}
		\left| \E_{i \sim [m]}\left[  a_i\left(\left( \iprod{v,X_i}^2- v^T \Sigma v \right)^2 - 2(v^T \Sigma v)^2  \right) \right]  \right| \leq \tilde{O}(\eps)(v^T \Sigma v)^2
	\end{align*}
	for all vectors $v \in \cU_k$.
\end{restatable}

\begin{proof}
	We first show the condition in the case where there are no weights ($a_i=1$, for all $i \in [m])$.
	\begin{lemma} \label{eq:bound_for_entire_set}
		If $m >C (k^4/\eps^2)\log^4(d/\gamma)$ for a sufficiently large constant $C$, then a set of $m$ samples from $\cN(0,\Sigma)$ for $I_d \preceq \Sigma \preceq 2 I_d$, with probability at least $1-\gamma$,  satisfies least $1-\gamma$
		\begin{align*}
			\left| \E_{i \sim [m]}\left[ \left( \iprod{v,X_i}^2- v^T \Sigma v \right)^2 - 2(v^T \Sigma v)^2   \right]  \right| \leq {O}(\eps)(v^T \Sigma v)^2
		\end{align*}
		for all $v \in \cU_k$. %
	\end{lemma}
	\begin{proof}
		We want to show concentration of polynomials of the form $(\iprod{v,x}^2 - v^T \Sigma v)^2$ for $v$, a $k$-sparse vector. Let $S$ be a set of $m$ samples from $\cN(0,\Sigma)$. First, let $u = (v v^T)^{\flat}$ (i.e., the vector having as elements all the products $v_i v_j$). This is a $k^2$-sparse vector. Define $M$ as the $d^2 \times d^2$ matrix with $M_{(ij),(k \ell)} = \E_{X \sim S}[(X_i X_j - \Sigma_{ij})(X_k X_\ell - \Sigma_{k \ell})] - 2\Sigma_{ij}\Sigma_{k \ell}$ for all $i,j,k,\ell \in [d]$. We note the rewriting:
		\begin{align*}
			&\E_{X \sim S}[(\iprod{v,X}^2 - v^T \Sigma v)^2] - 2(v^T\Sigma v)^2 \\
			&= \E_{X \sim S}\Big[\Big(\sum_{i,j \in [d]} v_i v_j(X_i X_j - \Sigma_{ij})\Big)^2 \Big] - 2 \Big(\sum_{i,j \in [d]} v_i v_j \Sigma_{ij} \Big)^2   \\
			&= \E_{X \sim S}\Big[\sum_{i,j \in [d]} v_i v_j(X_i X_j - \Sigma_{ij}) \sum_{k,\ell \in [d]} v_k v_\ell(X_k X_\ell - \Sigma_{k \ell})  
			 - 2 \sum_{i,j \in [d]} v_i v_j \Sigma_{ij}  \sum_{k,\ell \in [d]} v_k v_\ell \Sigma_{ij} \Big] \\
			&= \E_{X \sim S}\Big[\sum_{i,j \in [d]} u_{ij}(X_i X_j - \Sigma_{ij})\sum_{k,\ell \in [d]} u_{k \ell}(X_k X_\ell - \Sigma_{k \ell})  - 2 \sum_{i,j \in [d]} u_{ij} \Sigma_{ij}\sum_{k,\ell \in [d]}u_{k \ell} \Sigma_{ij}  \Big] \\
			&= u^T M u
		\end{align*}
		Hence, it is sufficient to show that $u^T M u \leq \tilde{O}(\eps)$ for all $u \in \cU_{k^2}(d^2)$. For a $Q \subset [d^2]$, we denote by $M_Q$ the $Q \times Q$ submatrix of $M$. We have that
		\begin{align*}
			\sup_{u \in \cU_{k^2}(d^2)} u^T M u = \sup_{|Q| \leq k^2} \| M_Q \|_2
			\leq \sup_{|Q| \leq k^2} \| M_Q \|_F \;,
		\end{align*}
		Thus, it suffices for every element of $M_Q$ to be $O(\eps)$, which holds if %
\begin{align} \label{eq:end_goal}
			\Abs{\E_{X \sim S}[p(x)] - \E_{X \sim \cN(0,\Sigma)}[p(x)]} \leq \frac{\eps}{k^2} \;,
		\end{align}
		for the polynomial $p(x):= (x_ix_j-\Sigma_{ij})(x_k x_\ell-\Sigma_{k \ell}) -2\Sigma_{i j}\Sigma_{k \ell}$.
		To this end, we use the following concentration inequality, which is a consequence of Gaussian Hypercontractivity:
		\begin{fact}[see, e.g., Corollary 5.49 in \cite{aubrun2017alice}]\label{fact:gaussianPoly}
		    Let $Z_1, \dots, Z_m$ be independent $\mathcal N(0, 1)$ variables and let $X = h(Z_1, \dots, Z_m)$, where $h$ is a polynomial of total degree at most $q$. Then, for any $t \geq (2e)^{q/2}$, 
		    \[ \Pr\Brac{ \Abs{X - \E[X] } \geq t \sqrt{\Var[X]}} \leq \exp\Paren{-\frac{q}{2e} t^{2/q}}.\]
		\end{fact}
		Note that we can still apply this lemma for polynomials $h(Z')$ of Gaussians $Z' \sim \mathcal{N}(0, \Sigma)$ with covariance $\Sigma \neq I$ by noting that  $Z' =\sqrt{\Sigma} Z$ where $Z \sim \cN(0, I)$ and replacing $h(Z_1', \dots, Z_m')$ by $h'(Z_1, \dots, Z_m) = h(Z_1', \dots, Z_m') = h( (\sqrt{\Sigma} Z)_1, \dots, (\sqrt{\Sigma} Z)_m)$ in \Cref{fact:gaussianPoly}.
		
		We apply the above to the appropriate degree $q=4$ polynomial of \Cref{eq:end_goal}, i.e., $h(X_1,\ldots, X_m) = \frac{1}{m}\sum_{i=1}^m(p(X_i) - \E_{X \sim \cN(0,\Sigma)}[p(X)])$. We note that in our case $\Var[h(X)] = \Var[p(X)]/m = O(1/m)$. $\Var[p(X)]$ is bounded by a constant since it is a degree $4$ polynomial with constant coefficients of Gaussian variables with bounded covariance.  We thus obtain that for $m > C (k^4/\eps^2) \log^4(1/\gamma')$ samples \Cref{eq:end_goal} holds with probability $1-\gamma'$. Using $\gamma' = \gamma/d^4$ and a union bound over $(i,j,k,\ell)$ yields the final sample complexity. We have thus shown  \Cref{eq:bound_for_entire_set} with $O(\eps)$ in the RHS. Assuming $I_d \preceq \Sigma$, this implies the final claim.

	\end{proof}

	Having \Cref{eq:bound_for_entire_set} in hand, we use it to complete the proof of \Cref{lem:resilience-4moment}. By convexity, it suffices to assume $a_i \in \{0,1\}$. Let $I$ be the set of indices such that $a_i=1$.
	For a given $v \in \cU_k$, define $p_v(x) = (\iprod{v,x}^2-v^T \Sigma v)$. Let $J^*$ be the set of $2\eps m$ indices with greatest $(p_v^2(X_i) - 2(v^T \Sigma v)^2)$ and define $J_1^* = \{i: p^2_v(X_i) \geq c \log^2(1/\eps) (v^T\Sigma v)^2 \}$. If $m > C \log(1/\gamma')/\eps^2$, with probability $1-\gamma'$, we have the following:
	\begin{enumerate}
		\item $\| X_i \|_2 = O(\sqrt{d \log(m/\gamma')})$ for all $i \in [m]$. \label{it:fact0}
		\item $\frac{1}{m}|\{ p^2(X_i) >c \log^2(1/\eps)(v^T \Sigma v)^2  \}| \leq 2\eps$. \label{it:fact1}
		\item $J_1^* \subseteq J^*$.\label{it:fact2}
		\item $\left|\frac{1}{m} \sum_{i \not\in J_1^*} (p^2(X_i) - 2(v^T \Sigma v)^2) \right| = O(\eps \log^2(1/\eps) ) (v^T \Sigma v)^2$.\label{it:fact3}
	\end{enumerate}
	The above claims can be shown like in Appendix B.1 of \cite{DKKLMS16} (see Equations (44), (45), (46) of the first arxiv version of that paper; Concretely, the second item follows from \Cref{fact:gaussianPoly}, the third follows from the second, and the last is shown in Claim B.4 of \cite{DKKLMS16}).
	
	Fix any $I \subseteq [m]$ with $|I|=(1-2\eps)m$. Partition $[m] \setminus I$ into $J^+ \cup J^{-}$, where $J^+ = \{ i \not\in I : p^2(X_i) \geq 2(v^T \Sigma v)^2 \}$, 
	and $J^- = \{ i \not\in I : p^2(X_i) < 2(v^T \Sigma v)^2 \}$. We will show that $(1/|I|)|\sum_{i \in I}(p^2_v(X_i) - 2(v^T \Sigma v)^2)| = \tilde{O}(\eps )(v^T \Sigma v)^2$. We first show the upper bound
	\begin{align*}
		\frac{1}{|I|} \sum_{i \in I}(p^2_v(X_i) - 2(v^T \Sigma v)^2) &\leq 
		\frac{1}{|I|} \sum_{i \in I \cup  J^+}(p^2_v(X_i) - 2(v^T \Sigma v)^2) - \frac{1}{|I|} \sum_{i \in J^-}(p^2_v(X_i) - 2(v^T \Sigma v)^2) \\
		&\leq \left| \frac{1}{|I|} \sum_{i=1}^m (p^2_v(X_i) - 2(v^T \Sigma v)^2)   \right| + 2\frac{1}{|I|} \left|  \sum_{i \in J^-}^m (p^2_v(X_i) - 2(v^T \Sigma v)^2) \right|\\
		&\leq {O}(\eps)(v^T \Sigma v)^2 + \frac{|J^-|}{|I|}O((v^T \Sigma v)^2) \\
		&=  O(\eps)(v^T \Sigma v)^2 \;,
	\end{align*}
	where we used \Cref{eq:bound_for_entire_set}. We now focus on the other direction. %
	We note that the lower bound is achieved when $I  = [m] \setminus J^*$. Thus  we obtain the following using \Cref{it:fact3,it:fact2}:
		\begin{align*}
		&\frac{1}{|I|} \sum_{i \in I}(p^2_v(X_i) - 2(v^T \Sigma v)^2)  \\
		&\geq \frac{1}{(1-2 \eps)m }  \Paren{\sum_{i \in m}(p^2_v(X_i) - 2(v^T \Sigma v)^2) -  \sum_{i \in J^*} (p^2_v(X_i) - 2(v^T \Sigma v)^2) } \\
		&= \frac{1}{(1-2 \eps)m }  \Paren{ \sum_{i \not \in J_1^* }(p^2_v(X_i) - 2(v^T \Sigma v)^2)  - \sum_{i \in J \setminus J_1^*} (p^2_v(X_i) - 2(v^T \Sigma v)^2)} \\
		&\geq  - \left|\frac{O(1)}{m }  \Paren{ \sum_{i \not \in J_1^* }(p^2_v(X_i) - 2(v^T \Sigma v)^2) }\right|  - \frac{ O(1)}{m} \sum_{i \in J \setminus J_1^*} p^2_v(X_i)  \\
&\geq - O(\eps \log^2(1/\eps) (v^T\Sigma v)^2)  - \frac{O(1)|J|}{m} c \log^2(1/\eps) (v^T\Sigma v)^2 \\
		&\geq - {O}(\eps \log^2(1/\eps) )(v^T \Sigma v)^2 \;.
	\end{align*}
	Note that this holds for a fixed $v$, and all subsets $I$ with $|I|=(1-2\eps)m$.  The last step is a cover argument. To that end, we first state that the desired expression is Lipschitz with respect to $v$:  
	\begin{claim} \label{claim:relate_p}
		Conditioned on the event of \Cref{it:fact0}, for any unit-norm $u,v \in \R^d$ and $i \in [m]$ we have that $|p_v(X_i)^2 - 2(v^T \Sigma v)^2 - (p_u(X_i)^2 - 2(u^T \Sigma u)^2)| \lesssim \|u-v\|_2 (R^2 + \| \Sigma \|^2_2)^2$, where $R=O(\sqrt{d \log(m/\gamma)})$.
	\end{claim}
	\begin{proof}
		We first claim the following for the difference between the polynomials without the squares:
		\begin{align*}
			| p_v(X_i) -  p_u(X_i)| &\leq |\iprod{v,X_i}^2 - \iprod{u,X_i}^2| + | v^T \Sigma v - u^T \Sigma u|
			\leq 2 \|v-u\|_2 (R^2 + \|\Sigma\|_2) \;.
		\end{align*}
		The first line uses the triangle inequality. For the second term of the following line we use that $| v^T \Sigma v - u^T \Sigma u| = |v^T \Sigma (v-u) - (u-v)^T \Sigma u| \leq \|\Sigma\|_2(\|v\|_2 \|v-u \|_2 + \|u\|_2 \|v-u \|_2) \lesssim \| \Sigma \|_2 \| u- v\|_2$. The second term is bounded using the same argument but with $X_i (X_i)^T$ in place of $\Sigma$ and using that $\|X_i\|_2 = O(R)$.
		
		We can now complete our proof.
		\begin{align*}
			|p_v(X_i)^2 - &2(v^T \Sigma v)^2 - (p_u(X_i)^2 - 2(u^T \Sigma u))| \leq |p_v(X_i)^2 - p_u(X_i)^2| + 2 | (v^T \Sigma v)^2 + (u^T \Sigma u)^2| \\
			&\lesssim \max\{|p_v(X_i)|,|p_u(X_i)| \} | p_v(X_i) -  p_u(X_i)| + \max\{v^T \Sigma v,u^T \Sigma u \}  | v^T \Sigma v - u^T \Sigma u| \\
			&\lesssim \|v-u\|_2 (R^2 + \|\Sigma\|_2)^2 \;,
		\end{align*}
		where the first line uses the triangle inequality, the second line uses $|a^2-b^2| \leq 2\max\{|a|,|b| \}|a-b|$, and the last one uses the bound $|p_v(X_i)| \leq \|X_i\|^2 + \|\Sigma\|_2 \leq R^2 + \|\Sigma\|_2$. 
	\end{proof}
	Recalling that $\|\Sigma\|_2 \leq 2$, the RHS of \Cref{claim:relate_p} is essentially $\|v-u\|_2 R^4$. In order for it to become $O(\eps)$, we would like our cover $S$ of $k$-sparse unit vectors to have accuracy $O(\eps/R^4)$, which results in a set of size  $|S| = \binom{d}{k} O(R^4/\eps)^k$ vectors. We thus choose the probability of failure $\gamma' = \gamma/|S|$, which means that we  need 
	\begin{align*}
		m > C \frac{\log(|S|/\gamma)}{\eps^2} = \frac{\log \binom{d}{k} + k\log(d \log(m/\gamma)/\eps) + \log(1/\gamma)}{\eps^2} \;.
	\end{align*}
	The sample complexity of \Cref{eq:bound_for_entire_set} scales with $k^4$ and dominates the sample complexity of this paragraph.
	This completes the proof of \Cref{lem:resilience-4moment}.
\end{proof}
We now have all the ingredients to prove \Cref{lem:resilience_integrated}, which we restate below for convenience.

\LemResilienceIntegrated*
\begin{proof}
	Without loss of generality, we assume $\mu=0$. We let $Z_i = X_i - \overline{\mu}$. We condition on the events of \Cref{lem:jerry-resilience,lem:resilience-4moment}. For simplicity, we also assume that the $a'_i$'s are scaled so that $\E_{ij}[a_{ij}]=\E_{i}[a'_i] = 1$ (since this consists of only a scaling of $1+O(\eps)$, it is without loss of generality). We show the individual claims below:
	\begin{enumerate}
		\item Proof of $| \iprod{v,\overline{\mu}}| \leq \tilde{O}(\eps) \sqrt{v^T \Sigma v}$: \label{it:proof1}
			Use \Cref{lem:jerry-resilience} with $a_i=1$.
		\item Proof of $| \E_{i \sim [m]}[a'_i\iprod{v,X_i-\overline{\mu}}]   | \leq \tilde{O}(\eps) \sqrt{v^T \overline{\Sigma} v}$: 
			By \Cref{it:proof1} and \Cref{lem:jerry-resilience}, we have that
			\begin{align*}
				\left| \E_{i \sim [m]}[a'_i\iprod{v,X_i - \overline{\mu}}]  \right| \leq \left| \E_{i\sim [m]}[a'_i\iprod{v,X_i}] \right| + \left| \E_{i\sim [m]}[a'_i\iprod{v,\overline{\mu}}]  \right|
				\leq  \tilde{O}(\eps) \sqrt{v^T \Sigma v} 
				\leq \frac{\tilde{O}(\eps) \sqrt{v^T\overline{\Sigma} v}}{\sqrt{1-\tilde{O}(\eps)}} \;,
			\end{align*}
			where the last inequality uses \Cref{it:res4}.
		\item Proof of $\left| \E_{i \sim [m]}\left[a'_i \left(\iprod{v,X_i-\overline{\mu}}^2    -  v^T \overline{\Sigma} v \right)  \right] \right| \leq \tilde{O}(\eps) v^T \overline{\Sigma} v $: We have the following inequalities.
			\begin{align*}
				&\left| \E_{i \sim [m]}\left[a'_i \left(\iprod{v,X_i-\overline{\mu}}^2    -  v^T \overline{\Sigma} v \right)  \right] \right|\\
				&= \left| \E_{i \sim [m]}\left[a'_i\iprod{v,X_i}^2 + a'_i\iprod{v,\overline{\mu}}^2 - 2a'_i \iprod{v,X_i}\iprod{v,\overline{\mu}} - a'_i v^T \overline{\Sigma} v  \right] \right| \\
				&\leq \left| (1 \pm \tilde{O}(\eps))v^T \Sigma v +  \tilde{O}(\eps^2) v^T \Sigma v \pm 2 \E_{i \sim [m]}[a'_i \iprod{v,X_i}]\tilde{O}(\eps) \sqrt{v^T \Sigma v} -  v^T\overline{\Sigma} v \right| \\
				&= |v^T(\overline{\Sigma} - \Sigma)v| + \tilde{O}(\eps) v^T \Sigma v + \tilde{O}(\eps^2) v^T \Sigma v \\
				&= \tilde{O}(\eps) v^T \overline{\Sigma} v \;,
			\end{align*}
			where the second line uses \Cref{lem:jerry-resilience} along with \Cref{it:res1} and $\E_i [a_i] = 1$, the next line uses \Cref{lem:jerry-resilience} and the last line relates $v^T \Sigma v$ to $v^T \overline{\Sigma} v$  using  \Cref{it:res4}.
		\item Proof of $|v^T(\overline{\Sigma} - \Sigma)v  | \leq \tilde{O}(\eps) v^T \Sigma v$:
			Repeating the steps from the proof of \Cref{it:res3}, we can show $\left| \E_{i \sim [m]}\left[a'_i \left(\iprod{v,X_i-\overline{\mu}}^2    -  v^T \Sigma v \right)  \right] \right| \leq \tilde{O}(\eps) v^T \Sigma v$. Taking $a'_i=1$ gives \Cref{it:res4}.
		\item Proof of $|\E_{i,j \sim [m]}[a_{ij}(v^TX_{ij} v - v^T\overline{\Sigma}v)  ] | \leq \tilde{O}(\eps) v^T \overline{\Sigma} v$:
			We clarify that we will often use the following: Whenever we have a term of the form $\E_{ij}[a_{ij} b_i]$ with $b_i\geq 0$, we will use that $\E_{ij}[a_{ij} b_i] \leq  \E_{i}[a'_i b_i]$ (since $0<a_{ij}\leq a'_i$).
			\begin{align*}
				 \Big|\E_{i,j \sim [m]}&\Brac{a_{ij}(v^TX_{ij}v - v^T\overline{\Sigma}v)  } \Big| \\ 
				&= \left|\E_{i,j \sim [m]}[a_{ij}\frac{1}{2}\iprod{v,X_i-\overline{\mu} -(X_j-\overline{\mu})}^2 - v^T\overline{\Sigma}v  ] \right| \\ 
				&\lesssim \left| \E_{i}[a'_i \frac{1}{2}\iprod{v,X_i-\overline{\mu}}^2] - \frac{1}{2} v^T \overline{\Sigma} v \right| 
				+ \left| \E_{j}\Brac{a'_j \frac{1}{2}\iprod{v,X_j-\overline{\mu}}^2}- \frac{1}{2} v^T \overline{\Sigma} v \right| \\
				&+  |\E_{ij}\Brac{a_{ij}\iprod{v,X_i-\overline{\mu}}\iprod{v,X_j-\overline{\mu}}} | \\
				&\leq \tilde{O}(\eps) v^T \overline{\Sigma} v \;,
			\end{align*}
			since each of the first two terms is $\tilde{O}(\eps)$ using \Cref{it:res2,it:res3,it:res4} and the same holds for the last term: Using Cauchy–Schwarz, this term becomes $|\E_{ij}[a_{ij}\iprod{v,X_i-\overline{\mu}}\iprod{v,X_j-\overline{\mu}}] |  \leq \sqrt{|\E_{ij}[ a_{ij} \iprod{v,X_i-\overline{\mu}}^2] \E_{ij}[a_{ij} \iprod{v,X_j-\overline{\mu}}^2] |}$ and then applying again \Cref{it:res2,it:res3,it:res4} bounds it by $\tilde{O}(\eps) v^T \overline{\Sigma} v$.
		\item Proof of $|\E_{i,j \sim [m]}[a_{ij}((v^TX_{ij}v-v^T\overline{\Sigma} v)^2 -        2(v^T\Sigma v)^2 )]  | \leq \tilde{O}(\eps) (v^T \overline{\Sigma} v)^2$:
			Using that $\E_{i,j}[a_{ij}]=1$ and some algebraic manipulations, we have that:
			\begin{align} 
				&\left|\E_{i,j \sim [m]}[a_{ij}((v^TX_{ij}v-v^T\overline{\Sigma} v)^2 -        2(v^T\Sigma v)^2 )]  \right| \notag \\
				&= \left|\E_{i,j \sim [m]}\left[a_{ij}\left(\frac{\iprod{v,X_i-X_j}^2-2v^T\overline{\Sigma} v}{2}\right)^2\right] -        2(v^T\Sigma v)^2  \right|  \notag\\
				&=\left| \frac{1}{4}\E_{i,j}\left[a_{ij}\left(\left(  \iprod{v,X_i}^2-v^T\overline{\Sigma} v + \iprod{v,X_j}^2-v^T\overline{\Sigma} v {-} 2\iprod{v,X_i}\iprod{v,X_j} \right)^2  \right)  \right] -        2(v^T\Sigma v)^2 \right| \notag \notag \\
				&= \left| \frac{1}{4}\E_{ij \sim [m]}\left[a_{ij}\left( A^2 + B^2 + C^2 + 2AB + 2 AC + 2 BC \right)\right] - 2(v^T\overline{\Sigma}v)^2  \pm \tilde{O}(\eps)     \right| \;, \label{eq:bunch_of_triangle_ineqs}
			\end{align}
			where we replaced $(v^T{\Sigma}v)^2$ by $(v^T\overline{\Sigma}v)^2 \pm \tilde{O}(\eps)$ using \Cref{it:res4} and the fact that $I \preceq \Sigma \preceq 2I$ and we let $A:= \iprod{v,X_i}^2-v^T\overline{\Sigma} v$, $B:= \iprod{v,X_j}^2-v^T\overline{\Sigma} v$, $C:= 2\iprod{v,X_i}\iprod{v,X_j}$. We work with each term individually. We have that
			\begin{align*}
				\E_{ij \sim [m]}[A^2] &= \E_{i\sim [m]}[a'_i(\iprod{v,X_i}^2-v^T \overline{\Sigma} v)^2] \notag \\
				&= \E_{i \sim [m]}[a'_i(\iprod{v,X_i}^2-v^T\Sigma v + v^T(\Sigma-\overline{\Sigma})v)^2] \notag  \\
				&\leq \E_{i \sim [m]}[a'_i(\iprod{v,X_i}^2-v^T\Sigma v)^2] 
				+ \E_{i \sim [m]}[a'_i](v^T(\Sigma-\overline{\Sigma})v)^2\\
				& \qquad + 2 \E_{i \sim [m]}[a'_i(\iprod{v,X_i}^2-v^T\Sigma v)] (v^T(\Sigma-\overline{\Sigma})v) \notag \\
				&\leq 2(v^T \Sigma v)^2+ \tilde{O}(\eps)(v^T \Sigma v)^2  \;, \notag
			\end{align*}
			since the first term is $(2+\tilde{O}(\eps))(v^T \Sigma v)^2$ and the other two $\tilde{O}(\eps)(v^T \Sigma v)^2$: the first term is bounded because of \Cref{lem:resilience-4moment}, the second because of \Cref{it:res4} and the last one because of both (an application of Cauchy-Schwarz is needed there).
			Similarly, we have that $\E_{ij \sim [m]}[B^2]\leq (2 \pm \tilde{O}(\eps))(v^T \Sigma v)^2$. Furthermore, $\E_{ij \sim [m]}[C^2]=(4\pm\tilde{O}(\eps))(v^T \Sigma v)^2$ and that all cross terms involving $AB$, $AC$, $BC$ are $\tilde{O}(\eps)$. Putting these together we get that the right-hand side of \Cref{eq:bunch_of_triangle_ineqs} is $\tilde{O}(\eps)$. Using \Cref{lem:resilience-4moment} one last time we have that this is at most $\tilde{O}(\eps)(v^T \overline{\Sigma }v)^2$.
	\end{enumerate}
	This completes the proof of \Cref{lem:resilience_integrated}.
	\end{proof}

\subsubsection{Proof of \Cref{lem:similar-tensor-bound}} \label{sec:simple_tensor_concentration}

This section contains the proof of the following result that was used in \Cref{sec:resilience_main}.

\LemBoundSimilarTensor*

Our proof uses  the following standard concentration inequality and Isserlis' Theorem.

\begin{claim} \label{claim:cov_est}
	Let $X^{(1)},\ldots , X^{(N)} \sim \cN(0,\Sigma)$ and $\Sigma_N:= \frac{1}{N} \sum_{i=1}^N X^{(i)}X^{(j)^T}$. Let $\eps',\delta \in (0,1)$. Then, with probability at least $1-\delta$, it holds that $\|  \Sigma_N -  \Sigma \|_\infty \leq \eps' \|  \Sigma \|_2$ 
	provided that $N > C\log(d/\delta)/\eps'^2$ for a sufficiently large universal constant $C$.
\end{claim}
\begin{proof}
	For $k,\ell \in [d]$, the random variable $\frac{1}{N}\sum_{i=1}^d X^{(i)}_k X^{(i)}_\ell$ is sub-exponential with Orlicz norm $\|\frac{1}{N}X^{(i)}_k X^{(i)}_\ell  \|_{\psi_1} \leq (C/N)  \| \Sigma \|_2$ for some $C>0$. Therefore, by Bernstein's inequality, if $N$ is a large enough multiple of $ \log(1/\delta')/\eps'^2$, we have that
	\begin{align*}
		\Pr \left[ \Bigabs{\frac{1}{N}X^{(i)}_k X^{(i)}_\ell - \Sigma_{k\ell}} \leq \eps' \| \Sigma \|_2 \right] \leq \delta'\;.
	\end{align*}
	By using $\delta' = \delta/d^2$ and taking a union bound over all $d^2$ elements of the matrix, the result follows.
\end{proof}

\begin{fact}[Isserlis' Theorem] \label{fact:isserlis}
	Let $(X_1,\ldots, X_t) \sim \cN(0,\Sigma)$. Then, 
	\begin{align*}
		\E[X_1 \cdots X_t] = \sum_{p \in P_t^2} \prod_{\{i,j\} \in p} \E[X_i X_j] \;,
	\end{align*}
	where $P_t^2$ is the set of all pairings of $[t]$.
\end{fact}

We are now ready to prove \Cref{lem:similar-tensor-bound}. 
\begin{proof}[Proof of \Cref{lem:similar-tensor-bound}]
	We fix $\ell_1,\ldots, \ell_t \in [d]$ and examine the $(\ell_1,\ldots, \ell_t)$-th entry $\left(\E_{Y \sim \cN(0,\Sigma)}[Y^{\otimes t}]\right)_{\ell_1,\ldots, \ell_t} = \E_{Y \sim \cN(0,\Sigma)}[Y_{\ell_1}\ldots Y_{\ell_t}]$. Using \Cref{fact:isserlis}, we can write it as a sum of products of elements of $\Sigma$:
	\begin{align*}
		\left(\E_{Y \sim \cN(0,\Sigma)}[Y^{\otimes t}]\right)_{\ell_1,\ldots, \ell_t} = \sum_{p \in P_t^2} \prod_{\{i,j\} \in p} \Sigma_{\ell_i \ell_j} \;.
	\end{align*}
	We note that each product has at most $(t/2)$-many factors. The same decomposition holds for each entry of the tensor $ \E_{Y \sim \cN(0,\overline{\Sigma})}[Y^{\otimes t}]$ by replacing $\Sigma$ with $\overline{\Sigma}$. Therefore,
	\begin{align}
		\left\lvert \left(  \E_{Y \sim \cN(0,\Sigma)}[Y^{\otimes t}] - \E_{Y \sim \cN(0,\overline{\Sigma})}[Y^{\otimes t}] \right)_{\ell_1,\ldots, \ell_t} \right\rvert
		&\leq \sum_{p \in P_t^2} \Bigabs{ \prod_{\{i,j\} \in p} \Sigma_{\ell_i \ell_j} - \prod_{\{i,j\} \in p} \overline{\Sigma}_{\ell_i \ell_j} } \label{eq:tensorformula} \;.
	\end{align}
	We now focus on a single term of the right-hand side. Assuming that we have an approximation $\| \Sigma - \overline{\Sigma}  \|_\infty \leq \delta$ for some $\delta < 1$ to be defined later, we can write $\overline{\Sigma}_{ij} = \Sigma_{ij} + \delta_{ij}$ with $|\delta_{ij}| \leq \delta$. Plugging this gives
	\begin{align*}
		\Bigabs{ \prod_{\{i,j\} \in p} \Sigma_{\ell_i \ell_j} - \prod_{\{i,j\} \in p} \overline{\Sigma}_{\ell_i \ell_j} }  \leq \delta \| \Sigma\|_{\infty}^{t/2-1} 2^{t/2} \;,
	\end{align*}
	where we used that $\prod_{\{i,j\} \in p}({\Sigma}_{\ell_i \ell_j} + \delta_{\ell_i \ell_j})$ produces one term which cancels out with $\prod_{\{i,j\} \in p} \Sigma_{\ell_i \ell_j}$ and every one of the $(2^{t/2}-1)$ remaining ones, is at most $\delta \| \Sigma\|_{\infty}^{t/2-1}$, because $\delta < 1$, and $\norm{\Sigma}_{\infty} \geq 1$. 
	
	Combining the above with \Cref{eq:tensorformula}, we have that 
	\begin{align*}
	\left\|  \E_{Y \sim \cN(0,\Sigma)}[Y^{\otimes t}] - \E_{Y \sim \cN(0,\overline{\Sigma})}[Y^{\otimes t}]  \right\|_\infty \lesssim  t^t \delta \| \Sigma\|_{\infty}^{t/2} 2^{t/2} \;,
	\end{align*}		
since a crude upper bound on the number of matchings of $[t]$ is $t!$. Imposing that the right-hand side is at most $\tau$, we find that it is sufficient to have
	\begin{align*}
		\delta \leq \frac{\tau}{ t^t 2^{t/2} \|\Sigma\|_{\infty}^{t/2-1}},
	\end{align*}
	which is indeed less than $1$ since $\tau \leq 1$ and $\|\Sigma\|_{\infty} \geq 1$.
	
	Therefore, we use \Cref{claim:cov_est} with $\eps' \leq \delta/\|\Sigma\|_2$ and $\delta=\gamma/d^t$, in order to do a union bound over the $d^t$ elements of the tensor. Substituting these parameters yields the claimed sample complexity.
\end{proof}

\subsection{Proof of Estimation Lemmata} \label{sec:estimation_lemmata}

We recall the following general result from prior work (note that our theorem statement is slightly different from the one in \cite{kothari2021polynomial}, this is a minor correction and doesn't change the overall correctness of the proof). 
\begin{lemma}[Lemma 22 in \cite{kothari2021polynomial}] \label{lem:generic_estimation}
	Let $X_1, \dots, X_m \in \R^d$ and $\overline{\mu} = \E_{i \sim [m]}[ X_i]$. Let $V(\overline{\mu}, v),V'(\mu', v)$ be degree-2 polynomials in $v$ and $\overline{\mu}$ and $v$ and $\mu'$ respectively and let $S \subseteq \R^d$ be a set such that $V(\overline{\mu}, v) \geq 0$ for all $v \in S$ and $\overline{\mu} \in \R^d$. Let $T \subseteq [0, 1]^m$. Suppose that for every $v \in S$ and $a \in T$ such that $\sum_i a_i \geq (1-\eps)m$, we have the following two
	\begin{align} 
		&\left| \E_{i \sim [m]}[a_i \iprod{v,X_i-\overline{\mu}}]  \right| \leq \tilde{O}(\eps) \sqrt{V(\overline{\mu},v)} \;,\label{eq:assumption11} \\
		&\left| \E_{i \sim [m]}[a_i(\iprod{v,X_i-\overline{\mu}}^2 - V(\overline{\mu},v))] \right| \leq \tilde{O}(\eps) V(\overline{\mu},v)  \;.\label{eq:assumption12}
	\end{align} 
	Let $Y_1,\ldots, Y_m$ be any $\eps$-corruption of $X_1,\ldots,X_m$, let $\pE$ be a degree-6 pseudo-expectation in the variables $X_1',\ldots, X_m' \in \R^d$ and $w_1,\ldots, w_n \in \R$. Let $\mu' = \E_{i \sim [m]}[X_i']$. Suppose that
	\begin{itemize}
		\item $\pE$ satisfies $w_i^2=w_i$ for all $i \in [m]$.
		\item $\pE$ satisfies $\sum_{i \in [m]}w_i = (1-\eps)m$.
		\item $\pE$ satisfies $w_iX'_i=w_iY_i$ for all $i \in [m]$.
		\item $\pE[\E_{i \sim [m]}[\iprod{v,X_i'-\mu'}^2]]\leq (1+\tilde{O}(\eps))\pE[V'(\mu',v)]$ for every $v \in S$.
		\item $a \in T$, where $a$ is the vector with $a_i=\pE[w_i]\mathbf{1}(X_i=Y_i)$ for $i \in [m]$
	\end{itemize}
	Then, for every $v \in S$, the following hold:
	\begin{align*}
		&\pE[\iprod{v,\mu'-\overline{\mu}}^2] \leq O(\eps)\left( \pE[V'(\mu',v)] + V(\overline{\mu},v) \right) \;,\\
		&| \iprod{v,\hat{\mu}-\overline{\mu}}| \leq \tilde{O}(\eps) \sqrt{V(\overline{\mu},v)} + \sqrt{\pE\left[\E_{i \sim [m]}[(1-w_i')\iprod{v,X_i'-\overline{\mu}}]^2    \right]}\;,
	\end{align*}
	where $\hat{\mu}:=\pE[\mu']$ and
	\begin{align*}
		\pE\left[\E_{i \sim [m]}[(1-w_i')\iprod{v,X_i'-\overline{\mu}}]^2    \right]
		&\leq O(\eps) \left( \pE[V'(\mu',v)] - V(\overline{\mu},v) \right) \\
		&\qquad + \tilde{O}(\eps) \left( \pE[V'(\mu',v)] + V(\overline{\mu},v)  \right),
	\end{align*}
	where $w_i' = w_i \mathbf{1}(Y_i = X_i)$.
\end{lemma}

We now prove how \Cref{lem:specialized_mean} and \Cref{lem:specialized_cov} follow from \Cref{lem:generic_estimation}.

\SpecializedMean*
\begin{proof}
	This is a corollary of \Cref{lem:generic_estimation} with  $V(\overline{\mu},v):=v^T \E_{i \sim [m]}[(X_i-\overline{\mu})(X_i-\overline{\mu})^T]v$, $V'(\mu',v):=v^T \E_{i \sim [m]}[(X_i'-\mu')(X_i'-\mu')^T]v$, and the set $S$ chosen to be the set of all $k$-sparse unit vectors of $\R^d$. 
	
	We now check that the assumptions of \Cref{lem:generic_estimation} are true. The assumption of \Cref{eq:assumption11,eq:assumption12} holds because of \Cref{it:res2,it:res3} of \Cref{lem:resilience_integrated}. The first three conditions about the pseudo-expectation $\pE$ hold because $\pE$ satisfies the program of \Cref{def:gaxioms} and the last one holds trivially since $ \pE[\E_{i \sim [m]}[\iprod{v,X_i'-\mu'}^2]] = \pE[v^T \Sigma' v] = \pE[V'(\mu',v)]$. Finally, $a'_i \in [0,1]$ since $\pE$ satisfies $w_i^2=w_i$ and $\sum w_i \geq 1-\eps$. 
\end{proof}

\SpecializedCov*
\begin{proof}
	We use  \Cref{lem:generic_estimation} with the following substitutions: $S := \{ (vv^T)^{\flat} \; | \; v \in \cU_k  \}$, that is, let $S$ be the set of all $d^2$-dimensional vectors that are obtained by flattening matrices $vv^T$ for $v$ $k$-sparse unit vectors. We let the differences of pairs $(X_{ij})^\flat:= (\frac{1}{2}(X_i-X_j)(X_i-X_j)^T)^\flat$ for $i,j \in [m]$ play the role of the vectors $X_i,\ldots, X_j$ that appear in the statement of \Cref{lem:generic_estimation} and $\E_{ij}[(X_{ij})^\flat]$ play the role of $\overline{\mu}$. We choose $V(\E_{ij}[(X_{ij})^\flat],(vv^T)^\flat):=(v^T \overline{\Sigma} v)^2= (v^T\E_{ij \sim [m]}[ X_{ij} ]v)^2 = \iprod{(vv^T)^\flat,\E_{ij}[(X_{ij})^\flat]}^2$ (from this re-writing it is seen that this is a degree-2 polynomial in its arguments). Similarly, we let $V'$ be as $V$ but where $X_i$ are replaced with $X_i'$, i.e., the program variables.   Thus, the assumption of \Cref{eq:assumption11,eq:assumption12} now corresponds to \Cref{it:res6,it:res7} of \Cref{lem:resilience_integrated}. For example, to see the correspondence for the case of \Cref{it:res6}, we note that for $u = (vv^T)^\flat \in S$, the LHS in  \Cref{eq:assumption11} becomes
	\begin{align*}
		\E_{ij}[a_{ij}\iprod{u,(X_{ij})^{\flat}-\E_{ij}[(X_{ij})^\flat]}] &= 
		\E_{ij}[a_{ij}\iprod{(vv^T)^{\flat},((1/2)(X_i-X_j)(X_i - X_j)^T)^{\flat} - \E_{ij}[(X_{ij})^\flat]}] \\
		&=\E_{ij}[a_{ij}( v^TX_{ij} v - v^T \overline{\Sigma} v)] \;,
	\end{align*}
	which is equal to the LHS in \Cref{it:res6} of \Cref{lem:resilience_integrated}. 
	
	It remains to show that the rest of the assumptions used in \Cref{lem:generic_estimation} hold.  We use $w_{ij}:=w_i w_j$ in place of the $w_i$'s appearing in \Cref{lem:generic_estimation}. Note that by requiring $\tilde{\E}$ to be degree-12 pseudo-expectation in the variables $X_i',w_i$, it follows that $\tilde{\E}$ is degree-6 in the new variables $X_{ij}',w_{ij}$. We want to check that
	\begin{enumerate}
		\item $\pE$ satisfies $w_{ij}^2 = w_{ij}$ for all $i,j \in [m]$.
		\item $\pE$ satisfies $\sum_{i,j \in [m]} w_{ij} = (1-\eps)^2m^2$.
		\item $\pE$ satisfies $w_{ij} X_{ij}' = w_{ij} Y_{ij}$ for every $i,j \in [m]$.
		\item $\pE[\E_{i,j \sim [m]}[(v^T(X_{ij}'-\Sigma')v)^2]] \leq (2+ \tilde{O}(\eps))\pE[(v^T\Sigma' v)^2]$. %
	\end{enumerate}
	The first three follow immediately from the constraints of the program of \Cref{def:gaxioms}. The last one is verified below.
	\begin{claim}
		Let $\tilde{\E}$ be a degree-$4$
		pseudo-expectation in $X_{ij}',\Sigma'$ (as defined in \Cref{lem:specialized_cov}) satisfying \Cref{def:gaxioms}. Then
		\begin{align*}
			\pE\left[ \E_{i,j \sim [m]}[(v^T(X_{ij}' - \Sigma')v)^2] \right] \leq (2 + \tilde{O}(\eps)) \pE\left[ (v^T \Sigma' v)^2 \right] \;.
		\end{align*}
	\end{claim}
	\begin{proof}
		Since $\pE$ satisfies the system of \Cref{def:gaxioms}, by taking pseudoexpectations on the second to last constraint of the program, we see,
		\[ \pE\Brac{  \left(  \E_{i \sim [m]} \Brac{\iprod{v, X_i'-\mu'}^4}- 3(v^T \Sigma' v)^2   \right)^2 } \leq \tilde{O}(\eps^2)  \pE \Brac{ (v^T \Sigma' v)^4 }. \]
		 Applying Cauchy-Schwarz for pseudoexpectations (\Cref{fact:CS_pseudo_exp}), we get,
		 \[ \left(   \pE\Brac{ \E_{i \sim [m]} \Brac{\iprod{v, X_i'-\mu'}^4}- 3(v^T \Sigma' v)^2 }   \right)^2  \leq \tilde{O}(\eps^2)  \pE \Brac{ (v^T \Sigma' v)^4 } \;. \]
		 We know from \Cref{fact:univariate-interval} that $\{ 0 \leq x \leq 9\} \sststile{}{} \{ 81x^2 - x \geq 0\}$. Letting $x = (v^T \Sigma' v)^2$ we see that $v^T \Sigma' v = \sum_{ij} \iprod{X'_i - X'_j, v}^2 > 0$ and the last constraint of the program implies $x \leq 9$. Taking pseudoexpectations, we see that 
		$\pE[(v^T \Sigma' v)^4] \leq O(1) \cdot \pE[(v^T \Sigma' v)^2] \leq O(1)  \pE[(v^T \Sigma' v)^2]^2$, where the final inequality follows from the fact that $\pE[(v^T \Sigma' v)^2]$ is bounded between constants. Taking square roots of the previous inequality, since all terms involved are powers of two, and hence positive, this implies 
		  \[  \pE\Brac{ \E_{i \sim [m]} \Brac{\iprod{v, X_i'-\mu'}^4}} \leq (3+ \tilde{O}(\eps))  \pE \Brac{ (v^T \Sigma' v)^2}. \]
		Hence, we have that 
		\begin{align} \label{eq:consequence}
			\pE\left[ \E_{i \sim [m]}[\iprod{v,X_i'-\mu'}^4] \right] \leq (3+\tilde{O}(\eps))\pE[(v^T\Sigma'v)^2] \;.
		\end{align}
		We also have the following polynomial equality
		\begin{align}  \label{eq:simple_ineq}
			\gax \sststile{4}{ } \E_{i,j \sim [m]}[(v^T(X'_{ij}-\Sigma')v)^2] = \frac{1}{2} \left( \E_{i \sim [m]}[\iprod{v,X_i'-\mu'}^4] + (v^T\Sigma' v)^2  \right) \;.
		\end{align}
		Taking pseudo-expectations in \Cref{eq:simple_ineq} and combining with \Cref{eq:consequence}, yields that
		\begin{align*}
			\pE\left[ \E_{i,j \sim [m]}[(v^T(X_{ij}' - \Sigma')v)^2] \right]
			&= \frac{1}{2}\left( \pE\left[ \E_{i \sim [m]}[\iprod{v,X_i'-\mu'}^4] \right] + \pE[(v^T\Sigma' v)^2]   \right)\\
			&\leq \frac{1}{2}\left( 4+\tilde{O}(\eps) \right)\pE[(v^T\Sigma' v)^2] \;,
		\end{align*}
		which is the claimed bound.
	\end{proof}
	This completes the proof of \Cref{lem:specialized_cov}.
\end{proof}

\subsection{Omitted Details from \Cref{sec:proof_of_main_thm_right_error}} \label{sec:last_part_of_proof}

We complete the proof of \Cref{thm:gaussian_well_cond} as in \cite{kothari2021polynomial} by using the estimation lemmata proved above.
By \Cref{lem:specialized_cov} we have that 
\[ \abs{\iprod{\widehat \Sigma - \Sigma^*, vv^T} } \leq \tilde O(\eps) v^T \Sigma_0 v + \sqrt{R}, \]
and additionally  
\begin{align*} 
R &:= \pE\left[ \E_{ij \sim [m]} \left[(1-w_{ij}') \cdot v^T (X_{ij}-\Sigma^*) v \right]^2 \right] \\
&\leq O(\eps) \left( \pE \brac{(v^T \Sigma' v)^2} - (v^T \Sigma^* v) \right) +   \tilde O(\eps^2) \left( \pE\paren{(v^T \Sigma' v)^2 } + (v^T \Sigma^* v)^2 \right). 
\end{align*}
We can write $\Sigma' = A + B$ with $B = \E_{ij} \brac{(1-w_{ij}') X_{ij}'}$ and $A = \E_{ij} \brac{w'_{ij} X_{ij}} = \E_{ij} \brac{w'_{ij} X'_{ij} }$.
The latter equality follows by the definitions of the quantities. We will use the notation $M_v := v^T M v$ for any matrix $M$ (in particular, we will use this for $M \in \{A, B, \Sigma^* \}$). We have that 
\begin{align*}
    \pE[A_v^2] &= \pE[(\E_{ij \sim [m]}[w_{ij}' v^T X_{ij}v])^2] = \E_{i_1, j_1} \E_{i_2, j_2} \pE[w'_{i_1 j_1} w'_{i_2 j_2}] \cdot v^T X_{i_1 j_1} v \cdot v^T X_{i_2 j_2}v \\
    &\leq\E_{i_1, j_1} \E_{i_2, j_2} \sqrt{\pE[w'_{i_1 j_1}] \pE[w'_{i_2 j_2}]} \cdot v^T X_{i_1 j_1} v \cdot v^T X_{i_2 j_2}v\\
    &=  \left(\E_{i, j \sim [m]}\left[\sqrt{\pE[w'_{ij}]} v^T X_{ij} v \right]\right)^2 \leq (1 + \tilde O(\eps)) \Sigma_v^2 \;,
\end{align*}
where the final inequality follows from the deterministic condition (\Cref{lem:resilience_integrated}) with $a_{ij} = \sqrt{\pE[w'_{ij}]}$ (note that $(a_{ij})_{i, j}$ satisfy the required properties of that lemma with $a'_i = \pE[w_i] \mathbf{1}(X_i=Y_i)$ for all $i$). We can now rewrite upper bound $R$ in terms of $A, B, \Sigma$. 
\begin{align*} 
R &= \pE[(B_v - \E_{ij}[1-w'_{ij}] \cdot \Sigma_v)^2] \\
&\leq O(\eps) \paren{\pE[(A_v + B_v)^2] - \Sigma_v^2 } + \tilde O(\eps^2) \paren{\pE[(A_v + B_v)^2] + \Sigma_v^2}.   
\end{align*} 
By expanding the square, we can also lower bound $R$ as follows:
\begin{align*} 
\pE[(B_v - \E_{ij}[1-w'_{ij}] \cdot \Sigma_v)^2] \geq \pE[B_v^2] - 4\eps \Sigma_v \pE[B_v],
\end{align*} 
as $\Sigma_v \geq 0$ and $\pE$ satisfies $B_v \geq 0$. As $\pE[A_v^2] \leq (1+\tilde{O}(\eps)) \Sigma_v^2$ and $\pE[A_v B_v] \leq \sqrt{\pE[A_v^2] \pE[B_v^2]}$ by pseudoexpectation Cauchy-Schwartz (\Cref{fact:CS_pseudo_exp}), 
\begin{align*}
    \pE[(A_v + B_v)^2] &\leq \pE[B_v^2] + \sqrt{\pE[A_v^2] \pE[B_v^2]} + (1+\tilde{O}(\eps))\Sigma_v^2 \\
    &\leq \pE[B_v^2] + 2 \Sigma_v \sqrt{\pE[B_v^2]} + (1+\tilde{O}(\eps))\Sigma_v^2.
\end{align*}
Thus we can combine upper and lower bounds on $R$ to get, 
\[ \pE[B_v^2] - 4\eps \Sigma_v \pE[B_v] \leq O(\eps) \left(\pE[B_v^2] + 2 \Sigma_v \sqrt{\pE[B_v^2]} + \tilde{O}(\eps)\Sigma_v^2\right) + \tilde{O}(\eps^2) \Sigma_v^2. \]
Rearranging, applying $\pE[B_v] \leq \sqrt{ \pE[B_v^2]}$ and solving for $\pE[B_v^2]$ yields $\pE[B_v^2] \leq \tilde{O}(\eps^2) \Sigma_v^2$. This implies an upper bound on $R$ of $\tilde O(\eps^2) \Sigma_v^2$. This, in turn implies 
\begin{align*}
    \abs{v^T(\widehat \Sigma - \Sigma^*) v} \leq \tilde O(\eps) v^T \Sigma^* v + \sqrt{R} = \tilde O(\eps) v^T \Sigma^* v.
\end{align*}
This is the desired guarantee with $\Sigma^*$ instead of $\Sigma$. Using property 4  in \Cref{lem:resilience_integrated} and the triangle inequality, we get the desired spectral norm guarantee in terms of $\Sigma$. This finishes the proof (we have already shown in the main body that $\widehat \mu$ satisfies this property given that $\widehat \Sigma$ does).

%% file: sq-appendix.tex
\section{Omitted Details from \Cref{sec:sq-lowerbd}} \label{sec:SQappendix}

We start by providing additional background of \Cref{sec:sq-lowerbd}. In \Cref{sec:sos-cert-hard-instance}, we show that the hard instance in \Cref{thm:bounded_moment_sparse} has SoS-certifiable bounded moments. Finally, we present the lower bounds against low-degree polynomial tests in \Cref{sec:low_degree}.

\subsection{Omitted Background}

We provide the proof of \Cref{cor:genericSQbound} below.

\begin{proof}[Proof of \Cref{cor:genericSQbound}]
	\Cref{prob:generic_hypothesis_testing} is a decision problem in the sense of \Cref{def:decision}, where $D=\cN(0,I_d)$ and $\cD = \{P_{A,v} \}_{v \in \cU_k}$, where $\cU_k$ is the set of all $k$-sparse unit vectors. We calculate the SQ-dimension of $\cB(\D,D)$ as follows: Let $\cD_D$ be the subset of $\cD$ defined as $\cD_D=\{P_{A,v} \}_{v \in S}$, for $S$ being the set from the following fact:
	
	\begin{fact}[Lemma 6.7 in \cite{DKS17-sq}] \label{fact:vectors} 
		Fix a constant $0<c<1$ and let $k\leq \sqrt{d}$. There exists a set $S$ of $k$-sparse unit vectors on $\R^d$ of cardinality $|S|= \floor{d^{ck^{c}/8}}$ such that for each pair of distinct vectors $v,v' \in S$ we have that $|v^T v'| \leq 2 k^{c-1}$.
	\end{fact}
	
	Using Lemma 3.4 from \cite{DKS17-sq}, for $v,v' \in S$ we have that
	\begin{align*}
		\chi_{D}(P_{A,v},P_{A,v'}) &\leq |v^Tv'|^{m+1} \chi^2(A,\cN(0,1)) = (2 k^{c-1})^{m+1} \chi^2(A,\cN(0,1))\;,
	\end{align*}
	therefore the statistical dimension is $\mathrm{SD}(\cB,\gamma,\beta) = \Omega(d^{ck^c/8})$ for $\gamma=2^{m+1} k^{(c-1)(m+1)} \chi^2(A,\cN(0,1))$ and $\beta= \chi^2(A,\cN(0,1))$. An application of \Cref{lem:sq-from-pairwise} with $\gamma' = \gamma$ yields that any SQ algorithm, either makes at least one query of tolerance at most $2^{(m/2+1)}k^{-(m+1)(1/2-c/2)}\sqrt{\chi^2(A,\cN(0,1))}$ or makes at least the following number of queries:
	\begin{align*}
	    \Omega(d^{ck^c/8}) \frac{2^{m+1} k^{(c-1)(m+1)} \chi^2(A,\cN(0,1))}{\chi^2(A,\cN(0,1))-2^{m+1} k^{(c-1)(m+1)} \chi^2(A,\cN(0,1))} &\geq \Omega(d^{ck^c/8}) 2^{m+1} k^{(c-1)(m+1)} \\
	    &\geq \Omega\left(d^{ck^{c}/8} k^{-(m+1)(1-c)}\right) \;.
	\end{align*}
	This completes the proof.
\end{proof}

\subsection{SoS Certifiability of Hard Instances}\label{sec:sos-cert-hard-instance}
In this section, we will show that the hard instances in our proof have SoS certifiable bounded $t$-th moments. 

\begin{claim}\label{cl:SoS_certifiability}
Fix a $t\in \N$ such that $t$ is a power of $2$.
Denote by $G(x)$ the pdf of $\cN(0,1)$. Let  $Q_1(x), Q_2(x)$ be the distributions from the proof of \Cref{lem:ABddMoment}, and define $Q:=(1-\eps)Q_1+ \eps Q_2$ where $\delta = 1/(2000t) \eps^{1-1/t}  $, $\delta' = -(1 -\eps) \delta/\eps$, and $|\delta'| \geq 1$.
Recall that the first $t$ moments of $Q$ has moments match with $\cN(0, 1)$.
Let $P_1,P_2,P$ the distributions that have $Q_1$, $Q_2$, and $Q$, respectively, in the $u$ direction and are standard Gaussian in all perpendicular directions. 
Let $\cA := \{\sum_i v_i^2 = 1\}$, and define $\mu := \E_{X \sim P_1} [X]$.
Then $P_1$ has SoS certifiably bounded moments, i.e., there exists a constant $C>0$ such that 
\begin{align*}
    \cA \sststile{O(t)}{} \E_{X \sim P_1}[\langle v, X - \mu \rangle^t ]^2 \leq (C t)^t  \;.
\end{align*}

\end{claim}
\begin{proof}
We will use the following claim that depends on the SoS triangle inequality (\Cref{fact:sos-triangle}).
\begin{claim}
\label{cl:FromCenteredToRaw}
Let $P$ be a distribution over $\R^d$ with mean $\mu$ and define $p_i(v) = \E[\langle v, X \rangle^i]$ and $p'_i(v) = \E[\langle v, X - \mu \rangle^i]$. Let $\cA = \{\sum_i v_i^2 = 1\}$.
There exists a $C > 0$ such that the following holds: Let $t$ be a power of $2$.
If $\cA \sststile{O(t)}{}  p'_i(v)^2 \leq R$ for some $R \geq 0$ and all $i \in [t]$, then
 $\cA \sststile{O(t)}{}  p_t(v)^2 \leq C^t R \max(1, \|\mu\|_2^{2t})$.
 Similarly, if $\cA \sststile{O(t)}{}  p_i(v)^2 \leq R$ for some $R \geq 0$ and all $i \in [t]$, then
 $\cA \sststile{O(t)}{}  p'_t(v)^2 \leq C^t R \max(1, \|\mu\|_2^{2t})$.
\end{claim}
\begin{proof}
For $i \in \{0,1,\dots,t\}$,  let $C_i = \binom{t}{i}$.
We have the following polynomial equalities:
\begin{align*}
    p_t'(v) &= \sum_{i=0}^t (-1)^iC_i p_i(v) \langle \mu, v\rangle^{t-i}, \\ 
    p_t(v) &= \sum_{i=0}^t C_i p_i'(v) \langle \mu, v\rangle^{t-i}.
\end{align*}
Applying SoS triangle inequality (\Cref{fact:sos-triangle}), using the fact that $\cA \sststile{2i}{v} \langle v , \mu \rangle ^{2i} \leq \|\mu\|_2^{2i} $, and $C_i \leq 2^t$, we get the desired claim.
\end{proof}

Note that the mean of  $P_1$ is $u\delta$ and $\|u \delta\|_2 \leq 1$, and thus  \Cref{cl:FromCenteredToRaw} implies that it suffices to show that $\cA \sststile{O(t)}{}  \E_{X \sim P_1}[\langle v,X \rangle^i]^2 \leq (O(t))^t$ for all $i \in [t]$.

	 Note that $P_2$ is $\cN(u \delta' , I)$ and $P$ matches the moments of $\cN(0, I)$ up to degree $t$ in every direction.
	 Thus \Cref{fact:gaussian-moments} implies the following: 
	 for all $i \in [t]$
	\begin{align}
	    \cA &\sststile{O(t)}{}  (Ct)^{t} - ( \E_{X \sim P}[\iprod{v,X}^i] )^2 \geq 0, \label{eq:certbddmomentssq'} \\
\cA &\sststile{O(t)}{}  (C't)^{t} -  (\E_{X \sim P_2}[(\iprod{v,X} - \delta' \iprod{u, v})^i] )^2 \geq 0. \label{eq:certbddmomentssq}
\end{align} 		   
Suppose for now that $P_2$ satisfies the following, which we will establish shortly: there exists a constant $C''$ such that for all $i \in [t]$,
\begin{align}
\label{eq:certbddmomentssq2}
    \cA \sststile{O(t)}{}  \eps ^2 \E_{X \sim P_2}[\iprod{v,X}^i]^2 \leq (C''t)^t.
\end{align}

To show $\cA \sststile{O(t)}{}  \E_{X \sim P_1}[\langle v,X \rangle^i]^2 \leq (O(t))^t$, we proceed as follows: 
\begin{align*}
\cA \sststile{O(t)}{}  \E_{X \sim P_1}[\langle v,X \rangle^i]^2
		& = (1/(1-\eps))^2 ( \E_{X \sim P}[\iprod{v,X}^i]  - \eps  \E_{X \sim P_2}[\iprod{v,X}^i] )^2\\
		&\leq 2/(1-\eps)^2 \Paren{\E_{X \sim P}[\iprod{v,X}^i]^2 + \eps ^2 \E_{X \sim P_2}[\iprod{v,X}^i]^2}\\
		&\leq (O(t))^t,
\end{align*}
where the first inequality uses SoS triangle inequality (\Cref{fact:sos-triangle}) and the second inequality uses \Cref{eq:certbddmomentssq'} for the first term and \Cref{eq:certbddmomentssq2} for the second term.
Thus it only remains to show that \Cref{eq:certbddmomentssq2} holds to complete the proof.

Note that the mean of $P_2$ has norm $\delta'$ and $|\delta'| \geq 1$.
\Cref{cl:FromCenteredToRaw} and \Cref{eq:certbddmomentssq} imply the following: 
\begin{align*}
\cA \sststile{O(t)}{}   \eps^2    \E_{X \sim P_2}[\iprod{v,X}^i]^2 \leq  (CC')^t \eps^2|\delta' |^{2t}.
\end{align*}
By definition, $\eps^2 |\delta' |^{2t} \leq \eps^2 (\delta/\eps)^{2t} \leq \eps^2 (\eps^{-1/t})^{2t} = 1 $.  
This completes the proof.

\end{proof}

\subsection{Implications for Low-Degree Polynomial Algorithms}\label{sec:low_degree}
We can get quantitatively similar lower bounds in the low-degree model of computation using its connection with the SQ model \cite{BBHLS20}. The result of this section, roughly speaking, is that any polynomial algorithm for sparse non-Gaussian component analysis where the non-Gaussian component matches $m$ moments with $\cN(0,1)$, either uses more than $k^{m+1}$ samples or has degree more than $k^{\Omega(1)}$ (which in the worst case requires $d^{k^{\Omega(1)}}$ monomial terms that need to be computed). Plugging $m=3$ yields an analog of \Cref{thm:k4lower_bound} and letting $m$ be equal to the number of bounded moments of the inliers' distribution, i.e., $t$, gives an analog of \Cref{thm:bounded_moment_sparse}.

\cite{BBHLS20} uses a slightly different version of hypothesis testing problems, where in the alternative hypothesis, the ground truth is chosen according to a probability measure.

\begin{problem}[Non-Gaussian Component Hypothesis Testing with Uniform Prior] \label{def:nongaussian_new}
	Let a distribution $A$ on $\R$. For a unit vector $v$, we denote by $P_{A,v}$ the distribution with density $P_{A,v}(x) := A(v^Tx) \phi_{\perp v}(x)$, where $\phi_{\perp v}(x) = \exp\left(-\|x - (v^Tx)v\|_2^2/2\right)/(2\pi)^{(d-1)/2}$, 
i.e., the distribution that coincides with $A$ on the direction $v$ and is standard Gaussian in every orthogonal direction. Let $S$ be the set of nearly orthogonal vectors from \Cref{fact:vectors}. Let $\cS = \{ P_{A,v} \}_{u \in S}$. We define the simple hypothesis testing problem where the null hypothesis is $\mathcal{N}(0,I_d)$ and the alternative hypothesis is $P_{A,v}$ for some $v$ uniformly selected from $S$.
\end{problem}

 We now describe the model in more detail. We will consider tests that are thresholded polynomials of low-degree, i.e., output $H_1$ if the value of the polynomial exceeds a threshold and $H_0$ otherwise. We need the following notation and definitions.
For a distribution $D$ over $\cX$, we use $D^{\otimes n}$ to denote the joint distribution of $n$ i.i.d.\ samples from $D$.
For two functions $f:\cX \to \R$, $g: \cX \to R$ and a distribution $D$, we use $\langle f, g\rangle_{D}$ to denote the inner product $\E_{X \sim D}[f(X)g(X)]$.
We use $ \|f\|_{D}$ to denote $\sqrt{\langle f, f \rangle_{D} }$.
We say that a polynomial $f(x_1,\dots,x_n):\R^{n \times d} \to \R$ has sample-wise degree $(r,\ell )$ if each monomial uses at most $\ell$ different samples from $x_1,\dots,x_n$ and uses degree at most $r$ for each of them.
Let $\cC_{r,\ell}$ be linear space of all polynomials of sample-wise degree $(r,\ell)$ with respect to the inner product defined above.
For a function $f:\R^{n \times d} \to \R$, we use $f^{\leq r, \ell}$ to be the orthogonal projection onto $\cC_{r,\ell}$   with respect to the inner product $\langle \cdot , \cdot \rangle_{D_0^{\otimes n}}$.  Finally, for the null distribution $D_0$ and a distribution $P$, define the likelihood ratio $\overline{P}^{\otimes n}(x) := {P^{\otimes n}(x)}/{D_0^{\otimes n}(x)}$.

\begin{definition}[$n$-sample $\tau$-distinguisher]
For the hypothesis testing problem between two distributions $D_0$ (null distribution) and $D_1$ (alternate distribution) over $\cX$,
we say that a function $p : \cX^n \to \R$ is an $n$-sample $\tau$-distinguisher if $|\E_{X \sim  D_0^{\otimes n}}[p(X)] - \E_{X \sim D_1^{\otimes n}}[p(X)]| \geq \tau \sqrt{\Var_{X \sim D_0^{\otimes n}} [p(X)] }$. We call $\tau$ the \emph{advantage} of the polynomial $p$. 
\end{definition}
Note that if a function $p$ has advantage $\tau$, then the Chebyshev's inequality implies that one can furnish a test $p':\cX^n \to \{D_0,D_1\}$ by thresholding $p$ such that the probability of error under the null distribution is at most $O(1/\tau^2)$. 
We will think of the advantage $\tau$ as the proxy for the inverse of the probability of error (see  Theorem 4.3 in \cite{kunisky2019notes}  for a formalization of this intuition under certain assumptions) and we will show that the advantage of all polynomials up to a certain degree is $O(1)$. 
It can be shown that for hypothesis testing problems of the form of   \Cref{def:nongaussian_new}, 
 the best possible advantage among all polynomials in $\cC_{r,\ell}$ is captured by the low-degree likelihood ratio (see, e.g., ~\cite{BBHLS20,kunisky2019notes}):
\begin{align*}
    \left\| \E_{v \sim \cU(S)}\left[ \left( \overline{P}_{A,v}^{\otimes n}  \right)^{\leq r, \ell } \right]  - 1  \right\|_{D_0^{\otimes n}},
\end{align*}
where in our case $D_0 = \cN(0,I_d)$.

To show that the low-degree likelihood ratio is small, we use the result from \cite{BBHLS20} stating that a lower bound for the SQ dimension translates to an upper bound for the low-degree likelihood ratio. Therefore, given that we have already established in  previous section that $\mathrm{SD}(\cB(\{P_{A,v} \}_{v \in S},\cN(0,I_d)), \gamma,\beta)=\Omega(d^{ck^c/8})$ for $\gamma=2^{m+1} k^{(c-1)(m+1)} \chi^2(A,\cN(0,1))$ and $\beta= \chi^2(A,\cN(0,1))$, we obtain the corollary:
    \begin{theorem}\label{thm:hypothesis-testing-hardness}
        Let  $0<c<1$. Consider the hypothesis testing problem of Problem~\ref{def:nongaussian_new} where $A$ matches $m$ moments with $\cN(0,1)$. For any $d,k,m \in \Z_+$ such that $k \leq \sqrt{d}$ and $ck^c \geq \Omega(m \log k)$,  any $n\leq k^{(1-c)(m+1)}/(2^{m+1}\chi^2(A,\cN(0,I_d))$  and any even integer $\ell \leq (c k^c \log d)/(32 m \log k)$, we have that 
        \begin{align*}
            \left\| \E_{u \sim S} \left[(\bar{P}_{A,u}^{\otimes n})^{\leq \infty,\Omega(\ell)}\right] - 1 \right\|_{\cN(0,I_d)^{ \otimes n }}^2 \leq 1\;.
        \end{align*}
    \end{theorem}

The interpretation of this result is that unless the number of samples used $n$ is greater than $k^{(1-c)(m+1)}/(2^{m+1}\chi^2(A,\cN(0,I_d))$, any polynomial of degree roughly up to $k^c\log d$  fails to be a good test (note that any polynomial of degree $\ell$ has sample-wise degree at most $(\ell,\ell)$).
Using the lower bounds for the SQ dimension, we also obtain lower bounds for the low-degree polynomial tests for problems in \Cref{thm:bounded_moment_sparse,thm:k4lower_bound}  with qualitatively similar guarantees.

Finally, the connection to the estimation problem is again done via the reduction of \Cref{lem:reduction}, which also works in the low-degree model family of algorithms. 

\begin{remark}
[Reduction within low-degree polynomial class] \label{lem:reductionn}
Let $\cA$ be a low-degree polynomial algorithm for \Cref{prob:search_problem} with degree $\ell$. Then the  reduction in \Cref{lem:reduction} gives us an algorithm $\cA'$ for \Cref{prob:hypothesis_testing}  which can be implemented as a polynomial test of degree $2\ell$.
\end{remark}

%% file: app-sample-complexity.tex
 \section{Information-Theoretic Error and Sample Complexity}
\label{app:info-th-sample-comp}

\begin{theorem}[Sample Complexity of Robust Sparse Mean Estimation with Bounded Moments] \label{thm:sample-comp-bdd-mmnts}
Let $C$ be a sufficiently large constant and $c$ a sufficiently small positive constant. There is a (computationally inefficient) algorithm that, given any $\eps<c$ and an $\eps$-corrupted set of size $n > C k \log (d/k)/\eps^{2 - 2/t}$ from any distribution with $k$-sparse mean and $t$-th moments bounded by $M$, finds a $\widehat{\mu}$, such that $\|\widehat{\mu} - \E_{X \sim D}[X] \|_2 = O(M^{1/t} \eps^{1-1/t})$, with probability at least $0.9$.
\end{theorem}
\begin{proof}
Let $S$ be the given (corrupted) data set of cardinality $n$ and $\mu = \E_{X \sim D}[X]$.
For a unit vector $v$, let $S_v$ be the projection of the points along $v$, that is $S_v = \{v^Tx:x \in S\}$.
Note that we have assumed that in any $k$-sparse direction $v$, the $t$-th moment of inliers is at most $M$.

Let $\cC$ be a $1/2$-net of the unit-norm $k$-sparse vectors (which we denote by $\cU_k$). The cardinality of $\cC$  is bounded by ${d \choose k}5^k$ since there are at most ${d \choose k}$ ways to select the non-zero coordinates and a $(1/2)$-net of $\R^k$ has size at most $5^k$.

For $\tau<1$, let $f_{\tau}$ be the real-valued function on univariate sets that computes the $\tau$-trimmed mean of the given data set as in \cite{LugosiM19robust}.
From that, it is implied that for any unit vector $v$ and $\tau = \Theta(\eps + \log(1/\gamma')/n)$ (where the parameters are such that $\tau<1$), with probability $1 - \gamma'$ we have that 
\begin{align*}
   |f_\tau(S_v) - \mu^Tv| =  O(M^{1/t}(\eps^{1-1/t} + \sqrt{\log(1/\gamma)/n} ) ) \;.
\end{align*}

Setting $\gamma' = \gamma/|\cC|$ and  $n \geq C \log(1/\gamma')/\eps^{2 - 2/t}$ and using a union bound, we have that with probability $1- \gamma$,
for each $v \in \cC$,
$|f_{\tau}(S_v) - \mu^Tv| \leq \delta $, where $\delta =O(M^{1/t}\eps^{1-1/t})$. By noting that $\log(1/\gamma') = O(k \log(d/k) + \log(1/\gamma))$ it is sufficient to have $n > C (k \log (d/k))/\eps^{2 - 2/t}$ for the above to hold for a constant failure probability $\gamma$.
We denote this event by $\cE$.
We will assume that $\cE$ holds for the remainder of the proof. 
For each $v \in \cC$, define $\widehat{\mu}_v := f_{\tau}(S_v)$ and let  the estimate $\widehat{\mu}'$ to be any point with the property $|v^T \widehat{\mu}' - \widehat{\mu}_v| \leq \delta $ for all $v \in \cC$ (such a point always exists under $\cE$, since the true mean satisfies that property). For that $\widehat{\mu}'$, we have that
\begin{align*}
    |v^T(\widehat{\mu}' - \mu)| \leq |v^T \widehat{\mu}' - \widehat{\mu}_v | + |\widehat{\mu}_v - v^T \mu | 
    \leq 2\delta \;,
\end{align*}
for every $v \in \cC$. We claim that $|v^T(\widehat{\mu}' - \mu)| \leq 4\delta$ for all $v \in \cU_k$. To see this, let $v_0:= \arg\max_{v \in \cU_k} |v^T(\widehat{\mu}' - \mu)|$  and $w := \arg\min_{x \in \cC} \|v_0-x\|_2$. We have that
\begin{align*}
|v_0^T(\widehat{\mu}' - \mu)|  \leq  |w^T(\widehat{\mu}' - \mu)| + |(w-v)^T(\widehat{\mu}' - \mu)|
\leq |w^T(\widehat{\mu}' - \mu)| + \frac{1}{2}|v_0^T(\widehat{\mu}' - \mu)| \;.
\end{align*}
Solving for $|v_0^T(\widehat{\mu}' - \mu)|$ shows the claim. Let the final estimate be $\widehat{\mu} = h_k(\widehat{\mu}')$, where $h_k$ is the operator that truncates a vector to its largest $k$ coordinates. Applying \Cref{fact:sparseTruncation}, we get that $\|\widehat{\mu} - \mu\|_2 = O(\delta)$.
\end{proof}

\begin{theorem}[Sample Complexity of Robust Sparse Mean Estimation of Gaussian]
Let $C$ be a sufficiently large constant and $c$ a sufficiently small positive constant. There is a (computationally inefficient) algorithm that, given any $\eps<c$ and any $\eps$-corrupted set of size $n > C k \log (d/k)/\eps^{2}$ from $\cN(\mu,\Sigma)$ with $k$-sparse $\mu$, finds a $\widehat{\mu}$, such that $\|\widehat{\mu} - \mu \|_2 = O(\eps \sqrt{\|\Sigma\|_2} )$, with probability at least $0.9$.
\end{theorem}
\begin{proof}
We will use the same notation in as the proof of \Cref{thm:sample-comp-bdd-mmnts}.
For each $v \in \cC$, define $\widehat{\mu}_v := \text{Median} (S_v)$.
Standard results (see, for example, \cite[Lemma 3.3]{LaiRV16}) imply that with probability $1 - \exp(-  n\eps^2)$,
$|v^T \mu - \widehat{\mu}_v| \leq \delta$, where  $\delta = O(\eps \sqrt{\|\Sigma\|_2})$.
Let $\cE$ be the event where for each $v \in \cC$, $|\widehat{\mu}_v - \mu^Tv| \leq \delta$.
By a union bound, $\cE$ happens with probability at least $ 1 - \gamma $, if $n \geq C ( k \log (d/k) + \log(1/\gamma))/\eps^{2})$ for a large enough constant $C$.
Following the same argument as the proof of \Cref{thm:sample-comp-bdd-mmnts}, we get the desired result.
\end{proof}

We now state the following folklore results for the information-theoretic lower bound. Although we present the results for univariate distributions, it is easy to see that the same lower bounds also hold for $k$-sparse distributions for any $k\geq 1$.

\begin{fact}[Information-theoretic Lower Bounds]
The following hold:
\begin{itemize}
    \item There exist univariate distributions $D_1,D_2$ such $D_2 = (1-\eps)D_1 + \eps N$ for some $N$, the $t$-th moments of both $D_1,D_2$ are at most $M$, and $|\E_{X \sim D_1}[X] - \E_{X \sim D_2}[X]| = \Omega(M^{1/t}\eps^{1-1/t})$.
    \item There exist Gaussian distributions $D_1=\cN(\mu_1,\sigma^2),D_2 = \cN(\mu_2,\sigma^2)$ such that  $|\mu_1-\mu_2| = \Omega(\eps \sigma)$ and $(1-\eps)D_1+\eps N_1 = (1-\eps)D_2+\eps N_2$.
\end{itemize}

\end{fact}
\begin{proof}
By scaling, we focus on the $M=1$ case. Let $D_1$ be the Dirac delta at zero and let $D_2 = (1-\eps)D_1 + \eps N$, where $N$ has all its mass at $\eps^{-1/t}$. Then, the means are indeed separated by $\eps^{1-1/t}$. For the $t$-th moment of $D_2$ we have that $\E_{X \sim D_2}[(X-\mu_2)^t] = (1-\eps)\eps^{t-1} + \eps(\eps^{-1/t} - \eps^{1-1/t})^t\leq  \eps + \eps (\eps^{-1/t}(1-\eps))^t\leq  \eps + (1-\eps) \leq 1$.

For the Gaussian distributions, the claim is based on the fact that $\mathrm{dtv}(\cN(\mu_1,\sigma^2),\cN(\mu_2,\sigma^2))=\Theta(\eps)$ whenever $|\mu_1-\mu_2| = \sigma \eps$, thus an additive adversary can make the two distributions look the same. A version of the lower bound can also be found in 
\cite[Observation 1.4]{LaiRV16} and 
\cite[Theorem 2.2]{chen2018}.

\end{proof}

\section{Bit Complexity Analysis}
We briefly describe how our algorithm can be implemented in the standard RAM model of  computation. We will focus on \Cref{thm:main-poincare-informal} for this note and a similar remark applies to \Cref{thm:main-gaussian-informal}.
Let $R = C M^{1/t}\sqrt{d/\eps}$ for some sufficiently large constant $C$. 
We assume that the bound $M$ on the $t$-th moments of the distribution of inliers has bit complexity bounded  by $\poly(d)$. 
It is a standard fact that there exists a na\"ive estimation step that uses $O(\log(1/\tau))$ samples and finds a $\hat{\mu}$ such that with probability $1-\tau$, $\|\hat{\mu} - \mu\|_{2,k} = O(R)$. Because of the moment bounds on the distribution, we can ignore all samples for which $\| X - \hat{\mu} \|_2> R$, and this will translate to having and additional $\eps$ fraction of corruptions. 

Let $D$ be the original distribution. To implement the algorithm in the RAM model, every coordinate of all samples gets rounded to precision  $\eta \leq \eps M^{1/t}/\sqrt{d}$. 
Let $D'$ be the distribution of these rounded samples.
Since the coordinates of these samples lie in a range of $O(R)$, by the naive-filtering step of the previous paragraph, a word of $O(\log(R/\eta))$ bits is enough for storing each coordinate. The rounding procedure can be treated as an additive noise of magnitude $\eta$ to every coordinate of the inlier distribution. 
This can change its mean in $(2,k)$-norm by at most $\eta \sqrt{k}$. Thus, to estimate the mean of $D$, it suffices to estimate the mean of the rounded distribution $D'$. 
Moreover, $D'$  continues to satisfy appropriate tail bounds, worsened by at most constant factors. 
Additionally, it is not too hard to show that if the bit complexity of the SoS proof of the moment upper bound of $D$ is polynomial, there is an SoS proof of a moment upper bound of $D'$ that is worse by a multiplicative constant factor and also has polynomial bit complexity. 
Since these are the only requirements that the algorithm needs for guaranteed correctness, the algorithm on the rounded data set would yield an estimate that is within $O( \eps M^{1/t})$ of the true mean in $(2,k)$-norm, as desired.

Having implemented the rounding step for the samples, it is guaranteed that the bit complexity of the system of \Cref{def:axioms} is at most some $\poly((md)^{t^2})$. This enables the SoS algorithm to terminate after $\poly((md)^{t^2})$ time with an approximate pseudodistribution satisfying the constraints with slack $\tau = 2^{-\poly(tB)}$, where $B$ is the bit complexity of the SoS proofs. Finally, because of  that slack, there may be some error accumulated in our analysis. The error incurred will be $\tau 2^{O(B)}$ which is negligible because of our choice of $\tau$.